\documentclass[fleqn,11pt,oneside]{article}

\usepackage{amsmath}
\usepackage{amsthm}
\usepackage{amsfonts}
\usepackage{amssymb}
\usepackage{physics}
\usepackage{xcolor}
\usepackage{graphicx}
\usepackage{authblk}
\usepackage{accents}
\usepackage{mathtools}
\usepackage{mathrsfs}
\usepackage[hidelinks]{hyperref}
\usepackage{subfigure}
\usepackage{multirow}

\newcommand{\dbtilde}[1]{\widetilde{\raisebox{0pt}[0.85\height]{$\widetilde{#1}$}}}
\renewcommand{\tilde}{\widetilde}

\newcommand{\R}{\mathbb{R}}
\newcommand{\C}{\mathbb{C}}
\newcommand{\Z}{\mathbb{Z}}

\newcommand{\vct}[1]{\boldsymbol{#1}}

\newtheorem{theorem}{Theorem}[section]
\newtheorem{lemma}[theorem]{Lemma}
\newtheorem{corollary}[theorem]{Corollary}
\newtheorem{remark}[theorem]{Remark}

\newcommand{\mat}[1]{\begin{matrix}#1\end{matrix}} 
\newcommand{\bb}[1]{\left(#1\right)} 
 
\renewcommand{\sb}[1]{\left[#1\right]} 

\newcommand{\cb}[1]{\left\{#1\right\}}
\newcommand{\D}{\mathrm{d}}
\newcommand{\I}{\mathrm{i}}

\newcommand{\half}{\frac{1}{2}}

\renewcommand{\bar}{\overline}

\begin{document}

%
%
%
%

\begin{center}
	\begin{minipage}[t]{6.0in}
We present a rapidly convergent scheme for computing globally optimal Wannier functions of isolated single bands for matrix models in two dimensions. The scheme proceeds first by constructing provably exponentially localized Wannier functions directly from parallel transport (with simple analytically computable corrections) when topological obstructions are absent. We prove that the corresponding Wannier functions are real when the matrix model possesses time-reversal symmetry. When a band has a nonzero Berry curvature, the resulting Wannier function is not optimal, but it is transformed into the global optimum by a single gauge transformation that eliminates the divergence of the Berry connection. Complete analysis of the construction is presented, paving the way for further improvements and generalizations. The performance of the scheme is illustrated with several numerical examples.

		\thispagestyle{empty}
		
		\vspace{ -100.0in}
		
	\end{minipage}
\end{center}

\vspace{ 3.60in}

\begin{center}
	\begin{minipage}[t]{4.4in}
		\begin{center}
			
			\textbf{Constructing optimal Wannier functions via potential theory: isolated single band for matrix models  } \\
			
			\vspace{ 0.50in}
			
			Hanwen Zhang$\mbox{}^{\dagger }$  \\
			\today
			
		\end{center}
		\vspace{ -100.0in}
	\end{minipage}
\end{center}

\vspace{ 2.00in}

\vfill
\noindent
$\mbox{}^{\dagger}$ Dept. of Applied and Computational Mathematics, Yale University, New Haven, CT 06511

\vspace{2mm}

\clearpage

\tableofcontents
\vspace{0.2in}

\section{Introduction}
A standard task in solid state physics and quantum chemistry is the computation of localized molecular
orbitals known as Wannier functions\,\cite{marzari2012maximally}. 
This paper is a follow-up to \cite{gopal2024high}, extending its approach from isolated single bands of Schr\"odinger operators in one dimension to matrix models in two dimensions. We present a rapidly convergent scheme for computing globally optimal Wannier functions based on solving the parallel transport equation and the Helmholtz-Hodge decomposition of Berry connections (gauge fields). First, we construct provably exponentially localized Wannier functions by \emph{solely} solving the parallel transport equation (with simple corrections) to assign eigenvectors smoothly when no topological obstruction is present. Although the resulting Wannier functions are not optimally localized when the Berry curvature is nonzero, they are \emph{automatically} real (given a simple choice of the initial condition) when the model has time-reversal symmetry. Next, globally optimal Wannier functions are obtained via a single gauge transformation that eliminates the divergence of the Berry connections. Such gauge transformations are obtained by solving Poisson's equation on tori, which can be done in an efficient and accurate manner. (An alternative approach that utilizes the gauge invariance of the Berry curvature can be found in Remark\,\ref{rmk:alt} and Appendix\,\ref{sec:method2}.)
In addition, when topological obstructions are encountered, the resulting assignments are still analytic (but they are discontinuous when viewed as periodic functions); they are amenable to efficient interpolation schemes based on, for example, Chebyshev nodes, instead of equispaced ones in the Fourier basis in this paper.

The construction of Wannier functions can be viewed as a problem of assigning eigenvectors of some analytic family of matrices/operators to be as smooth as possible so that coefficients of the Fourier expansions of the eigenvectors decay exponentially.
For computational purposes, it is a common practice to view this problem as a nonlinear and nonconvex optimization problem. As a result, the focus of most work has been on obtaining robust initial guesses. For example, high-quality initial guesses are obtained by applying the interpolative decomposition to the so-called density matrices\,\cite{damle2017scdm, damle2019variational}. Parallel-transport-based approaches have been applied for assigning eigenvectors\,\cite{cances2017robust}, but it results in continuous assignments, corresponding to slowly decaying Wannier functions. Then, the initial guesses are optimized by minimizing the spread of the Fourier coefficients (Marzari–Vanderbilt functional) via gradient descent\,\cite{marzari1997maximally, marzari2012maximally}, where the gradients are often computed by finite differences. 
From this optimization point of view, the first step of the scheme in this paper can be viewed as obtaining nearly optimal initial assignments of eigenvectors (corresponding to exponentially localized Wannier functions) by only carrying out parallel transport.
Moreover, the minimization procedures can be drastically accelerated by making use of the potential theory of the physical quantities involved; for the single band case in this paper, it is an unconstrained quadratic problem whose global optimum is achieved in a single Newton step, where the Hessian inversion is essentially free.

It should be observed that, although we solve the parallel transport equation explicitly in this paper to obtain the band structure and Wannier functions simultaneously, the computation of the two can be easily decoupled by using the scheme for parallel transport in Remark\,\ref{rmk:twist} at the cost of lower accuracy. Such an approach is only a second-order scheme; it produces results of reasonable accuracy with almost no additional computational cost once the band structure is obtained. We refer the reader to Remark\,\ref{rmk:twist} and Example 1 in Section\,\ref{sec:eg1} for details.

The extension of the schemes in \cite{gopal2024high} and this paper to isolated multibands has been worked out and is in preparation for publication. We refer the readers to Section\,\ref{sec:final} for the generalizations to Schr\"odinger operators and higher dimensions.

The structure of this paper is as follows. Section\,\ref{sec:ps} contains the description of the problem to be solved. Section\,\ref{sec:ps} contains the physical and mathematical preliminaries, followed by Section\,\ref{sec:npre} consisting of the numerical tools used in this paper. In Section\,\ref{sec:ap}, we introduce the analytic apparatus for constructing Wannier functions. Section\,\ref{sec:ap} contains the procedures for the construction of optimal Wannier functions, followed by the corresponding numerical procedures in Section\,\ref{sec:numerics}.  Section\,\ref{sec:results} contains several examples that illustrate the performance of the scheme in this paper.

\section{Problem statement\label{sec:ps}}
Let $D^*$ be a two-dimensional (deformed) torus 
\begin{equation*}
	D^* = \cb{\kappa_1\vct{b}_1 + \kappa_1 \vct{b}_2: \kappa_1,\kappa_2\in [-1/2,1/2]}
\end{equation*}
parameterized by $(\kappa_1, \kappa_2)\in \sb{-\half, \half}\times \sb{-\half,\half}$\,, where $\vct{b}_1$ and $\vct{b}_2$ are basis vectors in $\R^2$ that span a Bravais lattice.
Suppose that $H:D^* \rightarrow \C^{n\times n}$ is a family of Hermitian matrices that satisfies the following periodicity condition
\begin{equation}
	H(\vct{k} + m_1\vct{b}_1 + m_2\vct{b}_2) = H(\vct{k}), \quad  \vct{k}\in D^*\,, m_1,m_2\in\Z\,.
	\label{eq:hper}
\end{equation}
Suppose further that the elements in $H$ are analytic functions on $D^*$.  
We consider an eigenvalue $E(\vct{k})$ and its corresponding normalized eigenvector $\vct{u}(\vct{k})$ for some $\vct{k}\in D^*$ defined by the formula
\begin{equation}
	H(\vct{k})\vct{u}(\vct{k}) = E(\vct{k})\vct{u}(\vct{k})\,,
	\label{eq:heig}
\end{equation}
and we assume that the eigenvalue $E(\vct{k})$ never becomes degenerate for any $\vct{k}\in D^*$. 
Since we may multiply the vector $\vct{u}(\vct{k})$ by a $\vct{k}$-dependent phase factor $e^{-\I \varphi(\vct{k})}$ (with a real $\varphi(\vct{k})$) without affecting the definition in (\ref{eq:heig}), the eigenvector  $\vct{u}(\vct{k})$ is only unique up to a phase factor. The phase factors are often referred to as the ``choice of gauge'' in physics literature. In this paper, we will use the terms phase factors and choice of gauge (or simply gauge) interchangeably.

In this paper, given $H:D^* \rightarrow \C^{n\times n}$ as in (\ref{eq:hper}) above, we wish to construct an assignment of $\vct{u}$ on $D^*$ such that $\vct{u}$ is analytic on $D^*$ and satisfies the same periodicity in (\ref{eq:hper}) (so that the Fourier coefficients of elements in $\vct{u}$ decay exponentially asymptotically).
Furthermore, we also require the assignment of $\vct{u}$ to be optimal over all possible choices of $\varphi$ in the sense that the Fourier coefficients of $\vct{u}$ have a minimum spread defined by some variance function. Such a problem can be viewed as a matrix version of the construction of Wannier functions for the Schr\"odinger operator case.
(The physical meaning of the quantities considered above will be further clarified in Section\,\ref{sec:tb} and the variance function is defined in (\ref{eq:var}) in Section\,\ref{sec:wannier}.)


\section{Mathematical and physical preliminaries \label{sec:pre}}
\subsection{Notations}
We let $\delta_{ij}$ denote the Kronecker delta function, defined by 
\begin{equation}
\delta_{ij} = \begin{cases} 1 & i = j \\ 0 & i \neq j \end{cases}.
\end{equation} Suppose $\vct{v}$ is a vector in $\C^n$. We denote its norm by $\norm{\vct{v}}$, which is given by the standard inner product as $\norm{\vct{v}}^2=\vct{v}^*\vct{v}$. 
For vectors $\vct{a}, \vct{b}$ in $\R^n$, we also use the dot product notation $\vct{a}\cdot\vct{b}$ to represent their standard inner product $\vct{a}^*\vct{b}$. For any $n$ by $n$ matrix $M$, we denote the matrix trace by $\Tr M = \sum_{j=1}^n M_{jj}$.

Unless otherwise specified, we choose the principal branch of the complex $\log$ function, i.e. $\log(z) = \log(\abs{z}) + \I \,\mathrm{arg}(z)$\,, where $-\pi <\mathrm{arg}(z) \le \pi$\,.




\subsection{Periodic lattices in $\R^2$}
In this section, we introduce standard definitions for a periodic lattice in $\R^2$. 

Suppose the canonical basis vectors in $\R^2$ are given by
\begin{equation}
	\vct{e}_x = (1,0), \quad \vct{e}_y = (0,1)\,.
\end{equation}
We define a periodic lattice in two dimensions by two linearly independent vectors $\vct{a}_1$ and $\vct{a}_2$ in $\R^2$:
\begin{equation}
	\Lambda = \cb{\vct{R}= m_1 \vct{a}_1 + m_2 \vct{a}_2: m_1 ,m_2 \in\Z}\,.
	\label{eq:lambda}
\end{equation}
The vectors $\vct{a}_1$ and $\vct{a}_2$  are referred to as the real space primitive lattice vectors. Given such vectors, we also denote the reciprocal space primitive lattice vectors by $\vct{b}_1$ and $\vct{b}_2$ in $\R^2$, such that their inner product satisfies
\begin{equation}
\vct{a}_i\cdot \vct{b}_j =  2\pi\delta_{ij}, \quad i,j= 1,2\,.
	\label{eq:ab}
\end{equation}
Similarly, the reciprocal lattice is defined as
\begin{equation}
	\Lambda^* = \cb{\vct{G}= m_1 \vct{b}_1 + m_2 \vct{b}_2: m_1, m_2\in\Z}.
	\label{eq:rgortho}
\end{equation}
Due to the relation in (\ref{eq:rgortho}), for any vector $\vct{R}\in \Lambda$ and $\vct{G}\in \Lambda^*$, we have
\begin{equation}
 \vct{G}\cdot \vct{R} = 2\pi l 	
 \label{eq:gr}
\end{equation}
for some integer $l$\,.

We define the primitive unit cell in the real space by the formula
\begin{equation}
	D = \cb{r_1\vct{a}_1 + r_2 \vct{a}_2: r_1,r_2\in [-1/2,1/2]},
		\label{eq:D}
\end{equation}
and we denote its area $\abs{D}$ by $V_{\rm puc}$. We also define the primitive unit cell in the reciprocal space by the formula
\begin{equation}
	D^* = \cb{ \kappa_1\vct{b}_1 + \kappa_2 \vct{b}_2: \kappa_1,\kappa_2\in [-1/2,1/2]}.
	\label{eq:Ds}
\end{equation}
For any $\vct{k}$ in $D^*$, we denote its $x, y$ component by $k_x, k_y$ given by the formula
\begin{equation}
	\vct{k} = k_x \vct{e}_x + k_y \vct{e}_y\,.
\end{equation}
We define the torus $T$ by 
\begin{equation}
	T = \left[-\frac{1}{2},\frac{1}{2}\right] \times \left[-\frac{1}{2},\frac{1}{2}\right] \,,
	\label{eq:torus}
\end{equation}
which provides a natural domain for parameterizing $\vct{k} \in D^*$ by the formula
\begin{equation}
 \vct{k} =  \vct{k}(\kappa_1,\kappa_2) = \kappa_1\vct{b}_1 + \kappa_2\vct{b}_2\,, \quad (\kappa_1,\kappa_2)\in T\,.
	\label{eq:xydir}
\end{equation}
The components of $\vct{k}$ are given by
\begin{equation}
	k_x = k_x(\kappa_1,\kappa_2)\,,\quad 	k_y = k_y(\kappa_1,\kappa_2) \,, \quad (\kappa_1,\kappa_2)\in T\,.
		\label{eq:xydir2}
\end{equation}


\begin{remark}
The region $D^*$ defined in (\ref{eq:Ds}) is not necessarily the Wigner-Seitz cell of the lattice $\Lambda^*$ that defines the first Brillouin zone (BZ) in physics literature.  Due to periodicity, the domain $D^*$ serves the same purpose as the BZ, except for not capturing the full symmetry of the lattice.
\end{remark}


\subsection{Tight-binding models \label{sec:tb}}
In this section, we introduce the so-called tight-binding model to provide the physical context for the eigenvalue problem in (\ref{eq:heig}). They are standard theories in solid state physics and can be found, for example, in \cite{kaxiras2019quantum}.

Given a Schr\"odinger operator with a periodic potential on the lattice $\Lambda$ defined in (\ref{eq:lambda}),
a common strategy in band structure modeling is to approximate the operator by a set of $n$ Bloch-like functions $\chi_{\vct{k}, i}: \R^2 \rightarrow \C$ defined by the formula
\begin{equation}
	\chi_{\vct{k}, i}(\vct{r}) = \sum_{\vct{R}\in \Lambda} e^{-\I\vct{k}\cdot\vct{R}}\phi_i(\vct{r}+\vct{R})\,, \quad \vct{k}\in D^*\,, \vct{r}\in \R^2\,,\quad i=1,2,\ldots,n\,,
	\label{eq:blochlike}
\end{equation}
where $\phi_i: \R^2\rightarrow \R$ represents some pre-specified atomic orbitals (localized near $\vct{R=0}$). 
The eigenfunctions of the $j$-th eigenvalues of the Scr\"odinger operator can thus approximated by some linear combination of $\chi_{\vct{k}, i}$ in (\ref{eq:blochlike}) given by the formula
\begin{equation}
	\psi^{(j)}_{\vct{k}}(\vct{r}) = \sum_{i=1}^n u^{(j)}_{i}(\vct{k})\chi_{\vct{k}, i}(\vct{r})\,,
	\label{eq:lincom}
\end{equation}
where $u^{(j)}_{i}(\vct{k})$ for $i=1,2,\ldots,n$ are the coefficients to be found. 

By approximating the Schr\"odinger operator in the subspace spanned by $\chi_{\vct{k}, i}$, the original eigenvalue problem of the operator is reduced into that of a family of $n$ by $n$ matrices $H(\vct{k})$ for $\vct{k}\in D^*$. The eigenvectors of the matrices determine the coefficients  $u^{(j)}_{i}(\vct{k})$ in (\ref{eq:blochlike}), thus obtaining the approximating eigenfunctions $\psi^{(j)}_{\vct{k}}$. The family of matrices $H(\vct{k})$ is commonly referred to as the tight-binding Hamiltonian.
In general, various assumptions are made about the properties of the inner product among functions $\phi_i$ in (\ref{eq:blochlike});  parameters in $H(\vct{k})$ are often determined by least-square procedures against experimental data or first-principle calculations. We refer the reader to \cite{kaxiras2019quantum,papaconstantopoulos2015handbook} for details.

We form a family of vectors $\vct{u}_n: D^*\rightarrow \C^n $ containing all coefficients in (\ref{eq:lincom}) of the form
\begin{equation}
	\vct{u}_j (\vct{k}) = (u^{(j)}_{1}(\vct{k}),u^{(j)}_{2}(\vct{k}),\ldots,u^{(j)}_{n}(\vct{k})),
	\label{eq:ulabel}
\end{equation}
so that the eigenvalue problem for the tight-binding Hamiltonian becomes
\begin{equation}
	H(\vct{k}) \vct{u}_{j}(\vct{k}) = E_j(\vct{k}) \vct{u}_{j}(\vct{k})\,,
	\label{eq:heig2}
\end{equation}
with the normalization condition
\begin{equation}
	\norm{\vct{u}_j(\vct{k})}^2 = \vct{u}^*_j(\vct{k}) \vct{u}_j(\vct{k}) = 1\,.
\end{equation}
The eigenvalue index $j$ is also referred to as the band index, and 
the family of eigenvalues $E_j(\vct{k})$ for  $\vct{k}\in D^*$ is the energy levels of the $j$-th band. After solving the eigenvalue problem (\ref{eq:heig2}), the family of vectors $\vct{u}_j(\vct{k})$ for  $\vct{k}\in D^*$ together with (\ref{eq:lincom}) defines the Block-like functions for the $j$-th band. The eigenvalue problem defined in (\ref{eq:heig}) can be viewed as that in (\ref{eq:heig2}) for a particular band with the band index $j$ dropped.

We assume that the elements in the matrix $H$ are analytic on $D^*$ and are periodic in the following form
\begin{align}
	H(\vct{k}+\vct{G}) = H(\vct{k})\,,\quad\vct{k}\in D^*, \vct{G}\in\Lambda^*\,,
	\label{eq:hper2}
\end{align}
which is identical to the condition in (\ref{eq:hper}). We refer to periodic functions in the form of (\ref{eq:hper2}) as being $\Lambda^*$-periodic. Thus the eigenvalue $E_j(\vct{k})$ is also $\Lambda^*$-periodic:
\begin{align}
	E_j(\vct{k}+\vct{G}) = E_j(\vct{k})\,,\quad\vct{k}\in D^*, \vct{G}\in\Lambda^*\,.
	\label{eq:eper}
\end{align}
Throughout this paper, we assume the family of eigenvalues $E_j(\vct{k})$ of interests never becomes degenerate. The eigenvector $\vct{u}_j(\vct{k})$ is chosen to be periodic copies of that in $D^*$ so that we have
\begin{align}
	\vct{u}_j(\vct{k}+\vct{G}) = \vct{u}_j(\vct{k})\,,\quad\vct{k}\in D^*, \vct{G}\in\Lambda^*\,.
	\label{eq:uper}
\end{align}
The projector $P_j(\vct{k})=\vct{u}_j(\vct{k}) \vct{u}^*_j(\vct{k})$ has the same periodicity
\begin{align}
	P_j(\vct{k}+\vct{G}) = P_j(\vct{k})\,,\quad\vct{k}\in D^*, \vct{G}\in\Lambda^*\,.
		\label{eq:ppper}
\end{align}
We observe that, although the eigenvector $\vct{u}_j$ is only determined up to a phase factor, the projector $P_j$ is independent of such choices, thus is 
uniquely defined on $D^*$.

\begin{remark}
	In this paper, most functions of interest both analytic and $\Lambda^*$-periodic on $D^*$, such as $H$ in (\ref{eq:hper2}). 
	However, there are cases where it is still convenient to view $D^*$ not as a deformed torus, but as a subset of $\R^2$, so that a function can be analytic on $D^*$ without being $\Lambda^*$-periodic. For this reason, we always make it explicit when a function is both analytic \emph{and} $\Lambda^*$-periodic to indicate that $D^*$ is viewed as a torus, \emph{not} as a subset of $\R^2$.
\end{remark}

Due to the periodicity in (\ref{eq:uper}), each element $u^{(j)}_{i}$ in $\vct{u}_j$ for $i=1,2,\ldots,n$ can be represented by its Fourier series of the form
\begin{align}
	u^{(j)}_{i}(\vct{k}) = \sum_{\vct{R}\in\Lambda} u^{(j)}_{i,\vct{R}} \cdot e^{\I \vct{R}\cdot\vct{k}}.
	\label{eq:ufour}
\end{align}
The orthogonality for two different lattice vectors $\vct{R}, \vct{R}'\in\Lambda$
\begin{align}
	 \int_{D^*} \D\vct{k} \, e^{\I (\vct{R} - \vct{R}')\cdot\vct{k}} = \frac{(2\pi)^2}{ V_{\rm puc}}\delta_{\vct{R},\vct{R}'} 
	 \label{eq:forth}
\end{align}
allows us to compute the Fourier coefficients by the formula
\begin{align}
	u^{(j)}_{i,\vct{R}} = \frac{V_{\rm puc}}{(2\pi)^2} \int_{D^*} \D\vct{k} \, e^{-\I \vct{R}\cdot\vct{k}} \cdot u^{(j)}_{i}(\vct{k}),
	\label{eq:ufourc}
\end{align}
where $V_{\rm puc}$ is the area of the primitive unit cell $D$ in (\ref{eq:D}).

\subsection{Time-reversal symmetry}
If the matrix $H$ in (\ref{eq:heig2}) satisfies
\begin{align}
	\bar{H}(\vct{k}) = H(-\vct{k})\,, \quad \vct{k}\in D^*\,,
	\label{eq:trs}
\end{align}
we call that $H$ has time-reversal symmetry.

Suppose that $H$ has time-reversal symmetry. It is straightforward to show that the eigenvalue $E_j$ and projector $P_j$ have the same symmetry
\begin{align}
	\bar{P}_{j}(\vct{k})= P_{j}(-\vct{k})\,,\quad \vct{k}\in D^*\,,
	\label{eq:ptrs}
\end{align}
\begin{align}
	E_j(\vct{k}) = 	E_j(-\vct{k})\,,\quad \vct{k}\in D^*\,.
	\label{eq:etrs}
\end{align}



%
%
%
%
\subsection{Wannier functions \label{sec:wannier}}
In this section, we define Wannier functions for the matrix models and their variance that measures their localization. It should be observed that the variance definition for the matrix case in this paper is different from the standard one in the Schr\"odinger operator case\,\cite{marzari1997maximally}; the purpose of the new definition is to make the matrix case almost identical to the operator one in terms of minimizing the variance.

Analogous to the Wannier functions in the Schr\"odinger operator case\,\cite{marzari1997maximally,vanderbilt2018berry}, we define the Wannier function for the $j$-th band centered at some $\vct{R}\in\Lambda$ by the formula
\begin{align}
	W^{(j)}_{\vct{R}}(\vct{r}) = \frac{(2\pi)^2}{V_{\rm puc}}\int_{D^*}\D\vct{k}\, e^{-\I \vct{k}\cdot \vct{R}}\psi^{(j)}_{\vct{k}}(\vct{r})\,,
\end{align}
where $\psi^{(j)}_{\vct{k}}(\vct{r})$ is given in (\ref{eq:lincom}) and the domain $D^*$ is defined in (\ref{eq:Ds}). 

Due to the Bloch-like nature of (\ref{eq:blochlike}), the Wannier functions $W^{(j)}_{\vct{R}}(\vct{r})$ centered at $\vct{R}$ are copies of those at $\vct{R}=0$. More explicitly, we have
\begin{align}
	W^{(j)}_{\vct{R}}(\vct{r}) = W^{(j)}_{\vct{0}}(\vct{r} - \vct{R}).
	\label{eq:copy}
\end{align}
Hence, we only need to consider those centered around $\vct{R}=\vct{0}$:
\begin{align}
	W^{(j)}_{\vct{0}}(\vct{r}) = \frac{V_{\rm puc}}{(2\pi)^2}\int_{D^*}\D\vct{k}\, \psi^{(j)}_{\vct{k}}(\vct{r})\,.
		\label{eq:wannn}
\end{align}
By the definition of $\psi^{(j)}_{\vct{k}}$ in (\ref{eq:blochlike}), (\ref{eq:lincom}) and the Fourier series formulas (\ref{eq:ufour} - \ref{eq:ufourc}), we rewrite $W^{(j)}_{\vct{0}}$ in the following form
\begin{align}
	W^{(j)}_{\vct{0}}(\vct{r}) = \sum_{i=1}^n \sum_{\vct{R}\in\Lambda} u^{(j)}_{i,\vct{R}} \cdot \phi_i(\vct{r}+\vct{R}),
	\label{eq:wann}
\end{align}
where $u^{(j)}_{i,\vct{R}}$ are the Fourier coefficients of $u^{(j)}_{i}(\vct{k})$, whose Fourier series is given by the formula 
\begin{align}
	u^{(j)}_{i}(\vct{k}) = \sum_{\vct{R}\in\Lambda} u^{(j)}_{i,\vct{R}} \cdot e^{\I \vct{R}\cdot\vct{k}}.
	\label{eq:ufour2}
\end{align}
Due to (\ref{eq:wann}), the Wannier function defined in (\ref{eq:wannn}) is well-localized near $\vct{R}=\vct{0}$, provided $\abs{u^{(j)}_{i,\vct{R}}}$ decays rapidly as $\norm{\vct{R}}$ gets large. This can be achieved if the components of the eigenvector $\vct{u}$ are chosen as analytic and $\Lambda^*$-periodic functions on $ D^*$, so that $\abs{u_{i,\vct{R}}}$ decays exponentially for large $\norm{\vct{R}}$.

For the rest of this paper, we drop the band index label $j$ since we only deal with a single band. Hence the eigenvalue equation for the eigenvalue $E$ and eigenvector $\vct{u}$ of interests for $\vct{k}\in D^*$ becomes
\begin{align}
	H(\vct{k}) \vct{u}(\vct{k}) = E(\vct{k}) \vct{u}(\vct{k})\,.
	\label{eq:heig3}
\end{align}
Thus the Wannier function is given by the formula
\begin{align}
	W_{\vct{0}}(\vct{r}) = \frac{V_{\rm puc}}{(2\pi)^2}\int_{D^*}\D\vct{k}\, \psi_{\vct{k}}(\vct{r}) =  \sum_{i=1}^n \sum_{\vct{R}\in\Lambda} u_{i,\vct{R}} \cdot \phi_i(\vct{r} +\vct{R}),
	\label{eq:wann2}
\end{align}
and the Fourier series of the component $u_i$ of the eigenvector $\vct{u}$ is given by
\begin{align}
	u_{i}(\vct{k}) = \sum_{\vct{R}\in\Lambda} u_{i,\vct{R}} \cdot e^{\I \vct{R}\cdot\vct{k}}.
	\label{eq:ufour22}
\end{align}


We quantify the localization of $W_{\vct{0}}(\vct{r})$ by 
the following variance function
\begin{align}
	\langle \norm{\vct{R}}^2 \rangle - \norm{\langle \vct{R} \rangle}^2.
		\label{eq:var}
\end{align}
where the first and second moment functions are defined by the formulas
\begin{align}
	\langle \vct{R} \rangle := \sum_{\vct{R}\in\Lambda} \abs{u_{i,\vct{R}}}^2 (-\vct{R})\,,\quad
	\langle \norm{\vct{R}}^2 \rangle := \sum_{\vct{R}\in\Lambda} \abs{u_{i,\vct{R}}}^2 \norm{\vct{R}}^2\,.
	\label{eq:R2}
\end{align}
The first moment $\langle \vct{R} \rangle$ is commonly referred to as the Wannier center, where the minus sign is introduced for reasons as follows.

We observe that the moment functions in (\ref{eq:R2}) are different from the conventional definitions
\begin{align}
	\langle \vct{r} \rangle = \int_{\R^2} \D\vct{r}\, \abs{W_{\vct{0}}(\vct{r})}^2 \vct{r}\,,\quad
	\langle r^2 \rangle = \int_{\R^2} \D\vct{r}\, \abs{W_{\vct{0}}(\vct{r})}^2 r^2\,.
		\label{eq:r2}
\end{align}
The justification for choosing (\ref{eq:R2}) over (\ref{eq:r2}) is as follows. 
First, since the tight-binding model is already a physical approximation of the original problem, the choice of (\ref{eq:R2}) does not affect much the physical relevance of the final results.
Second, and more importantly, since this paper is a stepping stone for the Schr\"odinger operator case, it is desirable for all definitions to mimic the Wannier problem for the operator case mathematically.

Suppose that $\vct{u}$ in (\ref{eq:ulabel}) is chosen to be analytic and $\Lambda^*$-periodic (see (\ref{eq:uper})) on $D^{*}$. Then the Fourier coefficients in  (\ref{eq:R2}) decay exponentially and the moment functions (\ref{eq:R2}) are well-defined. Furthermore, by differentiating (\ref{eq:ufour2}) with respect to $\vct{k}$, we obtain the following expressions of (\ref{eq:R2}) in terms of the eigenvector $\vct{u}$:
\begin{align}
	\langle \vct{R} \rangle = \I \frac{V_{\rm puc}}{(2\pi)^2}\int_{D^*} \D\vct{k}\, \vct{u}^*(\vct{k})\grad_{\vct{k}} \vct{u}(\vct{k})\,,
	\label{eq:Ru}
\end{align}
\begin{align}
	\langle \norm{\vct{R}}^2 \rangle = \frac{V_{\rm puc}}{(2\pi)^2}\int_{D^*} \D\vct{k}\, \norm{\grad_{\vct{k}} \vct{u}(\vct{k})}^2.
		\label{eq:R2u}
\end{align}
We observe that (\ref{eq:Ru}) and (\ref{eq:R2u}) resemble the formulas of (\ref{eq:r2}) in terms of Bloch functions in the operator case\,\cite{marzari1997maximally} with only the eigenfunction $u_{\vct{k}}$ replaced by the eigenvector $\vct{u}(\vct{k})$.

\subsection{Systems of coordinates\label{sec:coord}}
This section contains basic formulas for changing systems of coordinates between variables $k_x, k_y$ and $\kappa_1, \kappa_2$ for parameterizing  $D^*$ in (\ref{eq:Ds}) defined in (\ref{eq:xydir}). 

In this paper, we mainly deal with $\Lambda^*$-periodic functions on $D^*$ of the form
\begin{align}
	f(\vct{k}) = f(\vct{k}+\vct{G})\,, \quad \vct{k}\in D^*\,, \vct{G}\in \Lambda^*.
\end{align}
By the parameterization in (\ref{eq:xydir}), the $\Lambda^*$-periodicity of $f$ is equivalent to 
\begin{align}
	f(\vct{k}(\kappa_1+m,\kappa_2+n)) = 	f(\vct{k}(\kappa_1,\kappa_2))\,, \quad (\kappa_1,\kappa_2)\in T\,, m,n\in \Z\,.
	\label{eq:zper}
\end{align}

For derivatives, by the orthogonality in (\ref{eq:ab}), we have
\begin{align}
	\frac{\partial f}{\partial k_x} = \frac{1}{2\pi}\vct{a}_1\cdot\vct{e}_x \frac{\partial f}{\partial \kappa_1} + \frac{1}{2\pi}\vct{a}_2\cdot\vct{e}_x \frac{\partial f}{\partial \kappa_2}\,, \quad
	\frac{\partial f}{\partial k_y} = \frac{1}{2\pi}\vct{a}_1\cdot\vct{e}_y \frac{\partial f}{\partial \kappa_1} + \frac{1}{2\pi}\vct{a}_2\cdot\vct{e}_y \frac{\partial f}{\partial \kappa_2}\,,
		\label{eq:diffc}
\end{align} 
where $\vct{a}_1$ and $\vct{a}_2$ are the real-space primitive lattice vectors in (\ref{eq:D}). By (\ref{eq:xydir}), we have similarly
\begin{align}
	\frac{\partial f}{\partial \kappa_1} = \vct{b}_1\cdot\vct{e}_x \frac{\partial f}{\partial k_x} +\vct{b}_1\cdot\vct{e}_y \frac{\partial f}{\partial k_y}\,,\quad
	\frac{\partial f}{\partial \kappa_2} = \vct{b}_2\cdot\vct{e}_x \frac{\partial f}{\partial k_x} +\vct{b}_2\cdot\vct{e}_y \frac{\partial f}{\partial k_y}\,.
		\label{eq:diffck}
\end{align}

Accordingly, for integrals over $D^*$, we have
\begin{align}
	\int_{D^*} \D\vct{k}\, f(\vct{k}) = \int_{D^*} \D k_x\D k_y\, f(\vct{k}) = \frac{(2\pi)^2}{V_{\rm puc}} \int_T \D \kappa_1 \D \kappa_2 	f(\vct{k}(\kappa_1,\kappa_2))\,,
	\label{eq:intvarc}
\end{align}
where the torus $T$ is defined in (\ref{eq:torus}).
By applying (\ref{eq:zper}) and (\ref{eq:intvarc}), we have
\begin{align}
	\int_{D^*} \D\vct{k}\, \frac{\partial f(\vct{k})}{\partial \kappa_1} = \frac{(2\pi)^2}{V_{\rm puc}}  \int_{-\frac{1}{2}}^{\frac{1}{2}}\D \kappa_2 \,\sb{f(\vct{k}(1/2,\kappa_2)) - f(\vct{k}(-1/2,\kappa_2))} = 0\,,
	\label{eq:k10}
\end{align}
\begin{align}
	\int_{D^*} \D\vct{k}\, \frac{\partial f(\vct{k})}{\partial \kappa_2} = \frac{(2\pi)^2}{V_{\rm puc}}  \int_{-\frac{1}{2}}^{\frac{1}{2}}\D \kappa_1 \,\sb{f(\vct{k}(\kappa_1,1/2)) - f(\vct{k}(\kappa_1,-1/2))} = 0\,.
	\label{eq:k20}
\end{align}
Combing (\ref{eq:diffc}) with (\ref{eq:k10}) and (\ref{eq:k20})\,, we thus have
\begin{align}
	\int_{D^*} \D\vct{k}\, \frac{\partial f(\vct{k})}{\partial k_x} = 0\,, \quad \int_{D^*} \D\vct{k}\, \frac{\partial f(\vct{k})}{\partial k_y} = 0\,.
	\label{eq:kxy0}
\end{align}

\subsection{Poisson's equation on tori \label{sec:pois}}
This section contains basic results for Poisson's equation defined on the torus $D^*$ (see (\ref{eq:Ds})). Consider the following Poisson's equation
\begin{align}
	\frac{\partial^2 f}{\partial k^2_x} + \frac{\partial^2 f}{\partial k^2_y} = - g\,,\quad  \mbox{$f$ is $\Lambda^*$-periodic on $D^*$\,,}
	\label{eq:poiss}
\end{align}
where $g:D^*\rightarrow \R$ is assumed to be an analytic and $\Lambda^*$-periodic  function, whose Fourier series is of the form
\begin{align}
	g (\vct{k}) = \sum_{\vct{R}\in \Lambda} g_{\vct{R}}\cdot e^{\I \vct{R}\cdot\vct{k}}\,, \quad \vct{k}\in D^*\,.
\end{align}
The analyticity of $g$ implies $\abs{g_{\vct{R}}}$ decays exponentially for large $\norm{\vct{R}}$. 
Suppose that its zeroth Fourier coefficient $g_{\vct{0}}$ is zero; equivalently, this condition is given by the formula
\begin{align}
	\int_{D^*}\D \vct{k}\, g(\vct{k}) = 0\,.
		\label{eq:zero}
\end{align}
Due to the condition in (\ref{eq:zero}), it is easy to verify that 
\begin{align}
	f (\vct{k}) = f_{\vct{0}} + \sum_{\substack{\vct{R}\in\Lambda\\ \vct{R}\ne \vct{0}}} \frac{g_{\vct{R}}}{\norm{\vct{R}}^2}\cdot e^{\I \vct{R}\cdot\vct{k}}\,, \quad \vct{k}\in D^*\,,
	\label{eq:poisol2}
\end{align}
solves (\ref{eq:poiss}), 
where $f_{\vct{0}}$ can be any real constant and may be chosen to be zero, i.e. $f_{\vct{0}}=0$\,. The solution $f: D^*\rightarrow \R$ in (\ref{eq:poisol2}) is analytic since its Fourier coefficients also decay exponentially. Thus, we have the following observation. 
\begin{theorem}
	There exists a unique solution (up to a constant) to Poisson's equation in (\ref{eq:poiss}) if and only if 
	\begin{align}
		\int_{D^*}\D \vct{k}\, g(\vct{k}) = 0\,.
	\end{align}
Moreover, the solution is analytic on $D^*$ and given by (\ref{eq:poisol2}). 
\label{thm:poiss}
\end{theorem}

\subsection{Vector fields on tori}
This section contains the Helmholtz-Hodge decomposition\,\cite{nakahara2018geometry} of vector fields defined on the torus $D^*$ in (\ref{eq:Ds}). In this special case, the decomposition of such vector fields is easily obtained from their Fourier series. 

Consider a vector field $\vct{f}: D^*\rightarrow \R^2$ given by the formula
\begin{align}
	\vct{f}(\vct{k}) = (f_x (\vct{k}), f_y (\vct{k}))\,.
\end{align}
We assume that $\vct{f}$ is an analytic and $\Lambda^*$-periodic function (see (\ref{eq:hper2})) on $D^*$.
Thus we expand $\vct{f}$ in its Fourier series
\begin{align}
	\vct{f}(\vct{k}) = \sum_{\vct{R}\in \Lambda} \vct{f}_{\vct{R}}\cdot e^{\I \vct{R}\cdot\vct{k}}\,,
	\label{eq:ffour}
\end{align}
where the constant vector $\vct{f}_{\vct{R}}\in \C^2$ defined by
\begin{align}
	\vct{f}_{\vct{R}} = (f_{x,\vct{R}} , f_{y, \vct{R}} )
\end{align}
contains the Fourier coefficients $f_{x,\vct{R}}$ and $f_{y, \vct{R}}$ of $f_x$ and $f_y$, respectively, at some lattice point $\vct{R}$.
We decompose $\vct{f}_{\vct{R}}$ into
\begin{align}
	\vct{f}_{\vct{R}} = \vct{R} \frac{\vct{R}\cdot \vct{f}_{\vct{R}}}{\norm{\vct{R}}^2} + \bb{\vct{f}_{\vct{R}} - \vct{R} \frac{\vct{R}\cdot \vct{f}_{\vct{R}}}{\norm{\vct{R}}^2}}\,,\quad \mbox{for any $\vct{R}\ne \vct{0}$}\,,
	\label{eq:fd}
\end{align}
where the first term is parallel to $\vct{R}$ and the second is orthogonal to $\vct{R}$. By making use of the components in $\vct{R}=(R_x,R_y)$, the second term in (\ref{eq:fd}) can be written as 
\begin{align}
	\vct{f}_{\vct{R}} - \vct{R} \frac{\vct{R}\cdot \vct{f}_{\vct{R}}}{\norm{\vct{R}}^2} = \bb{R_y , - R_x } \frac{R_y f_{x,\vct{R}} - R_x  f_{y,\vct{R}}}{\norm{\vct{R}}^2} \,.
		\label{eq:fd2}
\end{align}
Based on the decomposition of Fourier series in (\ref{eq:fd}) and (\ref{eq:fd2}), we define a potential $\psi:D^*\rightarrow \R$ by the formula
\begin{align}
	\psi(\vct{k}) = \sum_{\substack{\vct{R}\in\Lambda\\ \vct{R} \ne \vct{0}}}  e^{\I \vct{R}\cdot\vct{k}}\cdot \frac{\I \vct{R}\cdot \vct{f}_{\vct{R}}}{\norm{\vct{R}}^2}\,,\quad \vct{k}\in D^*\,,
		\label{eq:psi}
\end{align}
for the curl-free component, and a potential $F:D^*\rightarrow \R$ by the formula
\begin{align}
	F(\vct{k}) = - \sum_{\substack{\vct{R}\in\Lambda\\ \vct{R} \ne \vct{0}}}  e^{\I \vct{R}\cdot\vct{k}}\cdot\frac{\I R_x f_{y,\vct{R}} - \I R_y  f_{x,\vct{R}}}{\norm{\vct{R}}^2} \,,\quad \vct{k}\in D^*\,,
	\label{eq:F}
\end{align}
for the divergence-free component. Obviously, both $\psi$ and $F$ are analytic and $\Lambda^*$-periodic.
We substitute (\ref{eq:fd}) and (\ref{eq:fd2}) into (\ref{eq:ffour}) and utilize
	$\frac{\partial}{\partial k_x} e^{\I\vct{R}\cdot\vct{k}} = \I R_x e^{\I\vct{R}\cdot\vct{k}}\,, \frac{\partial}{\partial k_y} e^{\I\vct{R}\cdot\vct{k}} = \I R_y e^{\I\vct{R}\cdot\vct{k}}$
to arrive at the Helmholtz-Hodge decomposition of the vector field $\vct{f}$, given by the formula
\begin{align}
	\vct{f} = -\bb{\frac{\partial \psi}{\partial k_x}, \frac{\partial \psi}{\partial k_y} } + \bb{\frac{\partial F}{\partial k_y}, -\frac{\partial F}{\partial k_x} } + (h_x, h_y)\,,
	\label{eq:hhd}
\end{align}
where the harmonic components $(h_x, h_y)$ is given by the zeroth Fourier coefficients of $\vct{f}$ by the formula
\begin{align}
	(h_x, h_y) = (f_{x,\vct{0}}, f_{y,\vct{0}})\,,
	\label{eq:h}
\end{align}
which is obviously both curl-free and divergence-free.
It is easily verified that (\ref{eq:hhd}) is identical to (\ref{eq:ffour}) due to (\ref{eq:psi}), (\ref{eq:F}) and (\ref{eq:h}). We summarize this decomposition in the following observation.
\begin{theorem}
	Let $\vct{f}: D^*\rightarrow \R^2$ be an analytic and  $\Lambda^*$-periodic  vector field. Then there is a unique decomposition of $\vct{f}$ into a curl-free component due to the potential $\psi$, a divergence-free component due to the potential $F$ and the harmonic components $(h_x,h_y)$ in the form of (\ref{eq:hhd}).
	\label{thm:vecdecom}
\end{theorem}
Furthermore, according to Theorem\,\ref{thm:poiss}, (\ref{eq:psi}) and (\ref{eq:F}) imply that $\psi$ and $F$ satisfy the following Poisson's equations
\begin{align}
		\frac{\partial^2 \psi}{\partial k^2_x} + \frac{\partial^2 \psi}{\partial k^2_y} = - \bb{\frac{\partial f_x}{\partial k_x} +  \frac{\partial f_y}{\partial k_y} }\,, 
		\label{eq:pois}
\end{align}
\begin{align}
		\frac{\partial^2 F}{\partial k^2_x} + \frac{\partial^2 F }{\partial k^2_y} = - \bb{\frac{\partial f_y}{\partial k_x} -  \frac{\partial f_x}{\partial k_y} }\,,
		\label{eq:poif}
\end{align}
subject to the condition that $\psi$ and $F$ are both $\Lambda^*$-periodic on $D^*$. We observe that solving (\ref{eq:pois}) and (\ref{eq:poif}) only determines $\psi$ and $F$ up to a constant by Theorem\,\ref{thm:poiss}. It does not pose any ambiguity since we are mostly interested in their derivatives as in (\ref{eq:hhd}).

\subsection{Analyticity of eigenvalues and eigenvectors \label{sec:eig}}
In this section, we first introduce standard smoothness results for eigenvalues/vectors of an analytic family of matrices. The results in this section can be viewed as the higher-dimensional extension of the one-dimensional results in \cite{gopal2024high}, which is a summary of those in \cite{kato2013perturbation}. 
Moreover, we also include, in this section, formulas for the directional derivative of eigenvalues and eigenvectors; they are basically identical to those in \cite{gopal2024high} and are only included here for completeness.

In this section, we consider a family of eigenvalues $\lambda$ and its corresponding eigenvectors $\vct{v}\in\C^n$ of an analytic family of $n\times n$ Hermitian matrices $M: R \subseteq \R^2\rightarrow \C^{n\times n}$ , defined by the formula
\begin{align}
	M(x_1,x_2)\vct{v}(x_1,x_2) = \lambda(x_1,x_2) \vct{v}(x_1,x_2) \,, \quad (x_1,x_2) \in R\,,
	\label{eq:katoeig}
\end{align}
where the elements in $M$ are assumed to be analytic both in $R$. Furthermore, similar to the one-dimensional case in \cite{gopal2024high}, we assume that, for any $(z_1, z_2)\in \C^2$ in a complex neighborhood of $(x_1 ,x_2)\in R$, we have the extension of $A$ into the complex plane $\C^2$ satisfying the condition
\begin{align}
	M(z_1,z_2)^* = 	M(\bar{z}_1,\bar{z}_2)\,.
		\label{eq:anac}
\end{align}

It should be observed that the generality of the one-dimensional results is lost  (see Remark\,2.6.3 in \cite{kato2013perturbation}); it is {essential} that we assume the eigenvalue family $\lambda(x_1,x_2)$ is \emph{never degenerate} since we only consider a single eigenvalue in this paper.

By the definition of the projector $P$ via the formula\,\cite{kato2013perturbation}
\begin{align}
	P(x_1,x_2) = -\frac{1}{2\pi \I}\oint_{\mathcal{C}}  (M(x_1,x_2) - \xi)^{-1}\, \D \xi\,, \quad (x_1,x_2)\in R\,,
	\label{eq:projdef}
\end{align}
where the integration is over a contour $\mathcal{C}\subset \C$ enclosing only the eigenvalue $\lambda(x_1,x_2)$ for any $(x_1,x_2)\in R$, it is obvious that the projector $P$ is analytic on $R$. (Such a contour $\mathcal{C}$ is possible since we assume $\lambda$ is never degenerate.) The eigenvalue family $\lambda(x_1,x_2)$ is also analytic on $R$ by the formula
\begin{align}
	\lambda(x_1,x_2) = \Tr(P(x_1,x_2) M(x_1,x_2) )\,,\quad (x_1,x_2)\in R.
\end{align}
We summarize this observation in the following.

\begin{theorem}
	Suppose that $\lambda$ and $\vct{v}$ are the family of eigenvalues and eigenvectors of the matrices $M$ in (\ref{eq:katoeig}) defined on  $R\subseteq\R^2$. Suppose further that the eigenvalue $\lambda$ is not degenerate at any point in $R$. Then both the eigenvalue  $\lambda$ and the eigenprojector $P$ are analytic on $R$.
	\label{thm:eigana}
\end{theorem}
Since the projector $P$ (which can also be written as $P=\vct{v}\vct{v}^*$) is independent of the phase choice of $\vct{v}$, the projector  $P$ is always analytic even when $\vct{v}$ is not. 
In Section\,2.4.2 of \cite{kato2013perturbation}, Kato also gives a construction of an analytic family of $\vct{v}$ by solving an ODE. We define an analytic curve $\gamma: [a,b] \rightarrow \R^2$, starting at some point $\vct{x}_0\in R$:
\begin{align}
	\gamma(s) = (x_1(s),x_2(s))\in R\,, \quad \gamma(a) = \vct{x}_0\,,\quad s\in [a,b]\,.
	\label{eq:curve}
\end{align}
Applying this construction to the eigenvalue problem (\ref{eq:katoeig}) on $\gamma$ yields the following result.
\begin{theorem}
Let $\gamma$ be the curve given in (\ref{eq:curve}), $M$ be the family of matrices given in (\ref{eq:katoeig}) with the eigenvalue family $\lambda$. Let $\vct{v}_0$ be the eigenvector of a non-degenerate $\lambda(\vct{x_0})$ at $\gamma(a)=\vct{x}_0$. Consider the following initial value problem
\begin{align}
	\frac{\D}{\D s} \vct{v}(\gamma(s)) = Q(\gamma(s)) \vct{v}(\gamma(s))\,,\quad \vct{v}(\vct{x}_0) = \vct{v}_0\,,
	\label{eq:katoode}
\end{align}
where the matrix $Q$ is given by the commutator
\begin{align}
	Q(\gamma(s)) = \sb{\frac{\D P(\gamma(s))}{\D s}, P(\gamma(s))} = \frac{\D P(\gamma(s))}{\D s} P(\gamma(s))-  P(\gamma(s)) \frac{\D P(\gamma(s))}{\D s}\,.
\end{align}
The solution $\vct{v}$ on the curve $\gamma$ satisfies 
\begin{align}
	M(\gamma(s))\vct{v}(\gamma(s)) = 	\lambda(\gamma(s))\vct{v}(\gamma(s))\,,\quad s\in[a,b]
\end{align}
and is analytic on $[a,b]$.

\label{thm:katoode}
\end{theorem}
By differentiating $P^2=P$, we obtain
\begin{align}
	P\frac{\D P}{\D s} +  \frac{\D P}{\D s}P = \frac{\D P}{\D s}\,.
	\label{eq:piden1}
\end{align}
Multiplying (\ref{eq:piden1}) by $P$ on the left produces
\begin{align}
	P\frac{\D P}{\D s} P = 0\,.
	\label{eq:piden2}
\end{align}
Since $\vct{v}$ stays an eigenvector along the curve $\gamma$ in Theorem\,\ref{thm:katoode}, we have $P\frac{\D P}{\D s}\vct{v}=P\frac{\D P}{\D s}P\vct{v} =0$ by (\ref{eq:piden2}) and $\frac{\D P}{\D s}P\vct{v}=\frac{\D P}{\D s}\vct{v}$. We use these two relations to simplify the ODE in (\ref{eq:katoode}) into the following
\begin{align}
	\frac{\D}{\D s} \vct{v}(\gamma(s)) = \frac{\D P(\gamma(s))}{\D s}\vct{v}(\gamma(s))\,.
	\label{eq:katoode2}
\end{align}

Next, we apply the well-known perturbation formulas for projectors to (\ref{eq:katoode2}) (see Section\,2.5.4 in \cite{kato2013perturbation}) to derive the derivative of eigenvectors in terms of eigenvectors. We summarize this fact together with the derivative of eigenvalues in the following theorem.
\begin{theorem}
	Let $\gamma$ be the curve given in (\ref{eq:curve}) and $M$ be the family of matrices given in (\ref{eq:katoeig}) with a non-degenerate eigenvalue $\lambda$ and the corresponding eigenvector $\vct{v}$. Then we have the following formulas:
	\begin{align}
		\frac{\D \lambda(\gamma(s))}{\D s} = \vct{v}^*(\gamma(s)) \frac{\D M(\gamma(s))}{\D s} \vct{v}(\gamma(s)),
				\label{eq:evalpert}
	\end{align}
	\begin{align}
		\frac{\D \vct{v}(\gamma(s))}{\D s} = - (M(\gamma(s)) - \lambda(\gamma(s)))^\dagger \frac{\D M(\gamma(s))}{\D s} \vct{v}(\gamma(s)),
		\label{eq:evecpert}
	\end{align}
where $(M - \lambda)^\dagger$ is the pseudoinverse of the matrix $M - \lambda$ ignoring the eigensubspace of $\lambda$.
\end{theorem}
By the definition of pseudoinverse, the range of $(M - \lambda)^\dagger$ is orthogonal to that of $P$, so we have the following formula
\begin{align}
	(M - \lambda)^\dagger P = P(M - \lambda)^\dagger = 0.	
	\label{eq:orthopv}
\end{align}

 We observe that (\ref{eq:katoode}), (\ref{eq:katoode2}), and (\ref{eq:evecpert}) are mathematically equivalent. In this paper, we will use (\ref{eq:katoode2}) for analytic purposes and (\ref{eq:evecpert}) for numerical ones.
Moreover, similar to Remark 2.10 in \cite{gopal2024high}, we can easily replace the pseudoinverse in (\ref{eq:evecpert}) with a true inverse for a similar matrix, which is more attractive numerically. Since we only deal with small matrices in this paper, this is not significant.

\begin{remark}
\label{rmk:svd}
The formula (\ref{eq:evecpert}) can be written by the spectral decomposition of $M$ in the following form
	\begin{align}
		\frac{\D \vct{v}}{\D s} = \sum_{\substack{j=1 \\ \lambda_j \ne \lambda }}^{N}\vct{v}_j \frac{\vct{v}_j^* \frac{\D M}{\D s} \vct{v}}{\lambda- \lambda_j},
		\label{eq:evecpertsvd}
	\end{align}
	where $\lambda_j$ and $\vct{v}_j$ are the rest of the eigenvalues and eigenvectors of $M$. 
This is equivalent to computing $\frac{\D \vct{v}}{\D s}$ by
the singular value decomposition of the matrix $M- \lambda$.
\end{remark}


\subsection{Ordinary differential equations \label{sec:ode}}
In this section, we introduce a smoothness result for solutions to systems of linear ODEs with analytic coefficients. It is a straightforward generalization of the continuity of solutions for ODEs with continuous coefficients, which can be found, for example, in \cite{pontryagin2014ordinary}. 

Consider $R = [t_0,t_1]\times \sb{-\frac{a}{2}, \frac{a}{2} }$ for some $t_1>t_0$ and $a>0$, and a family of $n\times n$ matrices $M: R\rightarrow \C^{n\times n}$. We assume that $M$ is analytic in $R$ and is a periodic function of $\mu$ with period $a$:
\begin{align}
	M(t,\mu) = 	M(t,\mu + ma)\,, \quad (t,\mu)\in R\,, m\in \Z\,.
\end{align}
Consider the following ODE system 
\begin{align}
	\frac{\D}{\D t} \vct{v}(t,\mu) = M(t,\mu)\vct{v}(t,\mu)\,, \quad  (t,\mu)\in R\,,
	\label{eq:ode}
\end{align}
with the initial condition
\begin{align}
	\vct{v}(t_0,\mu) = \vct{v}^0(\mu)\,, \quad  \mu \in \sb{-\frac{a}{2}, \frac{a}{2} }\,,
		\label{eq:ode0}
\end{align}
where we also assume $\vct{v}_0$ is analytic and has period $a$
\begin{align}
	\vct{v}^0(\mu) = 	\vct{v}^0(\mu + na) \,,\quad \mu \in\sb{-\frac{a}{2}, \frac{a}{2} }\,,\quad n\in \Z\,.
\end{align}
The variable $t$ parameterizes the ODE solution $\vct{v}$ and the other variable $\mu$ can be viewed as a numerical parameter.

By applying the Picard iteration to (\ref{eq:ode}), the solution $\vct{v}$ is given by the limit of the following uniformly convergent sequence of functions
\begin{align}
	\vct{v}_0(\mu), \quad	\vct{v}_0(\mu) + \int_{t_0}^t \D s M(s,\mu)\vct{v}_0(\mu)\,,\quad \ldots\,.
\end{align}
The form of the functions shows that the solution has the same periodicity as $M$ and $\vct{v}^0$. 
Moreover, since the limit of a sequence of uniformly convergent analytic functions is analytic, we conclude that the solution $\vct{v}$ is also analytic on $R$. If we further assume that $M$ is Hermitian, i.e. $M^*=M$ on $R$, and that the transpose of the initial condition $(\vct{v}^0)^*$ is also analytic on $\sb{-\half,\half}$, the same argument applied to the transpose of (\ref{eq:ode}) shows that the transpose of the solution $\vct{v}^*$ is also analytic on $R$.
\begin{lemma}
	\label{lem:ode}
	The solution $\vct{v}$ to (\ref{eq:ode}) subject to the initial condition (\ref{eq:ode0}) is analytic in $R$ and is a periodic function of the form 
	\begin{align}
	 \vct{v}(t,\mu) = 	 \vct{v}(t,\mu + m a)\,, \quad (t,\mu)\in R\,, m\in \Z\,.
\end{align}
Moreover, if we assume $M=M^*$ in $R$ and $(\vct{v}^0)^*$ is analytic, the transpose of the solution $\vct{v}^*$ is also analytic on $R$ and has the same periodicity as $\vct{v}$.
\end{lemma}

\subsection{Pfaffian systems\label{sec:pfaf}}
In this section, we introduce the integrability condition for systems of first-order partial differential equations in $\R^2$. 

Let $R$ be a subset of $\R^2$ and $M_1, M_2: R\rightarrow \C^{n\times n}$ be two families of matrices containing twice continuously differentiable coefficients. Consider a system of first-order partial differential equations given by the formulas
\begin{align}
	\frac{\partial }{\partial x_1} \vct{v}(x_1,x_2) = M_1(x_1,x_2)\vct{v}(x_1,x_2)\,, \quad 	\frac{\partial }{\partial x_2} \vct{v}(x_1,x_2) = M_2(x_1,x_2)\vct{v}(x_1,x_2)\,,
	\label{eq:pfaf}
\end{align}
with the initial condition 
\begin{align}
	\vct{v}(x_1^0,x_2^0) = \vct{v}^0
	\label{eq:pfinit}
\end{align}
for some $(x_1^0,x_2^0) \in R$ and a constant vector $\vct{v}^0\in\C^n$. The system (\ref{eq:pfaf}) is often referred to as a Pfaffian system\,\cite{ciarlet2013linear}.

We observe that each equation in (\ref{eq:pfaf}) together with the initial condition (\ref{eq:pfinit}) defines an initial value problem in one direction and their solutions might not be compatible with each other. As a result, the system (\ref{eq:pfaf}) is overdetermined and may not define $\vct{v}$ as a function over $R$. The following theorem states that the system (\ref{eq:pfaf}) becomes integrable, i.e. it admits a solution $\vct{v}$ on $R$, if and only if the second order derivatives are commutative and $R$ is simply connected. (It is analogous to the classical Poincar\'e lemma for vector fields.) For the particular choice of derivatives with respect to $x_1,x_2$, the commutativity condition is given by the formula 
\begin{align}
	\frac{\partial }{\partial x_1} 	\frac{\partial }{\partial x_2} \vct{v} = \frac{\partial }{\partial x_2} 	\frac{\partial }{\partial x_1} \vct{v} \,.
	\end{align}
The solution is also unique by the uniqueness theorem for initial value problems. 
Thus we have the following observation, which is a special case of Theorem\,6.20 in \cite{ciarlet2013linear} stated in a slightly different form. 

\begin{theorem}
	Suppose that $R$ is a simply connected domain in $\R^2$.
	Then the Pfaffian system  (\ref{eq:pfaf}) with the initial condition (\ref{eq:pfinit}) has a unique solution $\vct{v}$ on $R$ if and only if the second order derivatives of $\vct{v}$ are commutative on $R$. 
\label{thm:pfaf}
\end{theorem}

\section{Numerical preliminaries \label{sec:npre}}
This section introduces the numerical tools used in this paper. 

\subsection{Trapezoidal rule for periodic functions \label{sec:trapz}}
We introduce the trapezoidal rule for approximating integrals of periodic functions.

Let $f:\sb{0,L}\rightarrow \C$ be an analytic and periodic function of the form
\begin{align}
	f(x) = 	f(x+Lm)\,, \quad x\in \sb{0,L}\,,\quad m\in \Z\,.
	\label{eq:trapzper}
\end{align}
Let $N$ be a positive integer and $h=\frac{L}{N}$. 
We denote an $N$-point trapezoidal rule approximation of the integral of $f$ over $[0,L]$ by $I_{N}(f)$, defined by the formula
\begin{align}
	I_N(f) = h\sum_{j=0}^{N-1} f(j h)\,.
	\label{eq:trapz}
\end{align}
The following fact is a slightly different version of Theorem\,3.2 in \cite{trefethen2014exponentially}, stating that the approximation $I_N$ converges exponentially as the number of points $N$ increases.
\begin{theorem}
	Suppose that $f:\sb{0,L}\rightarrow \C$ is analytic and periodic in the form of (\ref{eq:trapzper}). Then for any positive integer $N$, there exist positive real numbers $C$ and $a$ such that
	\begin{align}
		\abs{I_N(f) - \int_{0}^{L} f(x)\,\D x } < C e^{-a N}\,.
	\end{align}
	\label{thm:trapz}
\end{theorem}
We observe that Theorem\,\ref{thm:trapz} is easily generalized to higher dimensions by recursively applying (\ref{eq:trapz}) to each dimension.

\subsection{Discrete Fourier transform in two dimensions \label{sec:dft}}
Suppose that $g:D^*\rightarrow \C$ is analytic and $\Lambda^*$-periodic. We denote its Fourier coefficient by $g_{\vct{R}}$ at the lattice point $\vct{R} = m_1\vct{a}_1 + m_2\vct{a}_1 \in \Lambda$ for some integer $m_1,m_2$ and $g_{\vct{R}}$ is given by the formula (see (\ref{eq:ufourc}))
\begin{align}
	g_{\vct{R}} =& \frac{V_{\rm puc}} {{  (2\pi)^2}}\int_{D^*}\D\vct{k}\, e^{-\I \vct{k}\cdot \vct{R}}g(\vct{k})\\
	 =& \int_T \D\kappa_1\D\kappa_2 \,e^{-2\pi\I m_1\kappa_1} e^{-2\pi\I m_2\kappa_2}g(\vct{k}(\kappa_1,\kappa_2)) \,,
	\label{eq:predft}
\end{align}
where (\ref{eq:intvarc}) and (\ref{eq:ab}) are used for obtaining the second equality. For simplicity, let $N$ be a positive even integer and $h=\frac{1}{N}$. Suppose that $m_1, m_2$ are restricted to $-N/2,-(N/2-1),\ldots,N/2-1$. Applying the trapezoidal rule approximation (\ref{eq:trapz}) to both variables $\kappa_1$ and $\kappa_2$ in (\ref{eq:predft}) yields  the approximation $\hat{g}_{m_1 m_2}$ of the Fourier coefficient  $g_{\vct{R}}$ via a discrete Fourier transform
\begin{align}
\hat{g}_{m_1 m_2}= \frac{1}{N^2} \sum_{j_1=-N/2}^{N/2-1}\sum_{j_2=-N/2}^{N/2-1} e^{-\frac{2\pi\I}{N} m_1 j_1} e^{-\frac{2\pi\I}{N} m_1 j_1} g(\vct{k}(j_1 h, j_2 h))\,.
	\label{eq:dft}
\end{align}
We observe that the approximation $\hat{g}_{mn}$ converges to $g_{\vct{R}}$ exponentially as $N$ increases by the two-dimensional version of Theorem\,(\ref{thm:trapz}).
We denote the discrete Fourier transform on the right of (\ref{eq:dft}) by $\mathscr{F}_N $, turning (\ref{eq:dft}) into
\begin{align}
	\hat{g}_{m_1 m_2} = \bb{ \mathscr{F}_N(g)}_{m_1 m_2}\,.
\end{align}
It is straightforward to verify that the inverse $\mathscr{F}^{-1}_N$ is given by the formula
\begin{align}
	g(\vct{k}(j_1 h, j_2 h)) =\bb{ \mathscr{F}^{-1}_N(\hat{g})}_{j_1 j_2} = \sum_{m_1=-N/2}^{N/2-1}\sum_{m_2=-N/2}^{N/2-1} e^{\frac{2\pi\I}{N} m_1 j_1} e^{\frac{2\pi\I}{N} m_2 j_2} \hat{g}_{m_1 m_2}\,,
\end{align}
for $j_1,j_2 = -N/2,-(N/2-1),\ldots,N/2-1$\,. Both $\mathscr{F}_N$ and $\mathscr{F}^{-1}_N$ can be applied via the Fast Fourier transform (FFT) in  $O(N^2\log{N})$ operations. The details can be found, for example, in \cite{briggs1995dft}.

Since derivatives $\frac{\partial}{\partial \kappa_1}$ and $\frac{\partial}{\partial \kappa_2}$ on the basis $e^{2\pi \I m_1 \kappa_1}e^{2\pi\I  m_2 \kappa_2}$ produce $2\pi\I m_1$ and $2\pi\I m_2$ respectively, we compute an approximation of $\frac{\partial}{\partial \kappa_1} g$ and  $\frac{\partial}{\partial \kappa_2} g$ by the formulas
\begin{align}
	\begin{split}
			\frac{\partial}{\partial \kappa_1} g(\vct{k}(j_1 h, j_2 h)) \approx \sum_{m_1=-N/2}^{N/2-1}\sum_{m_2=-N/2}^{N/2-1} e^{\frac{2\pi\I}{N} m_1 j_1} e^{\frac{2\pi\I}{N} m_2 j_2} (2\pi\I m_1\hat{g}_{m_1 m_2})\,,\\
			\frac{\partial}{\partial \kappa_2} g(\vct{k}(j_1 h, j_2 h)) \approx \sum_{m_1=-N/2}^{N/2-1}\sum_{m_2=-N/2}^{N/2-1} e^{\frac{2\pi\I}{N} m_1 j_1} e^{\frac{2\pi\I}{N} m_2 j_2} (2\pi\I m_2\hat{g}_{m_1 m_2})\,,
				\end{split}
	\label{eq:ddftg}
\end{align}
for $j_1,  j_2 = -N/2,-(N/2-1),\ldots,N/2-1$\,.
It is convenient to introduce the notation $\mathscr{D}_1$ and $\mathscr{D}_2$ for the component-wise multiplication by $2\pi\I m_1$ and $2\pi\I m_2$ in (\ref{eq:ddftg}) respectively, so that (\ref{eq:ddftg}) is compactly written as 
\begin{align}
	\frac{\partial}{\partial \kappa_1} g \approx  &  \mathscr{F}^{-1}_N \mathscr{D}_1 \mathscr{F}_N(g)\,,\quad \frac{\partial}{\partial \kappa_2} g \approx    \mathscr{F}^{-1}_N \mathscr{D}_2 \mathscr{F}_N(g)\,.
\end{align}
Since the approximation $\hat{g}$ in (\ref{eq:dft}) converges exponentially, so are those in (\ref{eq:ddftg}). 

\subsection{Fourth-order Runge-Kutta method \label{sec:rk4}}
A standard fourth-order Runge-Kutta method (see  \cite{dahlquist2003numerical}, for example)  solves the following initial value problem 
\begin{equation*}
	y'(t) = f(t,y)\,,\quad y(t_0) = y_0\,,
\end{equation*}
defined for $t\in [t_0,t+L]$ via the formulas
\begin{align*}
&	t_{i+1} = t_i + h \\
&	k_1 = h f(t_i,y_i)\,,\quad	k_2 = h f(t_i +\frac{1}{2}h, y_i + \frac{1}{2}k_1)\,,\\
&	k_3 = h f(t_i +\frac{1}{2}h, y_i + \frac{1}{2}k_2)\,,\quad k_4 = h f(t_i + h, y_i + k_3)\,,\\
&   y(t_{i+1})= y(t_i) +\frac{1}{6} (k_1 +2 k_2 +  2k_3 +  k_4)\,,
\end{align*}
with $i=0,1,\ldots,N$ and  $h=\frac{L}{N}$\,. For a smooth $f$, the global truncation error is $O(h^4)$. More explicitly, the approximation $y(t,h)$ at a fixed point $t$ has an expansion of the form
\begin{align}
	y(t,h) = y(t) + c_4(t) h^4 + c_5(t) h^5 + O(h^6)\,.
\end{align}
Applying the Richardson extrapolation via the formulas
\begin{align}
\begin{split}
	&\Delta_1(t) = \frac{1}{15}(16y(t,h/2) - y(t,h)), 	\quad \Delta_2(t) = \frac{1}{15}(16y(t,h/4) - y(t,h/2))\,,\\
	& \Delta_3(t) = \frac{1}{31}(32\Delta_2(t)- \Delta_1(t))\,,
\end{split}
	\label{eq:richard}
\end{align}
yields an approximation $\Delta_3(t)$ with $O(h^6)$ global truncation error.

\section{Analytic apparatus \label{sec:ap}}
This section introduces the analytic apparatus for constructing optimal Wannier functions in Section\,\ref{sec:method1}. In Section\,\ref{sec:para}, we define the parallel transport equation for assigning eigenvectors. In Section\,\ref{sec:bc}, we introduce the Berry connection and curvature. Moreover,
we relate the integrability of Pfaffian systems arising from parallel transport to the Berry curvature. Section\,\ref{sec:gt} introduces gauge transformations. Section\,\ref{sec:chern} introduces the first Chern number as the topological obstruction to exponentially localized Wannier functions. Section\,\ref{sec:opt} contains the formulas for the variance of Wannier functions and the optimal gauge choice. 

Consider the eigenvalue problem defined by some analytic family of matrices $H:D^* \rightarrow \C^{n\times n}$ introduced in Section\,\ref{sec:tb} by the formula
\begin{align}
	H(\vct{k}) \vct{u}(\vct{k}) = E(\vct{k}) \vct{u}(\vct{k})\,,\quad \vct{k}\in D^*\,,
	\label{eq:heig3}
\end{align}
where $\vct{u}$ and $E$ are the family of eigenvectors and eigenvalues. We assume that $E$ is never degenerate. The family of matrices $H$ is $\Lambda^*$-periodic and analytic on $D^*$ (see (\ref{eq:hper2})) and is assumed to satisfy the condition (\ref{eq:anac}). By applying Theorem\,\ref{thm:eigana} to (\ref{eq:heig3}), we conclude the both the eigenvalue $E$ and the projection $P=\vct{u}\vct{u}^*$ are $\Lambda^*$-periodic and analytic on $D^*$\,. More explicitly, the projector $P$ satisfies
\begin{align}
	P(\vct{k})= 	P(\vct{k} + \vct{G}) \,,\quad \vct{k}\in D^*\,, \vct{G}\in \Lambda^*\,.
	\label{eq:pper}
\end{align}


\subsection{Parallel transport equation \label{sec:para}}
In this section, we apply results in Section\,\ref{sec:eig} to derive the parallel transport equation to assign eigenvector $\vct{u}$ in (\ref{eq:heig3}) along curves in $D^*$.

Suppose that $\gamma: [0,1]\rightarrow D^*$ is an analytic curve with $\gamma(0) = \vct{k}^0$. From Kato's construction in Theorem\,\ref{thm:katoode}, we can assign eigenvectors analytically as a function of $\vct{k}$ along $\gamma\subset D^*$ by solving the following ODE 
\begin{align}
	\frac{\D \vct{u}(\gamma(s))}{\D s} =\frac{\D P(\gamma(s))}{\D s}  \vct{u}(\gamma(s))\,,
	\label{eq:oden}
\end{align}
subject to the initial condition
\begin{align}
	\vct{u}(\vct{k}^0) = \vct{u}^0\,,
			\label{eq:icondn}
\end{align}
where $\vct{u}_0$ is the eigenvector of $H(\vct{k}_0)$. We observe that 
\begin{align}
\vct{u}^*\frac{\D P}{\D s} \vct{u} = \vct{u}^*P\frac{\D P}{\D s}P \vct{u}	
\end{align}
since $\vct{u}$ is an eigenvector by Theorem\,\ref{thm:katoode}. Due to (\ref{eq:piden2}), we always have
\begin{align}
\vct{u}^*(\gamma(s))\frac{\D \vct{u}(\gamma(s))}{\D s} =0
		\label{eq:para}
\end{align}
along $\gamma$. Hence the change of $\vct{u}$ is always orthogonal to $\vct{u}$ and (\ref{eq:oden}) is referred to as the parallel-transport equation. Moreover, the condition (\ref{eq:para}) implies that the normalization is unchanged:
\begin{align}
	\norm{\vct{u}(\vct{k})} = 1\,,\quad \mbox{for any $\vct{k}$ in $\gamma$.}
\end{align}

The apparent difficulty in extending such solutions along curves  (families of one-dimensional assignments) to be a globally smooth two-dimensional assignment of $\vct{u}$ is as follows. Consider two curves passing through some $\vct{k}\in D^*$ with the tangent vector in two different directions. If we choose the directions to be in the
$\vct{e}_x$ and $\vct{e}_y$ (see (\ref{eq:xydir})), we obtain a Pfaffian system (see Section\,\ref{sec:pfaf}) given by the formulas
\begin{align}
		\frac{ \partial \vct{u}(\vct{k})}{\partial {k_x}} =\frac{\partial P(\vct{k})}{\partial {k_x}}  \vct{u}(\vct{k})\,,\quad		\frac{ \partial \vct{u}(\vct{k})}{\partial {k_y}} =\frac{\partial P(\vct{k})}{\partial {k_y}}  \vct{u}(\vct{k})\,.
		\label{eq:odek1}
\end{align}
Equivalently, if the directions are given by $\vct{b}_1$ and $\vct{b}_2$, we obtain a similar system
\begin{align}
	\frac{ \partial \vct{u}(\vct{k})}{\partial {\kappa_1}} =\frac{\partial P(\vct{k})}{\partial \kappa_1}  \vct{u}(\vct{k})\,,\quad		\frac{ \partial \vct{u}(\vct{k})}{\partial \kappa_2} =\frac{\partial P(\vct{k})}{\partial \kappa_2}  \vct{u}(\vct{k})\,,
	\label{eq:odekp1}
\end{align}
where the parameterization $\vct{k}=\vct{k}(\kappa_1,\kappa_2)$ is given in (\ref{eq:xydir}). (We will use (\ref{eq:odek1}) primarily for analytic purposes and (\ref{eq:odekp1}) for numerical ones.)
For any solution $\vct{u}$ obtained from solving (\ref{eq:oden}) in any subset of $D^*$, the solution $\vct{u}$ must satisfy both equations in (\ref{eq:odek1}). However, the system (\ref{eq:odek1}) is overdetermined and does not always meet the integrability condition in Theorem\,\ref{thm:pfaf}. 

The non-integrability of the Pfaffian system (\ref{eq:odek1}) can also be understood by solving (\ref{eq:oden}) with the initial condition (\ref{eq:icondn}) around a closed loop $\gamma_c$ (that starts and ends at the same point $\vct{k}^0$). The solution at the endpoint is also an eigenvector of $H(\vct{k}^0)$ of the same band (as we assume the eigenvalue is  non-degenerate), but it could differ from the initial $\vct{u}^0$ by a path-dependent phase factor $e^{\I \varphi_{\gamma_c}}$:
\begin{align}
	\vct{u}^0 \rightarrow e^{\I \varphi_{\gamma_c}}\vct{u}^0\,.
\end{align}
This is analogous to the case in classical differential geometry, where parallel transporting a vector on a surface around a closed loop results in a change in the angle between the initial and final vector when the Gaussian curvature is nonzero\,\cite{berry1989quantum}.

\subsection{Berry connection, curvature and integrability \label{sec:bc}}
This section introduces the so-called Berry connection and Berry curvature related to the system (\ref{eq:odek1}). These quantities play an important role in defining the assignment of $\vct{u}$ in $D^*$ and the localization of its corresponding Wannier functions. 
The approach here is based on standard ideas in differential geometry but is somewhat unconventional. The standard definitions of the Berry connection and curvature in physics literature are reproduced.

Let $U$ be a subset of $D^*$ that contains the curve $\gamma$ in (\ref{eq:oden}). Suppose that we define a continuously differentiable vector field $\vct{A}:U \rightarrow \R^2$ by the formula
\begin{align}
	\vct{A}(\vct{k}) = (A_x(\vct{k}), A_y(\vct{k}))\,.
	\label{eq:avec}
\end{align}
We modify (\ref{eq:odek1}) by adding an extra term so that we obtain a new Pfaffian system, given by the formulas
\begin{align}
		\frac{ \partial\tilde{\vct{u}}}{\partial {k_x}}= \frac{\partial P}{\partial {k_x} } \tilde{\vct{u}}  - \I A_x \tilde{\vct{u}}\,,\quad 
				\frac{ \partial\tilde{\vct{u}}}{\partial {k_y}}= \frac{\partial P}{\partial {k_y} } \tilde{\vct{u}}  - \I A_y \tilde{\vct{u}}\,.
		\label{eq:odekn1}
\end{align}
When (\ref{eq:odek1}) and (\ref{eq:odekn1}) are solved on the same curve $\gamma$ subject to the same initial condition (\ref{eq:icondn}),
we observe that  the solution $\vct{u}(\vct{k})$ to (\ref{eq:odek1}) only differs from $\tilde{\vct{u}}(\vct{k})$ to (\ref{eq:odekn1}) by a phase factor $e^{-\I\varphi_{\gamma}(\vct{k})}$ at any $\vct{k}$ in $\gamma$\,, where $\varphi_{\gamma}$ is given by the line integral of $\vct{A}$ from $\vct{k}^0$ to $\vct{k}$ along $\gamma$\,:
\begin{align}
	\varphi_{\gamma}(\vct{k}) = \int_{\vct{k}^0}^{\vct{k}} \vct{A}\cdot \D \vct{l}\,.
\end{align}
The vector field $\vct{A}$  and the phase $\varphi_{\gamma}$ are referred to as the Berry connection and the Berry phase\,\cite{berry1984quantal, nakahara2018geometry}. Due to (\ref{eq:para}), $A_x, A_y$ are also given by the formulas
\begin{align}
	A_x = \I \tilde{\vct{u}}^* \frac{\partial \tilde{\vct{u}}}{\partial {k_x}}\,,\quad A_y = \I \tilde{\vct{u}}^* \frac{\partial \tilde{\vct{u}}}{\partial {k_y}}\,,
	\label{eq:phybcon}
\end{align}
which are the definition in physics literature\,\cite{berry1984quantal,nakahara2018geometry}. Similarly, the system (\ref{eq:odek1}) becomes
\begin{align}
	\frac{ \partial\tilde{\vct{u}}}{\partial{\kappa_1}}= \frac{\partial P}{\partial{\kappa_1} } \tilde{\vct{u}}  - \I A_{1} \tilde{\vct{u}}\,,\quad 
	\frac{ \partial\tilde{\vct{u}} }{\partial{\kappa_2}}= \frac{\partial P}{\partial{\kappa_2} } \tilde{\vct{u}} - \I A_2 \tilde{\vct{u}}\,,
	\label{eq:odekk}
\end{align}
where $A_1$ and $A_2$ are respectively the components of $\vct{A}$ in (\ref{eq:avec}) in the $\vct{b}_1$ and $\vct{b}_2$ direction 
\begin{align}
	A_1 = \vct{b}_1\cdot\vct{A}\,,\quad 	A_2 = \vct{b}_2\cdot\vct{A}\,.
\end{align}

Next, we show that the new system (\ref{eq:odekn1}) can be chosen to be integrable locally by choosing $A_x$ and $A_y$ appropriately. Consider a point $\vct{k}$ in a simply-connected subset $U\subset D^*$. Theorem\,\ref{thm:pfaf} states that the Pfaffian system (\ref{eq:odekn1}) is integrable (so it defines a function $\tilde{\vct{u}}$ defined on $U$) if and only if the mixed derivatives are commutative:
\begin{align}
	\frac{\partial}{\partial {k_y}}\frac{\partial}{\partial {k_x}} \tilde{\vct{u}}(\vct{k}) = 	\frac{\partial}{\partial {k_x}}\frac{\partial}{\partial {k_y}}  \tilde{\vct{u}}(\vct{k})\,, \quad \vct{k}\in U\,.
	\label{eq:mix}
\end{align}
By differentiating the two equations in (\ref{eq:odekn1}) followed by simple manipulations, the condition (\ref{eq:mix}) is equivalent to 
\begin{align}
\bb{\frac{\partial A_{y}}{\partial {k_x}} - \frac{\partial A_{x}}{\partial {k_y}}}\tilde{\vct{u}} = \I \sb{\frac{\partial P}{\partial {k_x}},\frac{\partial P}{\partial {k_y}}} \tilde{\vct{u}}\,,
	\label{eq:comp}
\end{align}
where $[\frac{\partial P}{\partial {k_x}},\frac{\partial P}{\partial {k_y}}]$ is the commutator between $\frac{\partial P}{\partial {k_x}}$ and $\frac{\partial P}{\partial {k_y}}$ given by the formula
\begin{align}
	\sb{\frac{\partial P}{\partial {k_x}},\frac{\partial P}{\partial {k_y}}} = \frac{\partial P}{\partial {k_x}}\frac{\partial P}{\partial {k_y}} - \frac{\partial P}{\partial {k_y}}\frac{\partial P}{\partial {k_x}}\,.
\end{align}
We observe that $ \frac{\partial P}{\partial {k_x}}\frac{\partial P}{\partial {k_y}}\tilde{\vct{u}}$ is in the span of $\tilde{\vct{u}}$ since
\begin{align}
	P \frac{\partial P}{\partial {k_x}}\frac{\partial P}{\partial {k_y}}\tilde{\vct{u}} = (-\frac{\partial P}{\partial {k_x}} P + \frac{\partial P}{\partial_{k_x}}) \frac{\partial P}{\partial {k_y}}\tilde{\vct{u}} = \frac{\partial P}{\partial {k_x}}\frac{\partial P}{\partial {k_y}}\tilde{\vct{u}}\,,
\end{align}
where the first equality is obtained by applying (\ref{eq:piden1})
and the second equality comes from (\ref{eq:piden2}).
Similar argument can be applied to show that $ \frac{\partial P}{\partial_{k_y}}\frac{\partial P}{\partial_{k_x}}\tilde{\vct{u}}$ is also in the span of $\tilde{\vct{u}}$. As a result, by projecting the condition (\ref{eq:comp}) onto $\tilde{\vct{u}}$, we obtain the following  condition equivalent to (\ref{eq:comp}) given by the formula
\begin{align}
	\frac{\partial A_{y}}{\partial {k_x}} - \frac{\partial A_{x}}{\partial {k_y}}= \I \tilde{\vct{u}}^*\sb{\frac{\partial P}{\partial {k_x}},\frac{\partial P}{\partial {k_y}}} \tilde{\vct{u}}\,.
\end{align}
By the cyclic property of the matrix trace and $P = \tilde{\vct{u}}\tilde{\vct{u}}^*$, we obtain
\begin{align}
	\frac{\partial A_{y}}{\partial {k_x}} - \frac{\partial A_{x}}{\partial {k_y}} = \I \Tr\bb{P{\sb{\frac{\partial P}{\partial {k_x}},\frac{\partial P}{\partial_{k_y}}}}}\,.
		\label{eq:comp2}
\end{align}
We observe that the right-hand side of (\ref{eq:comp2}) is a real quantity. It is often referred to as the Berry curvature of the band\,\cite{brouder2007exponential}. We denote the Berry curvature by $\Omega_{x y}$, so we have
\begin{align}
	\Omega_{x y} = \I \Tr\bb{P{\sb{\frac{\partial P}{\partial {k_x}},\frac{\partial P}{\partial {k_y}}}}}\,.
			\label{eq:bcurv}	
\end{align}
The quantity on the left of (\ref{eq:comp2}) is the typical definition of the Berry curvature in the physics literature. As long as the Berry connection $(A_x,A_y)$ is chosen to satisfy the condition (\ref{eq:comp2}) in $U$, the Pfaffian system (\ref{eq:odekn1}) gives a well-defined assignment of $\tilde{\vct{u}}$  in $U$. 
We summarize this observation in the following.

\begin{theorem}
	Suppose $U$ is a simply connected subset of $\R^2$ and the system (\ref{eq:odekn1}) satisfies an initial condition of the form (\ref{eq:icondn}) at some $\vct{k}_0\in U$. Suppose further that the components of the Berry connection $\vct{A}$ in (\ref{eq:odekn1}) are chosen to satisfy the integrability condition (\ref{eq:comp2}) in $U$. Then the Pfaffian system  (\ref{eq:odekn1}) defines a unique solution $ \tilde{\vct{u}}$ on $U$ by Theorem\,\ref{thm:pfaf}.
	\label{thm:intecond}
\end{theorem}

	It should be observed that Theorem\,\ref{thm:intecond} only ensures local integrability in a simply connected region in $D^*$. As we will see in Section\,\ref{sec:method1}, extending it to a global one in $D^*$ may encounter topological obstructions since $D^*$ as a torus is not simply connected. Moreover, even when there is no such obstruction, the condition (\ref{eq:comp2}) is not sufficient to ensure a $\Lambda^*$-periodic assignment of $\tilde{\vct{u}}$ in $D^*$ since (\ref{eq:comp2}) only specifies the divergence-free part  (see Theorem\,\ref{thm:vecdecom}) of the Berry connection.

\subsection{Gauge transformations \label{sec:gt}}

Suppose that $\tilde{\vct{u}}$ satisfies (\ref{eq:odekn1}) defined by some Berry connection $\vct{A}$ satisfying the condition (\ref{eq:comp2}) in a simply-connected domain $U\subseteq D^*$. Given some initial condition, by Theorem\,\ref{thm:intecond}, we have a well-defined assignment of $\tilde{\vct{u}}$ in $U$. Thus, we can multiply $\tilde{\vct{u}}$ by a phase factor $e^{-\I\varphi}$ with a continuously differentiable function $\varphi: U\rightarrow \R$. We define the gauge transformation of $\dbtilde{\vct{u}}$ by the formula
\begin{align}
	\dbtilde{\vct{u}}(\vct{k}) = e^{-\I\varphi(\vct{k})}\tilde{\vct{u}} (\vct{k})\,, \quad \vct{k}\in U.
	\label{eq:gauget}
\end{align}
We observe that the new $\dbtilde{\vct{u}}$ satisfies the same system of differential equations as (\ref{eq:odekn1}) but with a new Berry connection $\tilde{\vct{A}}=(\tilde{A}_x,\tilde{A}_y)$:
\begin{align}
	\frac{ \partial\dbtilde{\vct{u}}}{\partial {k_x}}= \frac{\partial P}{\partial {k_x} } \dbtilde{\vct{u}}  - \I \tilde{A}_x \dbtilde{\vct{u}}\,,\quad 
	\frac{ \partial\dbtilde{\vct{u}}}{\partial {k_y}}= \frac{\partial P}{\partial {k_y} } \dbtilde{\vct{u}}  - \I \tilde{A}_y \dbtilde{\vct{u}}\,.
	\label{eq:odekn1new}
\end{align}
where the components of $\tilde{\vct{A}}$ are given by the formulas
\begin{align}
	\tilde{A}_x = A_x + \frac{\partial \varphi}{\partial k_x}\,,\quad \tilde{A}_y = A_y + \frac{\partial \varphi}{\partial k_y}\,.
	\label{eq:bcontran}
\end{align}
We observe that gauge transformations do not affect the divergence-free part of $\vct{A}$. Namely, we have
\begin{align}
	\frac{\partial \tilde{A}_{y}}{\partial k_x} - \frac{\partial \tilde{A}_{x}}{\partial k_y} = \frac{\partial A_{y}}{\partial k_x} - \frac{\partial A_{x}}{\partial k_y} = \Omega_{xy}\,,
	\label{eq:gind}
\end{align}
where the Berry curvature $\Omega_{xy}$ is defined in (\ref{eq:bcurv}).

\begin{remark}
	The Berry curvature $\Omega_{x y}$ given by (\ref{eq:bcurv}) only depends on the projector, thus is independent of the gauge choice of the eigenvector $\tilde{\vct{u}}$ in (\ref{eq:gauget}). As a result, in the physics literature, the Berry curvature is often computed via the following formula\,\cite{simon1983holonomy}
	\begin{align}
		\Omega_{xy} = \frac{\partial A_{y}}{\partial {k_x}} - \frac{\partial A_{x}}{\partial {k_y}} = \I \frac{\partial \vct{u}}{\partial {k_x}}^* \frac{\partial \vct{u}}{\partial {k_y}} - \I \frac{\partial \vct{u}}{\partial {k_y}}^* \frac{\partial \vct{u}}{\partial {k_x}}\,,
	\end{align}
where the derivatives $\frac{\partial \vct{u}}{\partial_{k_x}}$ and $\frac{\partial \vct{u}}{\partial_{k_y}}$ are given in (\ref{eq:odek1}). Then the perturbation formula (\ref{eq:evecpertsvd}) can be applied to compute (\ref{eq:odek1}) via the formulas
\begin{align}
			\frac{\partial \vct{u}}{\partial k_x} = - (H - E)^\dagger \frac{\partial H}{\partial k_x} \vct{u}\,, \quad 			\frac{\partial \vct{u}}{\partial k_y} = - (H - E)^\dagger \frac{\partial H}{\partial k_y} \vct{u}\,.
\end{align}
\label{rmk:bc}
\end{remark}

\subsection{First Chern number\label{sec:chern}}
In this section, we introduce the first Chern number associated with a non-degenerate band and various formulas for computing it. The first Chern number is the only obstruction to the construction of exponentially localized Wannier functions\,\cite{brouder2007exponential}.

Consider the eigenvalue problem in (\ref{eq:heig3}). Analogous to the Gauss-Bonnet formula in classical differential geometry,
the first Chern number of a band is defined by the integral of the Berry curvature (see (\ref{eq:bcurv})) over $D^*$\,\cite{simon1983holonomy,nakahara2018geometry}
\begin{equation}
	C_1 = \frac{1}{2\pi}\int_{D^*}\D \vct{k}\, \Omega_{x y}(\vct{k}) = \frac{1}{2\pi}\int_{D^*}\D k_x \D k_y\, \Omega_{x y}(\vct{k})\,,
	\label{eq:chern}
\end{equation}
and $C_1$ only takes integer values, i.e. 
$	C_1 \in \Z\,$.
Since $\Omega_{xy}$ is independent of the gauge choice (see (\ref{eq:gind})), so is the first Chern number $C_1$. Thus $C_1$ is an intrinsic property of the matrix $H$.

The following theorem states that the first Chern number being nonzero is the topological obstruction to exponentially localized Wannier functions. It is first observed in \cite{thouless1984wannier} and later formalized in \cite{panati2007triviality, brouder2007exponential}. This fact will emerge in the construction of Wannier functions in Section\,\ref{sec:method1}.
\begin{theorem}
There exists an analytic and $\Lambda^*$-periodic assignment of $\vct{u}$ on $D^*$ (so that the corresponding Wannier function is exponentially localized) if and only if the first Chern number is zero, i.e.
$
	C_1 = 0\,.
$
\label{thm:chern}
\end{theorem}

The integral in the definition (\ref{eq:chern}) is parameterized by variables $k_x,k_y$. 
When the $\kappa_1, \kappa_2$ parameterization in (\ref{eq:xydir}) is used, the formula in (\ref{eq:comp2}) is modified accordingly as 
\begin{align}
	\frac{\partial A_{2}}{\partial {\kappa_1}} - \frac{\partial A_{1}}{\partial {\kappa_2}} = \Omega_{12}\,,
		\label{eq:comp3}
\end{align}
where the Berry curvature $\Omega_{12}$ in the $\kappa_1, \kappa_2$ parameterization is given by the formula
\begin{align}
	\Omega_{12} = \I \Tr\bb{P{\sb{\frac{\partial P}{\partial {\kappa_1}},\frac{\partial P}{\partial {\kappa_2}}}}}\,.
\end{align}
By applying the change of variables formulas (\ref{eq:diffc}-\ref{eq:intvarc}), we have the following identity for the two parameterizations: 
\begin{align}
	C_1 = \frac{1}{2\pi}\int_{D^*}\D k_x \D k_y\, \Omega_{x y}(\vct{k}) = \frac{1}{2\pi}\int_{T}\D \kappa_1 \D \kappa_2\, \Omega_{1 2}(\vct{k}(\kappa_1,\kappa_2))\,.
	\label{eq:chernt}
\end{align}
By integrating (\ref{eq:comp3}) over $T$ and applying (\ref{eq:chernt}), we obtain the following formula
\begin{align}
	\int_{T}\D\kappa_1 \D\kappa_2 \bb{\frac{\partial A_{2}}{\partial {\kappa_1}} - \frac{\partial A_{1}}{\partial {\kappa_2}}} = 2\pi C_1\,.
	\label{eq:tint}
\end{align}
Applying Green's theorem to the left-hand side gives the following formula for computing $C_1$ by a line integral over the boundary of $T$:
\begin{align}
\begin{split}
	2\pi C_1 = & \int_{-\frac{1}{2}}^{\frac{1}{2}} \D \kappa_1\, (A_1(\vct{k}(\kappa_1,-1/2)) - A_1(\vct{k}(\kappa_1,1/2)))\\ 	
	&   + \int_{-\frac{1}{2}}^{\frac{1}{2}} \D \kappa_2\, (A_2(\vct{k}(1/2,\kappa_2)) - A_2(\vct{k}(-1/2,\kappa_2)))\,.	 
\end{split}
	\label{eq:lineint}
\end{align}
It should be observed that (\ref{eq:lineint}) does not require $A_1, A_2$ to be periodic on $T$; only continuity in their derivatives is needed in order to apply the Green's theorem to (\ref{eq:tint}). In fact, periodicity of $A_1, A_2$ will imply $C_1 = 0$ automatically by (\ref{eq:lineint}).

\begin{remark}
When $H$ has time-reversal symmetry (see (\ref{eq:trs})), by the realty of the Berry curvature $\Omega_{xy}$ and the symmetry of the projector in (\ref{eq:ptrs}) 
, we have
\begin{align}
	\Omega_{xy}(-\vct{k}) = - \Omega_{xy}(\vct{k})\,, \quad \vct{k}\in D^*\,.
		\label{eq:bcurvsym}
\end{align}	
In the $\kappa_1, \kappa_2$ parameterization, this symmetry becomes
\begin{align}
	\Omega_{12}(\vct{k}(-\kappa_1,-\kappa_2)) = - \Omega_{12}(\vct{k}(\kappa_1,\kappa_2))\,, \quad (\kappa_1,\kappa_2)\in T\,.
	\label{eq:bcurv12sym}
\end{align}	
Thus $C_1$ is automatically zero by (\ref{eq:chern}). As a result, time-reversal symmetry ensures that no topological obstruction will be encountered.
\label{rmk:trsp}
\end{remark}

\subsection{Gauge choice and Wannier localization\label{sec:opt} }
In this section, we introduce formulas for the moment functions (\ref{eq:Ru}) and (\ref{eq:R2u}) in terms of the Berry connection $\vct{A}$ introduced in Section\,\ref{sec:bc} and the optimal choice of the Berry connection in terms of the variance of the Wannier functions. We show that optimal Wannier functions correspond to divergence-free Berry connections and are unique up to lattice vector translations.
Results of this form are known and can be found in \,\cite{blount1962formalisms, marzari1997maximally}. 
The discussion in this section assumes that we have found a vector field $\vct{A}$ that is $\Lambda^*$-periodic and analytic on $D^*$. Furthermore, we also assume that the system (\ref{eq:odekn1})  defined by the vector field $\vct{A}$ produces a $\Lambda^*$-periodic and analytic assignment $\tilde{\vct{u}}$ on $D^*$.

\subsubsection{Variance formulas for Wannier functions}

We apply Theorem\,\ref{thm:vecdecom} to the vector field $\vct{A}$ to yield the formula
\begin{align}
	\vct{A} = (A_x, A_y) = -\bb{\frac{\partial \psi }{\partial k_x}, \frac{\partial \psi }{\partial k_y}} + \bb{\frac{\partial F }{\partial k_y}, -\frac{\partial F }{\partial k_x}} + (h_x,h_y)\,,
		\label{eq:adecom}
\end{align}
where $\psi$ and $F$ are $\Lambda^*$-periodic potentials on $D^*$ that generate the curl-free and divergence-free component respectively, and $h_x, h_y$ are constants that define the harmonic components. Obviously, by Theorem\,\ref{thm:intecond}, it is necessary that  $\vct{A}$ satisfies (\ref{eq:comp2}) in $D^*$. This is equivalent to the following:
\begin{align}
	\frac{\partial^2 F }{\partial k_x^2} + 	\frac{\partial^2 F }{\partial k_y^2} = -\Omega_{xy}\,, \quad \mbox{$F$ is $\Lambda^*$-periodic on $D^*$}\,.
		\label{eq:divless}
\end{align}
In other words, the divergence-free component of $\vct{A}$ is determined by the Berry curvature completely.


In the following lemma, we express the moment functions for the Wannier function determined by $\tilde{\vct{u}}$ in terms of quantities in (\ref{eq:adecom}). The formulas are a simple consequence of substituting  (\ref{eq:odekn1}) into (\ref{eq:Ru}) and (\ref{eq:R2u}) and using 
\begin{align}
	\tilde{\vct{u}}^* \frac{\partial P}{\partial k_x}\tilde{\vct{u}}=0\,,\quad 	\tilde{\vct{u}}^* \frac{\partial P}{\partial k_y}\tilde{\vct{u}}=0\,,
	\label{eq:para2}
\end{align}
which are consequences of (\ref{eq:para}). The details can be found in Appendix\,\ref{sec:appvar}.
\begin{lemma}
\label{lem:var}
Suppose that the Berry connection $\vct{A}$ is defined in (\ref{eq:adecom}) and its corresponding system (\ref{eq:odekn1}) defines a $\Lambda^*$-periodic and analytic assignment of $\tilde{\vct{u}}$ on $D^*$. We have the following formula for the moment functions:
\begin{align}
	\langle \vct{R} \rangle = (h_x,h_y)\,,
		\label{eq:mean2}
\end{align}
\begin{align}
	\langle \norm{\vct{R}}^2 \rangle = \frac{V_{\rm puc}}{(2\pi)^2}\int_{D^*} \D\vct{k}\,\bb{ \norm{\frac{\partial P(\vct{k})}{\partial k_x}\tilde{\vct{u}}(\vct{k})}^2 + \norm{\frac{\partial P(\vct{k})}{\partial k_y}\tilde{\vct{u}}(\vct{k})}^2  +  \norm{\vct{A}(\vct{k})}^2}\,,
	\label{eq:varr2}
\end{align}
where 
\begin{align}
\frac{V_{\rm puc}}{(2\pi)^2}\int_{D^*} \D\vct{k}\, \norm{\vct{A}(\vct{k})}^2 	&=
\frac{V_{\rm puc}}{(2\pi)^2}\int_{D^*} \D\vct{k} \bigg[  \bb{\frac{\partial \psi}{\partial k_x}}^2 + \bb{\frac{\partial \psi}{\partial k_y}}^2 \\
	& + \bb{\frac{\partial F}{\partial k_x}}^2 + \bb{\frac{\partial F}{\partial k_y}}^2 \bigg]+ h_x^2 + h_y^2  \nonumber\,.
\end{align}
Thus the variance is given by the formula
\begin{align}
	 \langle \norm{\vct{R}}^2 \rangle - \norm{\langle \vct{R} \rangle}^2  &= \frac{V_{\rm puc}}{(2\pi)^2}\int_{D^*} \D\vct{k}\,\bigg[ \norm{\frac{\partial P(\vct{k})}{\partial k_x}\tilde{\vct{u}}(\vct{k})}^2 + \norm{\frac{\partial P(\vct{k})}{\partial k_y}\tilde{\vct{u}}(\vct{k})}^2   \label{eq:varr}\\
& +\vspace{1cm} \bb{\frac{\partial \psi}{\partial k_x}}^2 + \bb{\frac{\partial \psi}{\partial k_y}}^2 + \bb{\frac{\partial F}{\partial k_x}}^2 + \bb{\frac{\partial F}{\partial k_y}}^2
\bigg] \,. \nonumber
\end{align}
\end{lemma}
\begin{remark}
We observe that Wannier center $\langle \vct{R} \rangle$ is determined by the harmonic components in (\ref{eq:adecom}), which is analogous to the Zak phase in the one-dimensional case\,\cite{zak1989berry,gopal2024high}. Thus $h_x$ and $h_y$ are proportional to the Zak phase in the $\vct{e}_x$ and $\vct{e}_y$ direction respectively.
\end{remark}
\subsubsection{Optimal Berry connections}

The derivatives of $F$ are determined by the Berry curvature completely (see (\ref{eq:divless})). To minimize the variance (\ref{eq:varr}), the only variable quantities are those related to $\psi$. We show next that a gauge transformation can be applied to make $\psi$ vanish, thus obtaining an optimal Wannier function whose variance only contains gauge-independent quantities.

To minimize the variance of the Wannier function defined by $\tilde{\vct{u}}$ in Lemma\,\ref{lem:var}, we apply the following gauge transformation 
\begin{align}
	\dbtilde{\vct{u}}(\vct{k}) = e^{-\I\varphi(\vct{k})}\tilde{\vct{u}}(\vct{k})\,,\quad \vct{k}\in D^*\,,
	\label{eq:gauget2}
\end{align}
where $\varphi: D^*\rightarrow \R$ is analytic (not necessarily $\Lambda^*$-periodic) on $D^*$. In the following, we first give the constraint on $\varphi$ such that the transformed $\tilde{\vct{u}}$ stays analytic.

According to the transformation rule (\ref{eq:bcontran}), the Berry connection $\vct{A}$ in Lemma\,\ref{lem:var} is transformed into a new $\tilde{\vct{A}}$ due to (\ref{eq:gauget2}). Since  $\vct{A}$ is $\Lambda^*$-periodic, for the new  $\tilde{\vct{A}}$ to remain $\Lambda^*$-periodic, both $\frac{\partial \varphi}{\partial k_x}$ and $\frac{\partial \varphi}{\partial k_y}$ must be $\Lambda^*$-periodic. Suppose their Fourier series are given by the formulas
\begin{align}
	\frac{\partial \varphi(\vct{k})}{\partial k_x} = \varphi_{x, \vct{0}} + \sum_{\substack{\vct{R}\in\Lambda\\ \vct{R} \ne \vct{0}}}  \varphi_{x, \vct{R}}\cdot e^{\I \vct{R}\cdot\vct{k}}\,,\quad 	\frac{\partial \varphi(\vct{k})}{\partial k_y} = \varphi_{y, \vct{0}} + \sum_{\substack{\vct{R}\in\Lambda\\ \vct{R} \ne \vct{0}}}  \varphi_{y, \vct{R}}\cdot e^{\I \vct{R}\cdot\vct{k}}\,,\quad\vct{k}\in D^*\,, \label{eq:phixyf}
\end{align}
where $\varphi_{x, \vct{R}}$ and $\varphi_{y, \vct{R}}$ are the Fourier coefficients for the derivatives and the zeroth ones are singled out. Integrating (\ref{eq:phixyf}) shows that $\varphi$ must be of the form
\begin{align}
	\varphi (\vct{k}) = \vct{c}_0\cdot\vct{k} + f(\vct{k})\,,\quad \vct{k}\in D^*\,,
	\label{eq:phi}
\end{align}
where $\vct{c}_0=(\varphi_{x,\vct{0}},\varphi_{y,\vct{0}})$ contains the zeroth Fourier coefficients in (\ref{eq:phixyf}) and $f$ is some $\Lambda^*$-periodic function. Furthermore, we have assumed that the assignment of $\tilde{\vct{u}}$ is analytic and $\Lambda^*$-periodic on $D^*$. As a result, for the new $e^{-\I\varphi}\tilde{\vct{u}}$ to be analytic and $\Lambda^*$-periodic,  we must impose the following condition for the  term linear in $\vct{k}$ in (\ref{eq:phi}):
\begin{align}
	\vct{c}_0\cdot(\vct{k}+\vct{G}) = \vct{c}_0\cdot \vct{k} + \vct{c}_0\cdot \vct{G} = \vct{c}_0\cdot \vct{k} + 2\pi l\,,\quad \vct{k}\in D^*\,,\vct{G}\in\Lambda^*\,,
	\label{eq:lattice}
\end{align}
for some integer $l$. By (\ref{eq:lambda}) and (\ref{eq:gr}), we observe that the above requirement is equivalent to that the constant vector $\vct{c}_0$ is given by some lattice point $\vct{R}_0$ in $\vct{\Lambda}$. Thus we have the following observation.
\begin{lemma}
Let the Berry connection $\vct{A}$ in the Pfaffian system (\ref{eq:odekn1}) be given by (\ref{eq:adecom}). Suppose that (\ref{eq:odekn1}) defines a $\Lambda^*$-periodic and analytic assignment of $\tilde{\vct{u}}$ on $D^*$. Then if a gauge transformation in the form of (\ref{eq:gauget}) defined by a function $\varphi:D^*\rightarrow \R$ does not change the smoothness of $\tilde{\vct{u}}$, the function $\varphi$ must be of the following form
\begin{align}
	\varphi(\vct{k}) = \vct{R}_0\cdot\vct{k} + f(\vct{k})\,,
\end{align}
where $\vct{R}_0$ is a lattice point in $\Lambda$ and $f$ is a $\Lambda^*$-periodic and analytic function on $D^*$.
\label{lem:phi}
\end{lemma}
The transformation of the Berry connection $\vct{A}$ in (\ref{eq:bcontran}) shows the gauge transformation defined by $\varphi$ in Lemma\,\ref{lem:phi} will only affect the divergence (curl-free component) and the harmonic components of $\vct{A}$. More explicitly, the transformed $\tilde{\vct{A}}$ is given by
\begin{align}
	\tilde{\vct{A}} = (\tilde{A}_x, \tilde{A}_y) = \bb{-\frac{\partial \psi }{\partial k_x} + \frac{\partial f }{\partial k_x}, -\frac{\partial \psi }{\partial k_y} + \frac{\partial f }{\partial k_y}} + \bb{\frac{\partial F }{\partial k_y}, -\frac{\partial F }{\partial k_x}} + (h_x,h_y)  + \vct{R}_0 \,.
		\label{eq:adecom2}
\end{align}
By applying Lemma\,\ref{lem:var} to the new $\tilde{\vct{A}}$, we observe that the choice of $\vct{R}_0$ only shifts the Wannier center by a lattice vector without affecting the variance. Moreover, the optimal choice for minimum variance is to choose $f=\psi$ so that $\tilde{\vct{A}}$ becomes divergence-free
\begin{align}
	\tilde{\vct{A}} = (\tilde{A}_x, \tilde{A}_y) =  \bb{\frac{\partial F }{\partial k_y}, -\frac{\partial F }{\partial k_x}} + (h_x,h_y)  + \vct{R}_0 \,.
		\label{eq:adecom3}
\end{align}
As discussed below Theorem\,\ref{thm:vecdecom}, $\psi$ can be obtained by solving
\begin{align}
	\frac{\partial^2 \psi }{\partial k_x^2} + 	\frac{\partial^2 \psi }{\partial k_y^2} = -\bb{\frac{\partial A_x}{\partial k_x} + \frac{\partial A_y}{\partial k_y} }\,, \quad \mbox{$g$ is $\Lambda^*$-periodic on $D^*$}\,,
		\label{eq:divless2}
\end{align}
where the right-hand side of (\ref{eq:divless2}) is the negative of the divergence of $\vct{A}$. Solving (\ref{eq:divless2}) can be viewed as computing the Newton step for minimizing the quadratic objective defined by (\ref{eq:varr}) with the Laplacian being the Hessian; the Laplacian is diagonal in the Fourier series representation and can be inverted with little cost via the fast Fourier transform (see Section\,\ref{sec:pois} and \ref{sec:dft}).

Furthermore, we observe that, if the divergence-free Berry connection $\tilde{\vct{A}}$ in (\ref{eq:adecom3}) exists, it will be unique up to a lattice vector in $\Lambda$. In other words, if two divergence-free  $\vct{A}_1$ and $\vct{A}_2$ that both produce some $\Lambda^*$-periodic and analytic assignment on $D^*$, we must have $\vct{A}_1 - \vct{A}_2=\vct{d}$ for some $\vct{d}\in\Lambda$. The reason is as follows. Since both $\vct{A}_1$ and $\vct{A}_2$ are divergence-free, only their harmonic components can differ so their difference is a constant vector $\vct{d}$. Thus $\vct{A}_1$ and $\vct{A}_2$ can be converted to each other by a gauge transformation defined by $e^{-\I\vct{d}\cdot\vct{k}}$. By Lemma\,\ref{lem:phi}, the vector $\vct{d}$ must be a lattice vector in $\Lambda$.

Due to the uniqueness result above, we simply choose $\vct{R}_0=\vct{0}$ in (\ref{eq:adecom3}). Moreover, since $\vct{R}_0$ in (\ref{eq:adecom3}) only shift the Wannier center by $\vct{R}_0$ by Lemma\,\ref{lem:var} and all Wannier functions centered at difference lattice points are copies of each other (see (\ref{eq:copy})), such a choice does not affect the physical consequences of Wannier functions. We summarize the above optimal conditions and their corresponding formulas for the Wannier center and the variance in the following theorem.
\begin{theorem}
	Let $\vct{A}$ be the Berry connection defined in (\ref{eq:adecom}) and its corresponding Pfaffian system (\ref{eq:odekn1}) defines a $\Lambda^*$-periodic and analytic assignment of $\tilde{\vct{u}}$ on $D^*$. Suppose that $\psi$ is given by (\ref{eq:divless2}). 
	Then the new $\dbtilde{\vct{u}}$ given by the following gauge transform 
		\begin{align}
			\dbtilde{\vct{u}}(\vct{k}) = e^{-\I \psi(\vct{k})}\tilde{\vct{u}}(\vct{k})\,, \quad \vct{k}\in D^*\,,
	\end{align}
is the optimal assignment for minimizing the variance (\ref{eq:varr}). The optimal Berry connection corresponding to $\dbtilde{\vct{u}}$ is given by the formula
\begin{align}
	\tilde{\vct{A}} = (\tilde{A}_x, \tilde{A}_y) =  \bb{\frac{\partial F }{\partial k_y}, -\frac{\partial F }{\partial k_x}} + (h_x,h_y) \,.
		\label{eq:adecom4}
\end{align}
The Wannier center and the variance of the optimal Wannier function corresponding to $\dbtilde{\vct{u}}$ are given by
    \begin{align}
    	\langle \vct{R} \rangle = (h_x,h_y)\,,
    	\label{eq:center}
    \end{align}
    \begin{align}
	& \langle \norm{\vct{R}}^2 \rangle - \norm{\langle \vct{R} \rangle}^2 \nonumber\\ 
	  &= \frac{V_{\rm puc}}{(2\pi)^2}\int_{D^*} \D\vct{k}\,\bigg[ \norm{\frac{\partial P(\vct{k})}{\partial k_x}\dbtilde{\vct{u}}(\vct{k})}^2 + \norm{\frac{\partial P(\vct{k})}{\partial k_y}\dbtilde{\vct{u}}(\vct{k})}^2  + \bb{\frac{\partial F}{\partial k_x}}^2 + \bb{\frac{\partial F}{\partial k_y}}^2
\bigg] \,.
\label{eq:var2}
    \end{align}
    \label{thm:opt}
Furthermore, such a divergence-free Berry connection $\tilde{\vct{A}}$ is unique up to a lattice vector in $\Lambda$. 
\end{theorem}
We emphasize that Theorem\,\ref{thm:opt} assumes the existence of analytic assignment of $\tilde{\vct{u}}$. Such an assignment will be explicitly constructed in Section\,\ref{sec:method1}. The uniqueness of the divergence-free Berry connection implies the optimally localized Wannier function is unique up to a lattice translation and an apparent constant phase factor. The constant phase will also be fixed to ensure the realty of the Wannier function in the construction in Section\,\ref{sec:method1} in cases when $H$ has time-reversal symmetry (see (\ref{eq:trs})).
\begin{remark}
	As observed in \cite{blount1962formalisms}, the optimal condition that the Berry connection is divergence-free (or satisfies the Coulomb gauge) can also be derived by applying the calculus of variations to (\ref{eq:varr2}). Since $F$-related terms are determined by the Berry curvature and $h_x, h_y$ are fixed up to a lattice vector (see Lemma\,\ref{lem:phi}), $\psi$ is the only variable quantity. Varying $\psi$ gives rise to the condition that $\psi$ satisfies the Laplace equation:
	\begin{align}
			\frac{\partial^2 \psi }{\partial k_x^2} + 	\frac{\partial^2 \psi }{\partial k_y^2} = 0\,, \quad \mbox{$\psi$ is $\Lambda^*$-periodic on $D^*$}\,.
	\end{align}
Obviously, the solution $\psi$ is a constant, which also leads to (\ref{eq:var2}).
\end{remark}

\begin{remark}
	The vector $\dbtilde{\vct{u}}$ in (\ref{eq:var2}) and $\tilde{\vct{u}}$ in (\ref{eq:varr}) can be replaced by a $\vct{u}$ of any phase choice since the quantity is independent of the phase of the vector. 
\end{remark}

\section{Construction of optimal Wannier functions \label{sec:method1}}
In this section, we describe a scheme for constructing globally optimal Wannier functions in the sense of Theorem\,\ref{thm:opt} in terms of the variance defined in (\ref{eq:var}).
We assume we are given a family of $n$ by $n$ matrix $H$ (a tight-banding Hamiltonian introduced in Section\,\ref{sec:tb}) that is analytic and $\Lambda^*$-periodic on $D^*$ defined by two primitive reciprocal lattice vectors $\vct{b}_1$ and $\vct{b}_2$ as in (\ref{eq:rgortho}). (We also implicitly assume the continuation condition (\ref{eq:katoeig}) is satisfied for $H$.) The corresponding real space primitive lattice vectors $\vct{a}_1$ and $\vct{a}_2$ define a lattice $\Lambda$ as in (\ref{eq:lambda}) and satisfy (\ref{eq:ab}). Moreover, it is assumed that we have picked an eigenvalue $E$ and eigenvector $\vct{u}$ of interests that define an eigenvalue equation as in (\ref{eq:heig3}). 
%

The approach in this section is purely based on the parallel transport equation in the form of (\ref{eq:oden}) along different lines in $D^*$\,. 
Along each line, we assign eigenvectors $\vct{u}$ according to the approach in  \cite{gopal2024high}, which constructs the optimal one-dimensional Wannier function.
Doing so over a family of lines in $D^*$ defines a two-dimensional assignment of $\vct{u}$. The analyticity of such an assignment is purely a consequence of the smoothness result of ODEs introduced in Section\,\ref{sec:ode}. After this step, a single gauge transform is performed to eliminate the divergence of the Berry connection to achieve the optimality in Theorem\,\ref{thm:opt}. This approach can be viewed as a direct extension of the method for constructing optimal single band Wannier function in one dimension in \cite{gopal2024high}. 
The topological construction -- the first Chern number $C_1 \ne 0$ introduced in Section\,\ref{sec:chern} -- emerges automatically in the parallel transport stage, causing the construction to fail as expected. When a topological obstruction is present, this approach still produces an analytic assignment of $\vct{u}$ on $D^*$ viewed as a subset of $\R^2$. (The assignment is discontinuous when $D^*$ is viewed as a torus.)



%



The method can be divided into three stages. In Stage 1, the parallel transport equation is solved on $\gamma_0$ in Figure\,\ref{fig:stages}. A simple correction is introduced (as in the one-dimensional case\,\cite{gopal2024high}) that yields an analytic and periodic assignment on $\gamma_0$. In Stage 2, the result in Stage 1 serves as the initial conditions for the parallel transport equation on the path $\gamma_{\kappa_1}$ for $\kappa_1\in\sb{-\half,\half}$ in Figure\,\ref{fig:stages}, followed by applying similar corrections as in Stage 1.  We show that the assignment after Stage 2 is analytic and $\Lambda^*$-periodic on $D^*$ when the topological obstruction is not present (i.e. $C_1 = 0$). Moreover, when the matrix $H$ has time-reversal symmetry, we show that the assignment results in a real Wannier function automatically provided that the initial condition in Stage 1 is chosen to be real. In Stage 3, the divergence of the Berry connection of the assignment is computed and eliminated by a single gauge transformation to yield the optimal assignment. The new Wannier function remains real if the result after Stage 2 is real.

\begin{figure}[h]
	\centering	
	\includegraphics[scale=0.50]{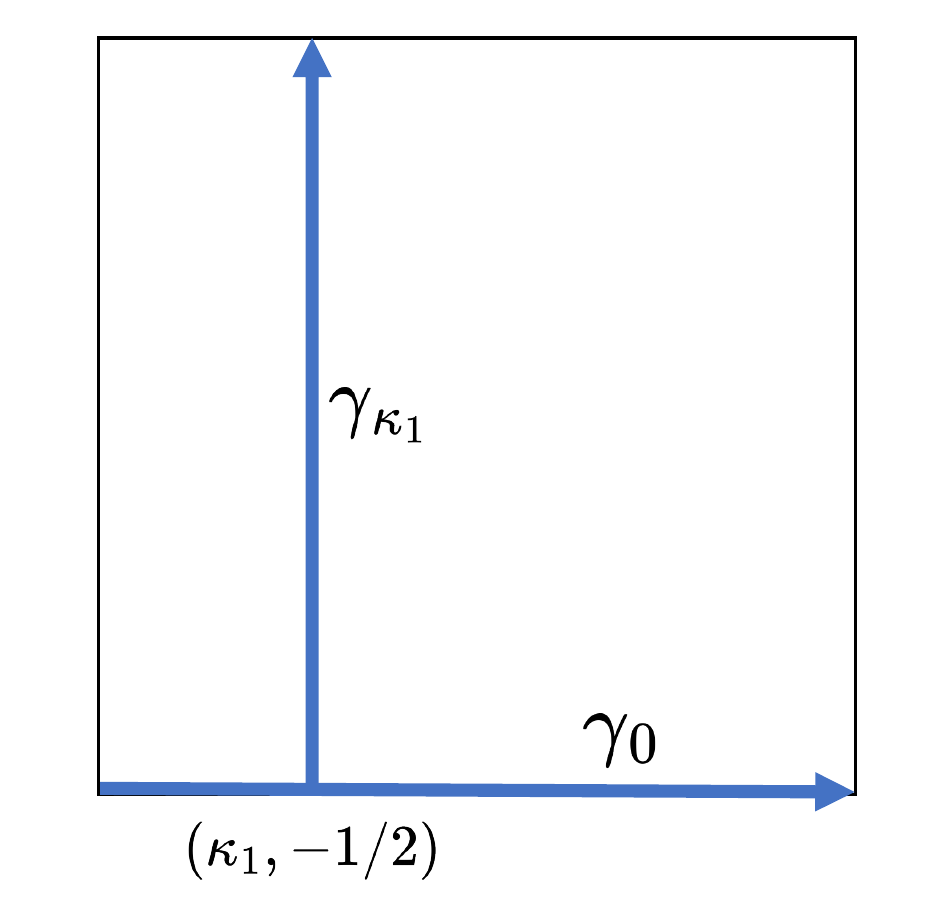}	
	\vspace{-1em}
	\caption{ The path $\gamma_{0}$ in Stage 1 and $\gamma_{\kappa_1}$ for $\kappa_1\in \sb{-\half,\half}$ in Stage 2 shown in $T$.
	\label{fig:stages}}	
\end{figure}

When describing the construction in the rest of this section, we parameterize $D^*$ in (\ref{eq:Ds}) by $T$ in (\ref{eq:torus}) via (\ref{eq:xydir}).
Hence, for any $\vct{k}\in D^*$, the eigenvector $\vct{u}$ is parameterized as 
\begin{align}
	\vct{u}(\vct{k}) = 	\vct{u}(\vct{k}(\kappa_1, \kappa_2))\,,\quad (\kappa_1, \kappa_2) \in T\,.
\end{align}
The matrix $H$ is also parameterized similarly by
\begin{align}
	H(\vct{k}) = 	H(\vct{k}(\kappa_1, \kappa_2))\,,\quad (\kappa_1, \kappa_2) \in T\,.
\end{align}
Formulas in Section\,\ref{sec:coord} can be applied to convert various expressions of $\vct{k}$ into those of $\kappa_1,\kappa_2$ and vice versa.

\begin{remark}
	\label{rmk:alt}
	Although the approach in this section is a direct extension of the method for one-dimensional Wannier problems, it is arguably not the most natural one from a physics point of view. In Appendix\,\ref{sec:method2}, we provide an alternative construction, where the gauge-invariant vector fields in (\ref{eq:adecom4}) are directly computed without creating the assignment of $\vct{u}$ first by exploiting the gauge-invariant nature of the Berry curvature in (\ref{eq:bcurv}) (see Remark\,\ref{rmk:bc}). Since the invariance no longer holds when multiple eigenvalues are considered, such an approach is not generalizable to the multiband case. However, it highlights the roles of the gauge-invariant vector fields in constructing Wannier functions, thus it is included as a complementary viewpoint.
\end{remark}

\subsection{Stage 1: constructing an assignment on a line \label{sec:stage1}} 
In the first stage, we start from the lower-left corner of $D^*$ (see Figure\,\ref{fig:stages}) given by
\begin{align}
\vct{k}^0= -\frac{1}{2}\vct{b}_1 -\frac{1}{2} \vct{b_2}\,,
\end{align}
and construct an assignment of the eigenvectors on the line $\gamma_0: \sb{-\half,\half}\rightarrow D^*$ defined by
\begin{align}
	\gamma_0(\kappa_1) = \vct{k}(\kappa_1,-1/2)= \kappa_1\vct{b}_1 -\frac{1}{2} \vct{b_2}\,,
	\label{eq:gam0}
\end{align}
where the starting point $	\gamma_0(-1/2) = \vct{k}^0$\,.

First, we compute the eigenvalue $E^0$ and eigenvector $\vct{v}^0$ at $\vct{k}^0$ given by the formula
\begin{align}
	H\bb{\vct{k}^0}\vct{v}^0 = E^0  \vct{v}^0 \,.
	\label{eq:init0}
\end{align}
Next, we assign $\vct{u}\bb{\gamma_0(\kappa_1)} =\vct{u}\bb{\vct{k}(\kappa_1, -1/2)}$ for $\kappa_1\in \sb{-\half,\half}$ by solving the parallel transport equation (see (\ref{eq:oden})) given by the formula
\begin{align}
	\frac{\partial \vct{u}\bb{\vct{k}(\kappa_1, -1/2)}}{\partial \kappa_1}   = 	\frac{\partial P\bb{\vct{k}(\kappa_1, -1/2)}}{\partial \kappa_1}   \vct{u} \bb{\vct{k}(\kappa_1, -1/2)}\,,\quad \kappa_1\in \sb{-\frac{1}{2},\frac{1}{2}}\,,
		\label{eq:dvk11}
\end{align}
subject to the initial condition
\begin{align}
	\vct{u}\bb{\vct{k}(-1/2, -1/2)}=	\vct{u}\bb{\vct{k}^0} = \vct{v}^0 \,.
		\label{eq:dvk10}
\end{align}
By Theorem\,\ref{thm:katoode}, the assignment $\vct{u}\bb{\vct{k}(\kappa_1, -1/2)}$ is analytic in  $\kappa_1\in\sb{-\half,\half}$. 
By periodicity, we always have
$	H\bb{\vct{k}(-1/2, -1/2)} = 	H\bb{\vct{k}(1/2, -1/2)}$\,,
but the same condition does not hold for $\vct{u}$ in general. In other words, we generally have
\begin{align}
\vct{u}\bb{\vct{k}(-1/2, -1/2)} \ne \vct{u}\bb{\vct{k}(1/2, -1/2)}. 
\end{align}
Hence, $\vct{u}$ is discontinuous as a periodic function on $\sb{-\frac{1}{2}, \frac{1}{2}}$.
Since both $\vct{u}\bb{-1/2, -1/2}$ and $\vct{u}\bb{1/2, -1/2}$ are both eigenvectors of the same (non-degenerate) eigenvalue, they can only differ by a constant phase factor, which we denote by $e^{\I \varphi_1}$ for some real $\varphi_1$:
\begin{align}
	\vct{u}\bb{\vct{k}(1/2, -1/2)} = e^{\I \varphi_1}\vct{u}\bb{\vct{k}(-1/2, -1/2)}\,.
	\label{eq:phik2}
\end{align}
Obviously, $\varphi_1$ can be computed via
\begin{align}
	\varphi_1 = -\I \log\bb{\vct{u}\bb{\vct{k}(-1/2, -1/2)}^*\vct{u}\bb{\vct{k}(1/2, -1/2)}}.
		\label{eq:logphik2}
\end{align}
Next, we apply a gauge transformation (see (\ref{eq:gauget})) to turn $\vct{u}$ into an analytic and periodic function on $\sb{-\frac{1}{2}, \frac{1}{2}}$\,.

We apply the gauge transformation given by the formula
\begin{align}
	\tilde{\vct{u}} \bb{\vct{k}(\kappa_1, -1/2)} = e^{-\I \varphi_1 (\kappa_1 + \half)}\vct{u} \bb{\vct{k}(\kappa_1, -1/2)}\,,\quad \kappa_1 \in \sb{-\frac{1}{2},\frac{1}{2}}\,,
	\label{eq:stage1gt}
\end{align}
so that we have
\begin{align}
	\tilde{\vct{u}} \bb{\vct{k}(-1/2, -1/2)} = \tilde{\vct{u}} \bb{\vct{k}(1/2,-1/2)} \,.
	\label{eq:uperk1}
\end{align}
Thus $\tilde{\vct{u}}\bb{\vct{k}(\kappa_1, -1/2)}$ is continuous and periodic on $\sb{-\frac{1}{2}, \frac{1}{2}}$.
It is obvious that the new $\tilde{\vct{u}}$ satisfies the following equation
\begin{align}
	\frac{\partial \tilde{\vct{u}}\bb{\vct{k}(\kappa_1, -1/2)}}{\partial \kappa_1}   = 	\frac{\partial P\bb{\vct{k}(\kappa_1, -1/2)}}{\partial \kappa_1}   \tilde{\vct{u}} \bb{\vct{k}(\kappa_1, -1/2)} -\I \varphi_1\tilde{\vct{u}}\bb{\vct{k}(\kappa_1, -1/2)}\,,
	\label{eq:dvk1}
\end{align}
subject to the same initial condition (\ref{eq:dvk10}).
By (\ref{eq:uperk1}) and  the periodicity of $P$ (see (\ref{eq:pper}))\,,
 the formula (\ref{eq:dvk1}) shows that the derivative satisfies
\begin{align}
	\frac{\partial}{\partial \kappa_1}  \tilde{\vct{u}}\bb{\vct{k}(-1/2, -1/2)} = \frac{\partial }{\partial \kappa_1} \tilde{\vct{u}}\bb{\vct{k}(1/2, -1/2)}\,.	
		\label{eq:ud1perk1}
\end{align}
By repetitive differentiation (\ref{eq:dvk1}) with respect to $\kappa_1$, the same argument shows that the derivative of $\tilde{\vct{u}}$ of any order $n=0,1,2,\ldots$ with respect to $\kappa_1$ satisfies
\begin{align}
	\frac{\partial^n}{\partial \kappa^n_1}  \tilde{\vct{u}}\bb{\vct{k}(-1/2, -1/2)} = 	\frac{\partial^n}{\partial \kappa^n_1}\tilde{\vct{u}}\bb{\vct{k}(1/2, -1/2)}\,.
		\label{eq:udperk1}
\end{align}
As a result, the constructed $\tilde{\vct{u}}\bb{\vct{k}(\kappa_1,-1/2)}$ defines an analytic and periodic function on $\sb{-\frac{1}{2}, \frac{1}{2}}$\,.
We summarize the result of the construction after the first stage in the following lemma.
\begin{lemma}
Suppose that $\varphi_1$ is the constant given in (\ref{eq:phik2}). Then the solution $\tilde{\vct{u}}$ to (\ref{eq:dvk1}) subject to the initial condition (\ref{eq:dvk10}) is analytic and periodic on the line $\gamma_0$ defined in (\ref{eq:gam0}).
\label{lem:stage1}
\end{lemma}
By $\frac{\partial}{\partial {\kappa_1}} P^*= \frac{\partial}{\partial {\kappa_1}} P$, we can transpose (\ref{eq:dvk11}) with its initial condition and repeat the same procedure above to show the same result for $\tilde{\vct{u}}^*$, the transpose of $\tilde{\vct{u}}$ in Lemma\,\ref{lem:stage1}.
\begin{corollary}
$\tilde{\vct{u}}^*$ is analytic and periodic on the line $\gamma_0$.
\end{corollary}

\subsection{Stage 2: constructing an assignment on the torus \label{sec:stage2}}
In the second stage, for  each $\kappa_1\in \sb{-\half,\half}$, we assign the eigenvectors along a family of lines $\gamma_{\kappa_1}:\sb{-\half,\half}\rightarrow D^*$ in Figure\,\ref{fig:stages} defined by
\begin{align}
	\gamma_{\kappa_1}(\kappa_2) = \kappa_1\vct{b}_1 + \kappa_2 \vct{b_2}\,, 
	\label{eq:gamk1}
\end{align}
where the starting points $\gamma_{\kappa_1}(-1/2)=\kappa_1\vct{b}_1  -\half\vct{b_2}$ for $\kappa_1 \in \sb{-\half,\half}$ are on the line $\gamma_0$ in (\ref{eq:gam0}).

Similar to the first stage, for  each $\kappa_1\in \sb{-\half,\half}$, the assignment is done by parallel transporting the eigenvector according to the equation
\begin{align}
	\frac{\partial \vct{u}\bb{\vct{k}(\kappa_1, \kappa_2)}}{\partial \kappa_2}   = 	\frac{\partial P\bb{\vct{k}(\kappa_1, \kappa_2)}}{\partial \kappa_2}   \vct{u} \bb{\vct{k}(\kappa_1, \kappa_2)}\,,\quad \kappa_2\in \sb{-\frac{1}{2},\frac{1}{2}}\,,
			\label{eq:dvk2}
\end{align}
with the initial condition
\begin{align}
	\vct{u}\bb{\vct{k}(\kappa_1, -1/2)}=	\tilde{\vct{u}}\bb{\vct{k}(\kappa_1,-1/2)}\,,
		\label{eq:dvk20}
\end{align}
where $\tilde{\vct{u}}$ is obtained in the first stage in Lemma\,\ref{lem:stage1}. By construction, the initial condition $\tilde{\vct{u}}$ and its transpose $\tilde{\vct{u}}^*$ are analytic and periodic on $\sb{-\half,\half}$\,. By the periodicity of $\frac{\partial}{\partial {\kappa_2}} P$ and $\frac{\partial}{\partial {\kappa_2}} P = \frac{\partial}{\partial {\kappa_2}} P^*$, we apply Lemma\,\ref{lem:ode} to (\ref{eq:dvk2}) with the initial conditions (\ref{eq:dvk20}) to conclude the following lemma.

\begin{lemma}
The solution $\vct{u}\bb{\vct{k}(\kappa_1,\kappa_2)}$ to (\ref{eq:dvk2}) with the initial condition (\ref{eq:dvk20}) and its transpose $\vct{u}^*\bb{\vct{k}(\kappa_1,\kappa_2)}$ are analytic in $\kappa_1, \kappa_2 \in\sb{-\half,\half}$. Furthermore,  they are periodic in $\kappa_1$ for any $\kappa_2\in\sb{-\half,\half}$:
\begin{align}
	\begin{split}
	&\vct{u}\bb{\vct{k}(\kappa_1, \kappa_2)} = \vct{u}\bb{\vct{k}(\kappa_1+n, \kappa_2 )}\,, \quad n\in\Z\,,\\
	&\vct{u}^*\bb{\vct{k}(\kappa_1, \kappa_2)} = \vct{u}^*\bb{\vct{k}(\kappa_1+n, \kappa_2 )}\,, \quad n\in\Z\,.
		\end{split}
\end{align}

\label{lem:stage21}
\end{lemma}

Lemma\,\ref{lem:stage21} shows that the only task remains is to make $\vct{u}\bb{\vct{k}(\kappa_1, \kappa_2)}$ also periodic in $\kappa_2$ so that $\vct{u}\bb{\vct{k}(\kappa_1, \kappa_2)}$ is analytic and periodic in $(\kappa_1,\kappa_2)\in T$. To do so, we repeat the same procedure in the first stage along each curve $\gamma_{\kappa_1}$ for $\kappa_1\in\sb{-\half,\half}$. First, we apply a similar argument for obtaining (\ref{eq:phik2}): since $\vct{u}\bb{\vct{k}(\kappa_1, -1/2)}$ and $\vct{u}\bb{\vct{k}(\kappa_1, -1/2)}$ are eigenvectors of the same eigenvalue due to periodicity, they only differ by a $\kappa_1$-dependent phase factor $z:\sb{-\half,\half}\rightarrow \C$, where $z(\kappa_1)$ is on the unit circle in the complex plane, i.e. $\abs{z(\kappa_1)} = 1$. More explicitly, we have
\begin{align}
	\vct{u}\bb{\vct{k}(\kappa_1, 1/2)}=	z(\kappa_1) \vct{u}\bb{\vct{k}(\kappa_1, -1/2)}\,,\quad \kappa_1 \in \sb{-\half,\half}\,,
	\label{eq:zdiff}
\end{align}
where $z(\kappa_1)$ can be computed via
\begin{align}
	z(\kappa_1) = \vct{u}^*\bb{\vct{k}(\kappa_1, -1/2)}\vct{u}\bb{\vct{k}(\kappa_1, 1/2)}\,.
	\label{eq:z}
\end{align}
By Lemma\,\ref{lem:stage21} and (\ref{eq:z}), the analyticity and periodicity of  $\vct{u}$ and $\vct{u}^*$ imply that $z$ is also analytic and periodic on $\sb{-\half,\half}$. 

Naturally, we wish to repeat the procedure in (\ref{eq:stage1gt}) to make $\vct{u}\bb{\vct{k}(\kappa_1, \kappa_2)}$ an analytic and periodic function in $\kappa_2$. This involves computing
\begin{align}
	\varphi_2(\kappa_1) = -\I\log(z(\kappa_1))\,,
	\label{eq:phitopo}
\end{align}
by which the following gauge transformation needs to be applied:
\begin{align}
	\tilde{\vct{u}}(\vct{k}(\kappa_1,\kappa_2)) = e^{-\I\varphi_2(\kappa_1)(\kappa_2+\half)}\vct{u}(\vct{k}(\kappa_1,\kappa_2))\,,\quad \kappa_1,\kappa_2\in \sb{-\half,\half}\,.
	\label{eq:stage2gt}
\end{align}
Using similar arguments for (\ref{eq:uperk1}), (\ref{eq:ud1perk1}) and  (\ref{eq:udperk1}), we conclude that $\tilde{\vct{u}}(\vct{k}(\kappa_1,\kappa_2)) $ is analytic and periodic in $\kappa_2\in\sb{-\half,\half}$. However, since $\varphi_2$ is $\kappa_1$-dependent, the transformed vector $\tilde{\vct{u}}$ is not necessarily analytic and periodic in $\kappa_1$ unless $\varphi_2$ is. Although $z$ is analytic and periodic (see (\ref{eq:z})) if the trajectory of $z$ winds around the origin in the complex plane before returning to its starting point,
the phase $\varphi_2$ is at best discontinuous as a periodic function on $[-\half,\half]$, since the $\log$ function has no continuous branch that includes the entire unit circle.  
To understand when this happens, we next establish the relation between the first Chern number $C_1$ (see (\ref{eq:chern})) and the winding number of the path of $z$ around the origin. The relation is well-known and can be found in \cite{gresch2017z2pack}, for example.

We compute the Berry connection $A_1$ and $A_2$ on the boundary of $T$ and use (\ref{eq:lineint}) to compute $C_1$, where the following quantities are required:
\begin{align}
	& A_2(\vct{k}(-1/2,\kappa_2)) - A_2(\vct{k}(1/2,\kappa_2))\,,\quad \kappa_2 \in \sb{-\half,\half}\,,\label{eq:ba2}\\
	& A_1(\vct{k}(\kappa_1, -1/2))- A_1(\vct{k}(\kappa_1, 1/2))\,,\quad \kappa_1 \in \sb{-\half,\half}\,. \label{eq:ba1}
\end{align}
We apply (\ref{eq:phybcon})) to compute the Berry connection $A_i=\I \vct{u}^*\frac{\partial}{\partial {\kappa_i}} \vct{u}$ for $i=1,2$.
Combining it with (\ref{eq:dvk2}) and (\ref{eq:para}), we conclude that $A_2$ is identically zero on $T$, so the term in (\ref{eq:ba2}) is zero. Next, we use the phase factor $z$ in (\ref{eq:zdiff}) to express the quantity in (\ref{eq:ba1}). By differentiating (\ref{eq:zdiff}), followed by multiplying it by $\vct{u}^*\bb{\vct{k}(\kappa_1, 1/2)}$, we have
\begin{align}
	A_1\bb{\vct{k}(\kappa_1, 1/2)} &= \I \vct{u}^*\bb{\vct{k}(\kappa_1, 1/2)}\frac{\partial}{\partial {\kappa_1}}\vct{u}\bb{\vct{k}(\kappa_1, 1/2)}\\
	& = \I \frac{z'(\kappa_1)}{z(\kappa_1)} + A_1\bb{\vct{k}(\kappa_1, -1/2)}\,,\label{eq:adiff}
\end{align}
where we have used (\ref{eq:zdiff}) and $\bar{z} = z^{-1}$ to obtain (\ref{eq:adiff}). Combining (\ref{eq:adiff}) and the fact that $A_2 = 0$, we apply (\ref{eq:lineint}) to obtain the following formula
\begin{align}
	C_1 = \frac{1}{2\pi\I}\int_{-\half}^{\half}\D \kappa_1\,\frac{z'(\kappa_1)}{z(\kappa_1)} = \frac{1}{2\pi\I}\int_{\Gamma}\D z\,\frac{1}{z}\,,
	\label{eq:reschern}
\end{align}
where the second equality is obtained by a simple change of variable and $\Gamma$ is the trajectory of $z(\kappa_1)$ on the unit circle as $\kappa_1$ goes from $-\half$ to $\half$. By the residue theorem, (\ref{eq:reschern}) shows that the first Chern number $C_1$ is the winding number of $\Gamma$ around the origin as $\kappa_1$ goes from $-\half$ to $\half$. This fact is useful for two reasons. 
First, since $z$ is analytic and periodic,  (\ref{eq:reschern}) can be evaluated to obtain $C_1$ numerically in a reliable and efficient manner. Second, and more importantly, when $C_1 \ne 0$, $z$ will at least go once around the unit circle before returning to its starting point so that $\varphi_2$ in (\ref{eq:phitopo}) is discontinuous as a periodic function on $[-\half,\half]$, which is the manifestation of the topological obstruction in Theorem\,\ref{thm:chern}. Conversely, when $C_1 = 0$\,, $z$ returns to the starting point without going around the origin. This means that, by choosing suitable branches if necessary, $\log$ over $\Gamma$ can be chosen to be an analytic function, so that $\varphi_2$ in (\ref{eq:phitopo}) is analytic and periodic. Thus we have the following lemma.
\begin{lemma}
\label{lem:topo}
The function $\varphi_2$ in (\ref{eq:phitopo}) is an analytic and periodic function on $\sb{-\half,\half}$ if and only if the first Chern number $C_1$ is zero. Furthermore, $\varphi_2$ is at best a discontinuous function if $C_1$ is nonzero.
\end{lemma} 

Suppose that $C_1 = 0$. Then we apply the gauge transformation defined in (\ref{eq:stage2gt}) to obtain $\tilde{\vct{u}}$. By Lemma\,\ref{lem:topo}, we conclude that  $\tilde{\vct{u}}$ is analytic and $\Lambda^*$-periodic on $D^*$. Moreover, for each $\kappa_1\in\sb{-\half,\half}$, $\tilde{\vct{u}}$ satisfies the new differential equation, obtained from (\ref{eq:dvk2}) according to the transformation (\ref{eq:stage2gt}):
\begin{align}
	\frac{\partial \tilde{\vct{u}}\bb{\vct{k}(\kappa_1, \kappa_2)}}{\partial \kappa_2}   = 	\frac{\partial P\bb{\vct{k}(\kappa_1, \kappa_2)}}{\partial \kappa_2} \tilde{\vct{u}} \bb{\vct{k}(\kappa_1, \kappa_2)} -\I\varphi_2(\kappa_1)\tilde{\vct{u}}\bb{\vct{k}(\kappa_1, \kappa_2)}\,,
	\label{eq:dvk2new}
\end{align}
subject to the same initial condition (\ref{eq:dvk20}).
We summarize the results in Stage 2 in the following theorem, which is the main tool for computing exponentially localized Wannier functions. 
\begin{theorem}
	Suppose that $C_1 =0$. Then $\tilde{\vct{u}}$ in (\ref{eq:stage2gt}) is analytic and $\Lambda^*$-periodic on $D^*$. Thus the Wannier function corresponding to $\tilde{\vct{u}}$ is exponentially localized.
	\label{thm:stage2}
\end{theorem}

When $C_1\ne 0$, the assignment will be at best discontinuous and we terminate the construction. For the remainder of the construction, we assume that $C_1 = 0$ so that $\tilde{\vct{u}}$ is analytic and $\Lambda^*$-periodic on $D^*$.

It should be observed that the integrability condition (\ref{eq:comp2}) in Theorem\,\ref{thm:intecond} is not explicitly used in the above construction; there is no contradiction since the paths in the construction that cover $D^*$, namely $\gamma_0$ in (\ref{eq:gam0}) and $\gamma_{\kappa_1}$ in (\ref{eq:gamk1}), do not intersect at any point. 
Although  the Berry connection $A_2$ corresponding to $\tilde{\vct{u}}$ is given by the formula (due to (\ref{eq:stage2gt} and (\ref{eq:dvk2}))
\begin{align}
	A_2(\vct{k}(\kappa_1,\kappa_2)) = \varphi_2(\kappa_1)\,, \quad \kappa_1,\kappa_2\in \sb{-\half,\half},
\end{align} 
information about $A_1$ of $\tilde{\vct{u}}$ is not directly available from the construction (except for the value at $\kappa_1\in\sb{-\half,\half}$ and $\kappa_2=-1/2$ by (\ref{eq:dvk1}), which is transformed accordingly based on (\ref{eq:stage2gt})). Nonetheless, by Theorem\,\ref{thm:intecond}, $A_1$ is automatically specified to satisfy the integrability condition (\ref{eq:comp2}) over $D^*$\,. Due to the lack of control over $A_1$ in the construction, the constructed $\tilde{\vct{u}}$ is not guaranteed to be optimal as stated in Theorem\,\ref{thm:opt}. In other words, although the Wannier function corresponding to $\tilde{\vct{u}}$ is already exponentially localized, the variance of its Fourier coefficients can still be further reduced. 

Despite the fact the constructed $\tilde{\vct{u}}$ is not optimal, we will prove in Theorem\,\ref{thm:real}
in Section\,\ref{sec:real} that, if $H$ has time-reversal symmetry (see (\ref{eq:trs})),  for any $\vct{k}\in D^*$, we will have
\begin{align}
		\bar{\tilde{\vct{u}}}(\vct{k}) = \tilde{\vct{u}}(-\vct{k})\,,
			\label{eq:utrs}
\end{align}
if $\vct{v}^0$ in the initial condition (\ref{eq:dvk10}) is chosen to be real (which is possible again due to time-reversal symmetry). 
When $D^*$ is parameterized by $(\kappa_1,\kappa_2)\in T$, this is equivalent to
\begin{align}
	\bar{\tilde{\vct{u}}}(\vct{k}(\kappa_1,\kappa_2)) = \tilde{\vct{u}}(\vct{k}(-\kappa_1,-\kappa_2))\,.
	\label{eq:utrs2}
\end{align}
The condition (\ref{eq:utrs}) implies that the Fourier coefficients of $\tilde{\vct{u}}$ are real, thus the Wannier function defined by (\ref{eq:wann}) is also real.  

In the next stage, we turn $\tilde{\vct{u}}$ into the globally optimal assignment.

\begin{remark}
The assignment constructed along $\gamma_{0}$ and $\gamma_{\kappa_1}$ for each $\kappa_1$ corresponds to the so-called hybrid Wannier functions \cite{marzari2012maximally}, lower dimensional ``slices'' that make up the higher dimensional ones. According to the analysis in \cite{gopal2024high} (that can be easily adapted to the matrix models here), every assignment along those lines is globally optimal. Furthermore, the integrand  $\frac{z'(\kappa_1)}{z(\kappa_1)}$ in (\ref{eq:reschern}) corresponds to the Wannier center of the hybrid Wannier function. It is well-known that the changes in the Wannier centers contain topological information of the band  \cite{gresch2017z2pack}.
\label{rmk:hwann}
\end{remark}

\begin{remark}
	Due to the global optimality of one-dimensional assignments along $\gamma_{0}$ and $\gamma_{\kappa_1}$ for each $\kappa_1$ (see Remark\,\ref{rmk:hwann}),  it is well-known that when the Berry curvature $\Omega_{xy}$ of the band (see (\ref{eq:bcurv})) is zero, the global optimality of one-dimensional assignments results in that of the higher-dimensional assignment\,\cite{marzari1997maximally}. Thus, when  $\Omega_{xy}$ is identically zero, the steps in Section\,\ref{sec:stage3} are not needed.
\end{remark}


\subsection{\label{sec:div}Stage 3: eliminating the divergence of the Berry connection \label{sec:stage3}}
Assuming $C_1=0$, according to Theorem\,\ref{thm:stage2}, we have already obtained an analytic and $\Lambda^*$-periodic assignment $\tilde{\vct{u}}$ on $D^*$. Hence, all results in Section\,\ref{sec:opt} are applicable. By Theorem\,\ref{thm:opt}, it remains to eliminate the divergence (curl-free component) of the Berry connection corresponding to $\tilde{\vct{u}}$ to achieve the globally optimal assignment. In this section, we describe the steps for achieving this goal.

It should be observed that the sole purpose of Stage 3 is to modify $\tilde{\vct{u}}$ (by a gauge transformation) to achieve the minimum variance in terms of its Fourier coefficients. Hence, it is optional if the goal is only to construct an exponentially localized Wannier function. (The Wannier function is also real by (\ref{eq:utrs}) if the matrix $H$ has time-reversal symmetry.)

To eliminate the divergence of the Berry connection of $\tilde{\vct{u}}$, we first compute the Berry connection in the $\vct{b}_1, \vct{b}_2$ direction by the formulas
\begin{align}
	A_1 =\I \tilde{\vct{u}}^*\frac{\partial}{\partial {\kappa_1}}\tilde{\vct{u}}\,,\quad A_2 =\I \tilde{\vct{u}}^* \frac{\partial}{\partial {\kappa_2}} \tilde{\vct{u}}\,.
	\label{eq:bckform}
\end{align}
We use (\ref{eq:diffc}) to express the Berry connection in (\ref{eq:bckform}) in the $\vct{e}_x, \vct{e}_y$ basis by the formulas
\begin{align}
	A_{x} = \frac{1}{2\pi} \vct{a}_1\cdot\vct{e}_x A_1 + \frac{1}{2\pi} \vct{a}_2\cdot\vct{e}_x A_2\,,\quad  A_{y} = \frac{1}{2\pi} \vct{a}_1\cdot\vct{e}_y A_1 + \frac{1}{2\pi} \vct{a}_2\cdot\vct{e}_y A_2\,.
	\label{eq:bcxyform}
\end{align}
Next, 
we compute the divergence $\frac{\partial A_x}{\partial {k_x}} + \frac{\partial A_{y}}{\partial {k_y}}$ and denote it by $g$\,:
\begin{align}
	g = \frac{\partial A_x}{\partial {k_x}} + \frac{\partial A_{y}}{\partial {k_y}} \,.
	\label{eq:rhs}
\end{align}
By Theorem\,\ref{thm:vecdecom} and (\ref{eq:pois}), we compute the potential $\psi: D^*\rightarrow \R$ generating the curl-free component
by solving the following Poisson's equation:
\begin{align}
	\frac{\partial^2  \psi}{\partial {k^2_x}} + \frac{\partial^2 \psi}{\partial {k^2_y}}  = - g \,, \quad \mbox{$\psi$ is $\Lambda^*$-periodic on $D^*$}\,.
		\label{eq:divlessc3}
\end{align}
Since $A_1$ and $A_2$ are analytic and $\Lambda^*$-periodic on $D^*$ by construction, so are $A_{x}$ and $A_{y}$ by (\ref{eq:bcxyform}). Hence, due to (\ref{eq:kxy0}), we always have
\begin{align}
	\int_{D^*} g(\vct{k})\, \D \vct{k} = \int_{D^*} \bb{\frac{\partial A_x(\vct{k})}{\partial {k_x}} + \frac{\partial A_{y}(\vct{k})}{\partial {k_y}}}\,\D\vct{k} = 0 \,.
	\label{eq:solva}
\end{align}
By applying Theorem\,\ref{thm:poiss}, we conclude that (\ref{eq:divlessc3}) is always solvable. Moreover, the solution $\psi$ is analytic and $\Lambda^*$-periodic on $D^*$, whose Fourier series is given by the formula
\begin{align}
	\psi (\vct{k}) = \sum_{\substack{\vct{R}\in\Lambda\\ \vct{R}\ne \vct{0}}} \frac{g_{\vct{R}}}{\norm{\vct{R}}^2}\cdot e^{\I \vct{R}\cdot\vct{k}}\,, \quad \vct{k}\in D^*\,,
	\label{eq:psisol}
\end{align}
where $g_{\vct{R}}$ is the Fourier coefficient of $g$ at $\vct{R}\in \Lambda$ given by the formula
\begin{align}
	g_{\vct{R}} = \frac{V_{\rm puc}}{(2\pi)^2} \int_{D^*} \D\vct{k}\, e^{-\I\vct{R}\cdot\vct{k}} \cdot g(\vct{k})\,.
\end{align}
Finally, we apply the following gauge transformation
\begin{align}
	\dbtilde{\vct{u}}(\vct{k}) = e^{-\I\psi(\vct{k})}\tilde{\vct{u}}(\vct{k})\,,\quad \vct{k}\in D^* \,,
	\label{eq:gtstage3}
\end{align}
that eliminates the divergence of the Berry connection $(A_x, A_y)$ by Theorem\,\ref{thm:opt}. Obviously, the new $\dbtilde{\vct{u}}$ is also analytic and $\Lambda^*$-periodic on $D^*$.
Thus we conclude that $\dbtilde{\vct{u}}$ is the globally optimal assignment (up to a lattice vector in $\Lambda$), whose Wannier function's center and variance are given by (\ref{eq:center}) and (\ref{eq:var2}) respectively.

Suppose that $\tilde{\vct{u}}$ satisfies the symmetry in (\ref{eq:utrs}). By combining (\ref{eq:bckform}), (\ref{eq:bcxyform}) with the definition of $g$ in  (\ref{eq:rhs}), we have
\begin{align}
	g(\vct{k}) = -g(-\vct{k})\,,\quad \vct{k}\in D^*\,.
	\label{eq:ggreal}
\end{align}
Combining (\ref{eq:psisol}) and (\ref{eq:ggreal}) shows that 
\begin{align}
	\psi(\vct{k}) = -\psi(-\vct{k})\,,\quad \vct{k}\in D^*\,.
	\label{eq:psitrs}
\end{align}
Combining (\ref{eq:gtstage3}), (\ref{eq:psitrs}) and  (\ref{eq:utrs}), we obtain the symmetry relation
\begin{align}
	\bar{\dbtilde{\vct{u}}}(\vct{k}) = \dbtilde{\vct{u}}(-\vct{k})\,,\quad \vct{k}\in D^*\,.
		\label{eq:dbttrs}
\end{align}
As a result, the Fourier coefficients of ${\dbtilde{\vct{u}}}$ are also real, so is its corresponding Wannier function defined by (\ref{eq:wann2}).

This completes the construction of the assignment as stated in Section\,\ref{sec:ps}.
We conclude
that the Wannier function defined by (\ref{eq:wann2}) corresponding to the assignment $\dbtilde{\vct{u}}$ is exponentially
localized and has the optimal variance given in (\ref{eq:var2}). Furthermore, it is also real if the matrix $H$ has time-reversal symmetry.
\subsection{Realty of Wannier functions\label{sec:real}}
In this section, we prove that when the matrix $H$ has time-reversal symmetry (see \ref{eq:trs}), the assignment $\tilde{\vct{u}}$ obtained in Theorem\,\ref{thm:stage2} can be chosen to satisfy (\ref{eq:utrs2}). Consequently, its Fourier coefficients are real, and so is the corresponding Wannier function defined by (\ref{eq:wann2}). Such a choice can be made by choosing the initial vector $\vct{v}^0$ in (\ref{eq:init0}) to be real. First, we show that such a choice for $\vct{v}^0$ is always possible.

\begin{lemma}
	Suppose the matrix $H$ has time-reversal symmetry (\ref{eq:trs}).  Then $\vct{v}^0$ in (\ref{eq:init0}) can be chosen to be a real vector.
	\label{lem:realv}
\end{lemma}
\begin{proof}
	By periodicity and the symmetry (\ref{eq:trs}) of $H$, we have 
	\begin{align}
		\bar{H}(\vct{k}(-1/2,-1/2)) = H(\vct{k}(1/2,1/2)) = H(\vct{k}(-1/2,-1/2)) \,.
		\label{eq:equiv}
	\end{align}
	By complex conjugating (\ref{eq:init0}), the relations in (\ref{eq:equiv}) show that both $\vct{v}^0$ and $\bar{\vct{v}}^0$ satisfy the eigenvalue equation (\ref{eq:init0}). Since the eigenvalue is non-degenerate, $\vct{v}^0$ and $\bar{\vct{v}}^0$ can only differ by a phase factor, i.e. $\bar{\vct{v}}^0 = e^{2\I\varphi_0}\vct{v}^0$, for some real $\varphi_0$. If $\vct{v}^0$ is not real, $\pm e^{\I\varphi_0}\vct{v}^0$ will be real. 
\end{proof}
Next, we prove the fact that $\varphi_2$ in  (\ref{eq:phitopo}) is an even function when time-reversal symmetry is present, which will be used for proving the realty of Wannier functions.
\begin{lemma}
		Suppose that $H$ has time-reversal symmetry. Then $\varphi_2:\sb{-\half,\half}\rightarrow \R$ in  (\ref{eq:phitopo}) is even, i.e. $			\varphi_2(\kappa_1) = 	\varphi_2(-\kappa_1)$\,.
\label{lem:even}
\end{lemma}
\begin{proof}
	We integrate (\ref{eq:adiff}) from $-1/2$ to $\kappa_1$ for $\kappa_1 \in \sb{-\half,\half}$.  By applying Green's theorem over $\sb{-\half, \kappa_1}\times \sb{-\half,\half} \subseteq T$ as in Section\,\ref{sec:chern}, we obtain
	\begin{align}
		\varphi_2(\kappa_1) = -\I\log{z(\kappa_1)} = \int_{-\half}^{\kappa_1} \D s \int_{-\half}^{\half} \D \kappa_2 \, \Omega_{12}(\vct{k}(s,\kappa_2))\,.
		\label{eq:phi2int}
	\end{align}
	(For $\kappa_1=1/2$, we obtain (\ref{eq:reschern}).) By Remark\,\ref{rmk:trsp}, time-reversal symmetry implies the first Chern number $C_1=0$. Due to (\ref{eq:chernt}), $C_1=0$ implies that 
	\begin{align}
		 \int_{-\half}^{\kappa_1} \D s \int_{-\half}^{\half} \D \kappa_2 \, \Omega_{12}(\vct{k}(s,\kappa_2)) =  -\int_{\kappa_1}^{\half} \D s \int_{-\half}^{\half} \D \kappa_2 \, \Omega_{12}(\vct{k}(s,\kappa_2))\,.
		 \label{eq:cancel}
	\end{align}
	We replace $\kappa_1$ with $-\kappa_1$ in (\ref{eq:phi2int})) and apply (\ref{eq:cancel}) to obtain
	\begin{align}
	\varphi_2(-\kappa_1) = -\int_{-\kappa_1}^{\half} \D s \int_{-\half}^{\half} \D \kappa_2 \, \Omega_{12}(\vct{k}(s,\kappa_2))\,.
	\label{eq:phi2intnew}
\end{align}	
We apply a change of variable $s\rightarrow -s$ and $\kappa_2\rightarrow -\kappa_2$ to (\ref{eq:phi2intnew}) and use the symmetry (\ref{eq:bcurv12sym}) to yield the desired result:
	\begin{align}
	\varphi_2(-\kappa_1) = \int_{-\half}^{\kappa_1} \D s \int_{-\half}^{\half} \D \kappa_2 \, \Omega_{12}(\vct{k}(s,\kappa_2)) = \varphi_2(\kappa_1)\,.
\end{align}	
\end{proof}

In the following, we prove the realty of Wannier functions. The idea is to use the time-reversal symmetry to run the construction in Section\,\ref{sec:stage1} and \ref{sec:stage2} backward in time $\kappa_1, \kappa_2$. It can be viewed as an extension of the proof for the realty of Wannier functions in \cite{gopal2024high} to two dimensions for matrix models.
\begin{theorem}
		Suppose the matrix $H$ has time-reversal symmetry (\ref{eq:trs}) and $\vct{v}^0$ is chosen to be real according to Lemma\,\ref{lem:realv}. Then the assignment $\tilde{\vct{u}}$ in Theorem\,\ref{thm:stage2} satisfies
		\begin{equation}
			\bar{\tilde{\vct{u}}}(\vct{k}(\kappa_1,\kappa_2)) = \tilde{\vct{u}}(\vct{k}(-\kappa_1,-\kappa_2))\,, \quad (\kappa_1,\kappa_2)\in T\,.
			\label{eq:utrs22}
		\end{equation}
Hence, the Fourier coefficients of  $\tilde{\vct{u}}$ are real, and so is its corresponding Wannier function.
\label{thm:real}
\end{theorem}
\begin{proof}
	In Section\,\ref{sec:stage1}, $\tilde{\vct{u}}$ is constructed by solving (\ref{eq:dvk1}) with the initial condition (\ref{eq:dvk10}). We define a new vector $\vct{w}_1(\kappa_1, -1/2) = \bar{\tilde{\vct{u}}}(\vct{k}(-\kappa_1, -1/2))$. 
	We observe that, due to (\ref{eq:ptrs}), the periodicity of $P$ and the chain rule, $\vct{w}_1$ satisfies the same equation (\ref{eq:dvk1}):
	\begin{equation}
		\frac{\partial \vct{w}_1(\kappa_1, -1/2)}{\partial \kappa_1}   = 	\frac{\partial P\bb{\vct{k}(\kappa_1, -1/2)}}{\partial \kappa_1} \vct{w}_1(\kappa_1, -1/2) -\I \varphi_1 \vct{w}_1(\kappa_1, -1/2)
		\label{eq:dvk1r}
	\end{equation}
	 for $\kappa_1\in\sb{-\half,\half}$\,, subject to the initial condition
\begin{equation}
	\vct{w}_1(-1/2, -1/2) = \bar{\tilde{\vct{u}}}(\vct{k}(1/2, -1/2)) = \bar{\tilde{\vct{u}}}(\vct{k}(-1/2, -1/2)) =  \bar{\vct{v}}^0 = \vct{v}^0\,,
	\label{eq:dvk1r0}
\end{equation}
where the second equality is due to periodicity (see Lemma\,\ref{lem:stage1}) and the third one is due to the realty of $\vct{v}^0$ by assumption.
We observe that (\ref{eq:dvk1r}) with the initial condition (\ref{eq:dvk1r0}) is identical to (\ref{eq:dvk1}) with the initial condition (\ref{eq:dvk10}). By the uniqueness theorem of initial value problems, we conclude that $\vct{w}_1(\kappa_1, -1/2)= \tilde{\vct{u}}\bb{\vct{k}(\kappa_1, -1/2)}$\,, which implies that 
\begin{align}
	\tilde{\vct{u}}\bb{\vct{k}(\kappa_1, -1/2)} = \bar{\tilde{\vct{u}}}\bb{\vct{k}(-\kappa_1, -1/2)}\,.
	\label{eq:g0trs}
\end{align}

In Section\,\ref{sec:stage2}, for each $\kappa_1\in\sb{-\half,\half}$,  $\tilde{\vct{u}}$ in (\ref{eq:stage2gt}) is constructed by solving (\ref{eq:dvk2new}) with the initial condition (\ref{eq:dvk20}). 
Similarly, we define a new vector $\vct{w}_2(\kappa_1,\kappa_2) =  \bar{\tilde{\vct{u}}} \bb{\vct{k}(-\kappa_1, -\kappa_2)}$. Again, due to (\ref{eq:ptrs}) and the chain rule, for each $\kappa_1\in\sb{-\half,\half}$, $\vct{w}_2(\kappa_1,\kappa_2)$ satisfies the same equation as (\ref{eq:dvk2new}):
\begin{equation}
	\frac{\partial \vct{w}_2(\kappa_1,\kappa_2)}{\partial \kappa_2}   = 	\frac{\partial P\bb{\vct{k}(\kappa_1, \kappa_2)}}{\partial \kappa_2} \vct{w}_2(\kappa_1,\kappa_2) -\I\varphi_2(\kappa_1)\vct{w}_2(\kappa_1,\kappa_2)\,,
	\label{eq:dvk22}
\end{equation}
where we have used $\varphi_2(\kappa_1)=\varphi_2(-\kappa_1)$ in Lemma\,\ref{lem:even}. The initial condition is given by the formula
\begin{align}
	\vct{w}(\kappa_1, -1/2) = \bar{\tilde{\vct{u}}}\bb{\vct{k}(-\kappa_1, 1/2)} = \bar{\tilde{\vct{u}}}\bb{\vct{k}(-\kappa_1, -1/2)}=\tilde{\vct{u}}\bb{\vct{k}(\kappa_1, -1/2)}\,,
	\label{eq:dvk220}
\end{align}
where the second equality is due to periodicity (see Theorem\,\ref{thm:stage2}) and the third one is due to (\ref{eq:g0trs}). We observe that (\ref{eq:dvk22}) with the initial condition (\ref{eq:dvk220})  is identical to  (\ref{eq:dvk2new})  with the initial condition (\ref{eq:dvk20}). By the uniqueness theorem of initial value problems, we conclude that $ \vct{w}_2(\kappa_1,\kappa_2) =  \tilde{\vct{u}} \bb{\vct{k}(\kappa_1, \kappa_2)}$. Taking the complex conjugate yields the desirable result
\begin{align}
	\bar{\tilde{\vct{u}}} \bb{\vct{k}(\kappa_1, \kappa_2)} =  \tilde{\vct{u}} \bb{\vct{k}(-\kappa_1, -\kappa_2)}\,,\quad (\kappa_1,\kappa_2) \in T\,.
\end{align}

\end{proof}

\section{Numerical procedures \label{sec:numerics}}
In this section, we describe the numerical procedures for constructing optimally localized Wannier functions in Section\,\ref{sec:method1}. Section\,\ref{sec:init}--\ref{sec:nstage1} correspond to Section\,\ref{sec:stage1}, Section\,\ref{sec:npara2}--\ref{sec:ntopo} correspond to \ref{sec:stage2}, and Section\,\ref{sec:nbc}--\ref{sec:nwann} correspond to Section\,\ref{sec:stage3}.
The detailed description of the algorithms is in Section\,\ref{sec:algo}.

 The inputs to the algorithms are the primitive lattice vectors $\vct{a}_1, \vct{a}_2$ defining the lattice $\Lambda$ in (\ref{eq:lambda}),  the reciprocal primitive lattice $\vct{b}_1, \vct{b}_2$ defining the reciprocal lattice $\Lambda^*$ in (\ref{eq:rgortho}) and the torus $D^*$ in (\ref{eq:Ds}), a family of $n$ by $n$ matrix $H$, and a real number $h=\frac{1}{N}$, where $N$ is an even integer that specifies number of points for discretizing $D^*$ in each dimension. Furthermore, we assume the family of eigenvalues $E$ and eigenvectors $\vct{u}$ of interest has been chosen. 

Roughly speaking, the numerical procedure involves using a fourth-order Runger-Kutta (RK4) with Richardson extrapolation (see Section\,\ref{sec:rk4}) that achieves $O(h^6)$ global truncation error for solving the parallel transport equation at $N$ equispaced points in each dimension of $D^*$.  Since all quantities of interests by construction are analytic and $\Lambda^*$-periodic, their Fourier coefficients are approximated (with exponential convergence) via a discrete Fourier transform introduced in Section\,\ref{sec:dft}.  Subsequently, all differentiation operations for computing Berry connections (and Berry curvatures in Appendix\,\ref{sec:method2}) and solving Poisson's equations are carried out with the help of their Fourier series; almost no accuracy is lost during these steps. 
As a result, the accuracy of the procedure in this paper is determined by the accuracy of the ODE solver, which is $O(h^6)$ in this case, for sufficiently large $N$. Moreover, the computation time is also dominated by solving the parallel transport equation. It should be observed that we choose RK4 in this paper for simplicity. The choice is by no means optimal and can be easily replaced by higher-order methods for better convergence, such as spectral deferred correction schemes\,\cite{dutt2000spectral}, or multistep methods for fewer function evaluations.

It also should be observed that, although we solve the parallel transport equation explicitly in this paper to obtain the band structure and the Wannier function simultaneously, the two parts can be decoupled according to Remark\,\ref{rmk:twist} after band structures are obtained. The results produced by this approach is only a second-order scheme but it can be done with very little cost. We refer the reader to Example 1 in Section\,\ref{sec:eg1} for details.

\subsection{Discretizing the torus $D^*$ \label{sec:init}}
We parameterize $D^*$ by $T$ in the form of (\ref{eq:xydir}) and define $(N+1)^2$ equispaced points $\cb{(\kappa^{(j_1)}_1,\kappa^{(j_2)}_2)}$ in $T$ by the formulas
\begin{align}
	\kappa^{(j_1)}_1 = j_1 h\,,\quad 	\kappa^{(j_2)}_2 = j_2 h\,, \quad j_1,j_2=-N/2, -(N/2-1),\ldots, N/2\,,
	\label{eq:npts}
\end{align} 
According to (\ref{eq:xydir}), these points 
correspond to 
\begin{equation}
	\vct{k}^{(j_1,j_2)} = \kappa^{(j_1)}_1\vct{b}_1 + \kappa^{(j_2)}_2\vct{b}_2	
		\label{eq:nkpts}
\end{equation}
 in $D^*$\,. Given any function $f$ defined on $D^*$, we use the notation $f^{(j_1,j_2)}$ to represent the computed value of $f$ at $\vct{k}^{(j_1,j_2)}$ so that we have
 \begin{equation}
 	f^{(j_1,j_2)} \approx f(\vct{k}^{(j_1,j_2)})\,.
 \end{equation}
 
\subsection{Step 1: computing an assignment on the line $\gamma_0$ \label{sec:nstage1}}
We compute the parallel transport of eigenvectors described in Section\,\ref{sec:stage1} on the line $\gamma_0$ in (\ref{eq:gam0}). 
First, we obtain an initial condition by finding the eigenvalue $E^0$ and eigenvector $\vct{v}^0$ 
\begin{align}
	H(\vct{k}(-1/2,-1/2))\vct{v}^0 = 	E^0\vct{v}^0\,.
	\label{eq:nint}
\end{align}
This is done by the standard QR algorithm in $O(n^3)$ operations. When $H$ has time-reversal symmetry, we choose $\vct{v}^0$ to be a real vector (see Lemma\,\ref{lem:realv}) for the realty of the Wannier function by Theorem\,\ref{thm:real}. If this is not the case, any choice of $\vct{v}^0$ is accepted.

Instead of using (\ref{eq:dvk11}), we solve its equivalent version stated in (\ref{eq:evalpert}) and (\ref{eq:evecpert}). Namely, we solve the following system of ODEs
\begin{align}
\begin{split}
	&\frac{\partial }{\partial \kappa_1} E(\vct{k}(\kappa_1,-1/2)) = \vct{u}^*(\vct{k}(\kappa_1,-1/2))\frac{\partial H}{\partial \kappa_1}\vct{u}(\vct{k}(\kappa_1,-1/2))\,,\\
		&\frac{\partial }{\partial \kappa_1} \vct{u}(\vct{k}(\kappa_1,-1/2)) = -(H- E)^\dagger \frac{\partial H}{\partial \kappa_1}\vct{u}(\vct{k}(\kappa_1,-1/2))\,,
\end{split}
		\label{eq:ne1}
\end{align}
subject to the initial condition computed in (\ref{eq:nint})
\begin{equation}
	E(\vct{k}(-1/2,-1/2)) = E^0\,,	\quad \vct{u}(\vct{k}(-1/2,-1/2)) = \vct{v}^0\,.
\end{equation}
We solve the system above by RK4 for $\kappa_1 \in \sb{-\half,\half}$ with $h=1/N$, so that we obtain the eigenvalues and eigenvectors
\begin{equation}
	E^{(j_1,-N/2)}\,,\quad \vct{u}^{(j_1,-N/2)}\,,\quad  \mbox{ for $j_1=-N/2,\ldots,N/2\,, j_2 = -N/2$}.
	\label{eq:neu1}
\end{equation}
Richardson extrapolation is done by repeating the above steps with $h$ replaced by $h/2$ and $h/4$, followed by taking their differences as in (\ref{eq:richard}).

At each step, the pseudoinverse $(H- E)^\dagger$ is computed by a singular value decomposition of
$H- E$ (see Algorithm 1A), discarding the component corresponding to the smallest singular value, and applying
the inverse of the decomposition directly to the right-hand side. The details can be found, for example, in \cite{dahlquist2003numerical} and this approach is related to the form in Remark\,\ref{rmk:svd}. At every step, each pseudoinverse computation requires $O(n^3)$ operations, so computing (\ref{eq:neu1}) requires $O(n^3 N)$ operations.

Next, we compute the phase $\varphi_1$ in (\ref{eq:phik2}) by (\ref{eq:logphik2}), which is given by
\begin{align}
	\varphi_1 = -\I\log	\bb{\vct{u}^{*(-N/2,-N/2)}\vct{u}^{(N/2,-N/2)}}\,.
	\label{eq:nphi1}
\end{align}
Next, we apply the gauge transform defined by (\ref{eq:stage1gt}) and we have
\begin{align}
	\tilde{\vct{u}}^{(j_1,j_2)} = e^{-\I \varphi_1 (j_1 h+\half)} \vct{u}^{(j_1,j_2)}
	\label{eq:ntu}
\end{align}
for $j_1=-N/2,\ldots,N/2 - 1$ and $j_2 = -N/2$. We observe that $j_1=N/2$ is removed since $\tilde{\vct{u}}$ is periodic by Lemma\,\ref{lem:stage1}.

\subsection{Step 2: computing parallel transport on lines $\gamma_{\kappa_1}$ \label{sec:npara2}}
We compute the parallel transport described in Section\,\ref{sec:stage1} on the line $\gamma_{\kappa_1}$ in (\ref{eq:gamk1}) for $\kappa_1\in\sb{-\half,\half}$.  For each $\kappa^{(j_1)}_1 = j_1 h$  with $j_1=-N/2,\ldots,-N/2-1$, the parallel transport equations are given by
\begin{align}
\begin{split}
	&\frac{\partial }{\partial \kappa_2} E(\vct{k}(\kappa^{(j_1)}_1,\kappa_2)) = \vct{u}^*(\vct{k}(\kappa^{(j_1)}_1,\kappa_2))\frac{\partial H}{\partial \kappa_2}\vct{u}(\vct{k}(\kappa^{(j_1)}_1, \kappa_2))\,,\\
		&\frac{\partial }{\partial \kappa_2} \vct{u}(\vct{k}(\kappa^{(j_1)}_1,\kappa_2)) = -(H- E)^\dagger \frac{\partial H}{\partial \kappa_2}\vct{u}(\vct{k}(\kappa^{(j_1)}_1,\kappa_2))\,,
\end{split}
\label{eq:ne2}
\end{align}
subject to the initial condition
\begin{align}
		E(\vct{k}(\kappa^{(j_1)}_1,-1/2)) = E^{(j_1,-N/2)}\,,	\quad \vct{u}(\vct{k}(\kappa^{(j_1)}_1,-1/2)) = \tilde{\vct{u}}^{(j_1,-N/2)} \,,
\end{align}
which are outputs from Section\,\ref{sec:nstage1}. Similar to solving (\ref{eq:neu1}), given each $\kappa^{(j_1)}_1$, the ODE system is solved by RK4 for $\kappa_2 \in \sb{-\half,\half}$ with $h=1/N$. Thus we obtained the computed eigenvalues  and eigenvectors  
\begin{align}
	E^{(j_1,j_2)}\,, \quad\vct{u}^{(j_1,j_2)}\,,\quad \mbox{for $j_1=-N/2,\ldots,N/2-1\,,j_2 = -N/2,\ldots,N/2$\,. }
	\label{eq:neu2}
\end{align}
Again, Richardson extrapolation is done by repeating the above steps with $h$ replaced by $h/2$ and $h/4$, followed by taking their differences as in (\ref{eq:richard}). For each $\kappa^{(j_1)}_1$, the cost of ODE solves is the same as the one in computing (\ref{eq:neu1}), and it is repeated $N$ times for different $j_1$. Thus the total cost of computing (\ref{eq:neu2}) is $O(n^3 N^2)$ operations. This is the dominant cost of the entire construction.

\begin{remark}
	We observe that the dependence on the eigenvalue in (\ref{eq:ne1}) and (\ref{eq:ne2}) can be removed by computing the eigenvalue from the computed eigenvector $\vct{u}$ via the Rayleigh quotient
	\begin{equation}
		E = \frac{\vct{u}^*H\vct{u}}{\vct{u}^*\vct{u}}\,.
	\end{equation}
This also squares the error of the computed eigenvector, doubling the convergence rate for $E$.
\end{remark}

\subsection{Step 3: determining the topological obstruction \label{sec:ntopo}}
In order to apply (\ref{eq:reschern}) to compute the first Chern number $C_1$,
we determine the value $z^{(j_1)} $ of $z$ in (\ref{eq:z}) at $\kappa_1 =\kappa_1^{(j_1)} $ with $j_1=-N/2,\ldots,-N/2-1$ by the formula
\begin{equation}
	z^{(j_1)} = \vct{u}^{*(j_1,-N/2)}\vct{u}^{(j_1,N/2)}\,,
	\label{eq:nzz}
\end{equation}
where the quantities on right are outputs from Section\,\ref{sec:npara2}. Lemma\,\ref{lem:stage21} implies that $z$ is analytic and periodic in $\kappa_1$. This permits efficient and accurate evaluation of  $z'(\kappa^{(j_1)}_1)$ by the Fourier series of $z$. We first compute the Fourier  coefficients $\hat{z}_m$ with $m=-N/2,\ldots,N/2-1$ via a discrete Fourier transform (see Section\,\ref{sec:dft} for the two-dimensional version)
\begin{align}
	\hat{z}_m = \frac{1}{N}\sum_{j_1 = -N/2}^{N/2-1} e^{-\frac{2\pi \I}{N} m j_1} z^{(j_1)} \,.
	\label{eq:zfour}
\end{align}
We differentiate $z$ in the Fourier basis and apply the inverse of the discrete Fourier transform, obtaining an approximation $z'^{(j_1)}$ of $z'$ at $\kappa_1 =\kappa_1^{(j_1)} $ with $j_1=-N/2,\ldots,-N/2-1$ by the formula 
\begin{equation}
	z'^{(j_1)} =  \sum_{m = -N/2}^{N/2-1} e^{\frac{2\pi \I}{N} m j_1} (2\pi\I m \hat{z}_m)\,.
		\label{eq:zpfour}
\end{equation}
Both (\ref{eq:zfour}) and (\ref{eq:zpfour}) are computed via the FFT in $O(N\log N)$ operations.

Then we apply the trapezoidal rule (\ref{eq:trapz}) to (\ref{eq:reschern}) and round the real part of the sum to the nearest integer to obtain $C_1$. More explicitly, we have
\begin{equation}
	C_1 = \left\lfloor \Re{\frac{h}{2\pi\I} \sum_{j_1 = -N/2}^{N/2-1} \frac{z'^{(j_1)}}{{z^{(j_1)}}}} \right\rceil\,,
	\label{eq:nchern}
\end{equation}
where $\lfloor \cdot\rceil$ denotes the function for rounding to the nearest integer. By Theorem\,\ref{thm:trapz}, all approximations above converge exponentially, thus highly accurate for sufficiently large $N$. Hence, the rounding operation in (\ref{eq:nchern}) is guaranteed to produce the correct integer value.

If $C_1$ is computed and found to be nonzero, we terminate the construction and return $\vct{u}^{(j_1,j_2)}$ for $j_1=-N/2,\ldots,N/2-1$ and $j_2 = -N/2,\ldots,N/2$ computed in Section\,\ref{sec:npara2}\,. By Lemma\,\ref{lem:stage21}, $\vct{u}(\vct{k}(\kappa_1,\kappa_2))$ is analytic in both $\kappa_1$ and $\kappa_2$. It is also periodic in $\kappa_1$ but not in $\kappa_2$. This is the best assignment one can hope for in the presence of the obstruction $C_1\ne 0$ by Lemma\,\ref{lem:topo}\,.

If $C_1$ is found to be zero, we continue the construction and compute $\varphi_2$ by (\ref{eq:phitopo}). Thus we determine the value $\varphi_2^{(j_1)}$ of $	\varphi_2$  at $\kappa_1 =\kappa_1^{(j_1)} $ with $j_1=-N/2,\ldots,-N/2-1$ by the formula
\begin{equation}
	\varphi_2^{(j_1)} = -\I \log(z^{(j_1)})\,,\quad j_1 = -N/2,\ldots,N/2-1\,.
	\label{eq:nphi2}
\end{equation}
By Lemma\,\ref{lem:topo}, the branch for evaluating the log function is chosen so that $\varphi_2$ is periodic and analytic. (This may not coincide with the principal branch.)
Then we carry out the gauge transform in (\ref{eq:stage2gt}) to obtain the new eigenvectors
\begin{align}
	\tilde{\vct{u}}^{(j_1,j_2)} = e^{-\I \varphi_2^{(j_1)} (j_2 h+\half)} \vct{u}^{(j_1,j_2)}
	\label{eq:ntu}
\end{align}
for $j_1=-N/2,\ldots,N/2-1$ and $j_2 = -N/2,\ldots,N/2-1$. The point $j_2=N/2$ is removed since $\tilde{\vct{u}}$ is periodic (and analytic) in $\kappa_2$ by Theorem\,\ref{thm:stage2}.

If $H$ has time-reversal symmetry, Theorem\,\ref{thm:real} shows that $\tilde{\vct{u}}^{(j_1,j_2)}$ has the approximate symmetry
\begin{equation}
	\bar{\tilde{\vct{u}}}^{(j_1,j_2)} \approx \tilde{\vct{u}}^{(-j_1,-j_2)}
	\label{eq:ntrs}
\end{equation}
for $j_1=-N/2,\ldots,N/2-1$ and $j_2 = -N/2,\ldots,N/2-1$.

 \begin{remark}
	If two vectors $\vct{v}_1 = \vct{u}(\vct{k})$ and $\vct{v}_2(\vct{k} + h\vct{n} )$	for small $h$ in the direction given by $\vct{n}$ are obtained by direct eigensolves, the phases of $\vct{v}_1$ and $\vct{v}_2$ are independent. In \cite{marzari1997maximally} and \cite{vanderbilt2018berry}, it is pointed out that the (approximate) parallel transported vector $\widetilde{\vct{v}}_2$ from $\vct{v}_1$ given by 
	$		\widetilde{\vct{v}}_2 = e^{\I \beta}\vct{v}_2$, 
	where $\beta$ is the phase given by $\beta = -\Im\log(\vct{v}_1^*\vct{v}_2)$. 
	Provided the eigenvectors at each grid point (up to any phase) have been obtained,	this approach can be used to obtain $\tilde{\vct{u}}^{(j_1,j_2)} $ in (\ref{eq:ntu}) directly. However, the scheme is only a second-order scheme; we refer the reader to Table\,\ref{tab:eg1twist} for Example 1 in Section\,\ref{sec:eg1} for more details.
	
	\label{rmk:twist}
\end{remark}

\begin{remark}
	It should be observed that the Wannier function corresponding to $\tilde{\vct{u}}$ is already exponentially localized. Furthermore, it is also real if $H$ has time-reversal symmetry by (\ref{eq:ntrs}). The following Section\,\ref{sec:nbc} and \ref{sec:div} only compute the optimal Wannier function by reducing the variance of the Fourier coefficients of $\tilde{\vct{u}}$. Hence, they can be skipped for  Section\,\ref{sec:nwann} directly if the optimality of the Wannier function is not needed with all $\dbtilde{\vct{u}}$ replace by $\tilde{\vct{u}}$ in (\ref{eq:ntu}).
\end{remark}

\subsection{Step 4: computing the Berry connection\label{sec:nbc}}
Having computed $\tilde{\vct{u}}^{(j_1,j_2)}$ for  $j_1, j_2=-N/2,\ldots,N/2-1$, we first compute the derivatives
\begin{align}
	\frac{\partial }{\partial \kappa_1} 	\tilde{\vct{u}}^{(j_1,j_2)}\,, \quad 	\frac{\partial }{\partial \kappa_2} 	\tilde{\vct{u}}^{(j_1,j_2)}
\end{align}
by the Fourier series of $\tilde{\vct{u}}$ in order to compute the Berry connection in (\ref{eq:bckform}).
We denote the components of $	\tilde{\vct{u}}$ by $	\tilde{u}_{i}$ for $i=1,2,\ldots,n$.  We compute the derivatives above by first computing the approximate Fourier coefficients by a discrete Fourier transform, then differentiating its Fourier series with respect to $\kappa_1$ and $\kappa_2$, followed by the inverse transform. More explicitly, by notations in Section\,\ref{sec:dft}, for each component $i=1,\ldots,n$, we have
\begin{align}
		\frac{\partial }{\partial \kappa_1} 	\tilde{u}_{i}^{(j_1,j_2)} = \bb{\mathscr{F}^{-1}_N \mathscr{D}_1 \mathscr{F}_N(	\tilde{u}_{i})}_{j_1j_2} \,, \quad 		\frac{\partial }{\partial \kappa_2} 	\tilde{u}_{i}^{(j_1,j_2)} = \bb{\mathscr{F}^{-1}_N \mathscr{D}_2 \mathscr{F}_N(	\tilde{u}_{i})}_{j_1j_2} \,,
		\label{eq:ndktu}
\end{align}
for $j_1, j_2=-N/2,\ldots,N/2-1$, where $\mathscr{D}_1 $ and $\mathscr{D}_2$  are defined below (\ref{eq:ddftg}). This involves in total $2n$ FFTs.

By applying (\ref{eq:bckform}), we compute the components of the Berry connection  parallel to $\vct{b}_1$ and $\vct{b}_2$ via
\begin{align}
	A_1^{(j_1,j_2)} = \sum_{i=1}^n 	\bar{	\tilde{u}}_{i}^{(j_1,j_2)} \frac{\partial }{\partial \kappa_1} 	\tilde{u}_{i}^{(j_1,j_2)}\,, \quad	A_2^{(j_1,j_2)} = \sum_{i=1}^n 	\bar{	\tilde{u}}_{i}^{(j_1,j_2)} \frac{\partial }{\partial \kappa_2} 	\tilde{u}_{i}^{(j_1,j_2)} \,,
	\label{eq:na12}
\end{align}
for  $j_1, j_2=-N/2,\ldots,N/2-1$\,.
By (\ref{eq:bcxyform}), the $x,y$ components of the Berry connection are given by
\begin{align}
		\begin{split}
	A_x^{(j_1,j_2)} = \frac{1}{2\pi}\vct{a}_1\cdot\vct{e}_x A_1 ^{(j_1,j_2)} + \frac{1}{2\pi}\vct{a}_2\cdot\vct{e}_x A_2^{(j_1,j_2)}\,,\\
	A_y^{(j_1,j_2)} = \frac{1}{2\pi}\vct{a}_1\cdot\vct{e}_y A_1 ^{(j_1,j_2)} + \frac{1}{2\pi}\vct{a}_2\cdot\vct{e}_y A_2^{(j_1,j_2)}\,,
		\end{split}
		\label{eq:nbcxy}
\end{align}	
for  $j_1, j_2=-N/2,\ldots,N/2-1$. 

\begin{remark}
	\label{rmk:var}
	We observe that the moment functions given by  (\ref{eq:Ru}) and (\ref{eq:R2u}) for $\tilde{\vct{u}}$ can be computed by first computing 
	\begin{align}
			\begin{split}
		\frac{\partial }{\partial k_x} 	\tilde{u}_{i}^{(j_1,j_2)} =&   \frac{1}{2\pi}\vct{a}_1\cdot\vct{e}_x \frac{\partial }{\partial \kappa_1} 	\tilde{u}_{i}^{(j_1,j_2)} +  \frac{1}{2\pi}\vct{a}_2\cdot\vct{e}_x\frac{\partial }{\partial \kappa_2} 	\tilde{u}_{i}^{(j_1,j_2)}\,, \\
		\frac{\partial }{\partial k_y} 	\tilde{u}_{i}^{(j_1,j_2)} =&   \frac{1}{2\pi}\vct{a}_1\cdot\vct{e}_y \frac{\partial }{\partial \kappa_1} 	\tilde{u}_{i}^{(j_1,j_2)} +  \frac{1}{2\pi}\vct{a}_2\cdot\vct{e}_y \frac{\partial }{\partial \kappa_2} 	\tilde{u}_{i}^{(j_1,j_2)}\,, 
			\end{split}
	\end{align}
by results computed in 	(\ref{eq:ndktu}) with the formula (\ref{eq:diffc}). Then, for the $x,y$ component of the Wannier center (\ref{eq:Ru}), we have
\begin{align}
	\begin{split}
	&\langle \vct{R}_x \rangle \approx  \I h^2 \sum_{j_1,j_2=-N/2}^{N/2-1} \sum_{i=1}^n \bar{	\tilde{u}}_{i}^{(j_1,j_2)} \frac{\partial }{\partial k_x} 	\tilde{u}_{i}^{(j_1,j_2)}\,,\\
	&\langle \vct{R}_y \rangle \approx  \I h^2 \sum_{j_1,j_2=-N/2}^{N/2-1} \sum_{i=1}^n \bar{	\tilde{u}}_{i}^{(j_1,j_2)} \frac{\partial }{\partial k_y} 	\tilde{u}_{i}^{(j_1,j_2)}\,,
	\end{split}
\end{align}
and
\begin{align}
	\langle \norm{\vct{R}}^2 \rangle \approx  h^2 \sum_{j_1,j_2=-N/2}^{N/2-1} \sum_{i=1}^n \bb{\abs{\frac{\partial }{\partial k_x} 	\tilde{u}_{i}^{(j_1,j_2)}}^2 + \abs{\frac{\partial }{\partial k_y} 	\tilde{u}_{i}^{(j_1,j_2)}}^2}\,.
\end{align}
By Theorem\,\ref{thm:trapz}, all trapezoidal rule approximations above converge exponentially as $N$ increases, and the accuracy is only limited by how accurate $\tilde{\vct{u}}$ is computed.
\end{remark}

\subsection{Step 5: eliminating the divergence of the Berry connection \label{sec:ndiv}}
Based on the computed Berry connection in (\ref{eq:nbcxy}), we solve the Poisson's equation (\ref{eq:divlessc3}) by finding the divergence $g$ in (\ref{eq:rhs}). Similar to steps in Section\,\ref{sec:nbc}, all derivatives are done in the Fourier domain, where the equation (\ref{eq:divless3}) is diagonal, thus easily solved.
First, we compute the Fourier coefficients $	\hat{A}_x,	\hat{A}_y$ of $A_x, A_y$ in (\ref{eq:nbcxy}) by
\begin{align}
	\hat{A}_x = \mathscr{F}_N(A_x)\,,\quad 	\hat{A}_y = \mathscr{F}_N(A_y)\,.
	\label{eq:naxyfour}
\end{align}
By the change of variable formulas in (\ref{eq:diffc}), we compute the Fourier coefficients $\hat{g}$ of $g$ in (\ref{eq:rhs}) by the formula
\begin{align}
			\begin{split}
	\hat{g} &=   \bb{\frac{1}{2\pi}\vct{a}_1\cdot\vct{e}_x\mathscr{D}_1+ \frac{1}{2\pi}\vct{a}_2\cdot\vct{e}_x \mathscr{D}_2}\hat{A}_x +  \bb{\frac{1}{2\pi}\vct{a}_1\cdot\vct{e}_y \mathscr{D}_1+ \frac{1}{2\pi}\vct{a}_2\cdot\vct{e}_y \mathscr{D}_2}\hat{A}_y \,,
			\end{split}
			\label{eq:ndivf}
\end{align}
 where $\mathscr{D}_1 $ and $\mathscr{D}_2$  are defined as in (\ref{eq:ndktu}), approximating $\frac{\partial}{\partial \kappa_1}$ and $\frac{\partial}{\partial \kappa_2}$ respectively in the Fourier basis. 
 Having computed $\hat{g}$ in (\ref{eq:ndivf}), we solve the Poisson's equation by (\ref{eq:psisol}) to obtain the potential $\psi$ for the divergence of the Berry connection. We first compute the Fourier coefficients $\hat{\psi}$ of $\psi$ by the formula
 \begin{align}
 		\hat{\psi}_{m_1 m_2} = 
 	\begin{cases}
 	 0\,,	 & m_1=m_2=0\,,\\
 	 \frac{\hat{g}_{m_1 m_2}}{\norm{m_1\vct{a_1} + m_2\vct{a_2}}^2} \,, &  \mbox{$m_1\ne 0$ or $m_2\ne 0$ }\,,
 	\end{cases}
 	\label{eq:npsifour}
 \end{align}
for  $m_1, m_2=-N/2,\ldots,N/2-1$. We apply the inverse discrete Fourier transform to $\hat{\psi}$ to obtain
\begin{align}
	\psi^{(j_1,j_2)}  = \bb{ \mathscr{F}^{-1}_N (\hat{\psi})}_{j_1 j_2}
	\label{eq:npsi}
\end{align}
for  $j_1, j_2=-N/2,\ldots,N/2-1$\,.
Thus solving for $	\psi^{(j_1,j_2)}$ takes in total $3$ FFTs. 

Finally, we apply the gauge transformation in (\ref{eq:gtstage3}) that eliminates the divergence of the Berry connection of $\tilde{\vct{u}}$ to obtain
\begin{align}
		\dbtilde{\vct{u}}^{(j_1,j_2)} = 	e^{-\I 	\psi^{(j_1,j_2)}}\tilde{\vct{u}}^{(j_1,j_2)}
\end{align}
for  $j_1, j_2=-N/2,\ldots,N/2-1$\,, which is the optimal assignment by Theorem\,\ref{thm:opt}.

\subsection{Step 6: computing the Wannier function \label{sec:nwann}}
To compute the Wannier function corresponding to $	\dbtilde{\vct{u}}$, we use the definition in (\ref{eq:wann2}). We compute the Fourier coefficients $\hat{\dbtilde{u}}_i$ of $	\dbtilde{u}_i$, the $i$-th component of $	\dbtilde{\vct{u}}$ for $i=1,\ldots,n$ by the formula
\begin{align}
	\hat{\dbtilde{u}}_i = { \mathscr{F}_N (\dbtilde{u}_i)}\,.
	\label{eq:dbufour}
\end{align}
 This takes $n$ FFTs.
By (\ref{eq:wann2}), we have
\begin{align}
	W_0(\vct{r}) \approx \sum_{i=1}^n \sum_{m_1, m_2 =-N/2}^{N/2-1} 	\hat{\dbtilde{u}}_{i, m_1 m_2}	 \cdot \phi_i(\vct{r} +\vct{R}_{m_1 m_2})
	\label{eq:nwann}
\end{align}
where $\vct{R}_{m_1 m_2} = m_1 \vct{a}_1 + m_2 \vct{a}_2$. 

\subsection{Detailed description of the algorithms\label{sec:algo}}
This section contains a detailed description of the algorithms described in Section\,\ref{sec:numerics}.

\noindent \underline{\textbf{Algorithm}} 

\noindent \textbf{Initialization} 
\begin{enumerate}
	\item Choose an even integer $N$ and compute $h=1/N$.
	\item Form $\cb{(\kappa_1^{(j_1)}, \kappa_2^{(j_2)}) }_{j_1,j_2=-N/2}^{N/2}$   by (\ref{eq:npts}) and $\cb{\vct{k}^{(j_1,j_2)} }_{j_1,j_2=-N/2}^{N/2}$   by (\ref{eq:nkpts}).
\end{enumerate}
\noindent \textbf{Step 1} [Parallel transport on $\gamma_{0}$] 
\begin{enumerate}
	\item Compute $E^0$ and $\vct{v}^0$ in (\ref{eq:nint}) by QR and normalize $\norm{{v}^0} = 1$.
	\item \textbf{If}  $H$ has time-reversal symmetry \textbf{then}\\
	Choose $\vct{v}^0$ to be a real vector.\\
	\textbf{End if}
	\item Solve (\ref{eq:ne1}) by RK4  to compute $\cb{E_{(1)}^{(j_1,-N/2)}}_{j_1=-N/2}^{N/2}$ and $\cb{\vct{u}_{(1)}^{(j_1,-N/2)}}_{j_1=-N/2}^{N/2}$\,,\\
	where Algorithm 1A is at each time step for evaluating $\frac{\partial}{\partial \kappa_1}E$ and $\frac{\partial}{\partial \kappa_1}\vct{u}$\,.
	\item Set $h_{(2)} = h/2$ and repeat 1 to obtain  $\cb{E_{(2)}^{(j_1,-N/2)}}_{j_1=-N/2}^{N/2}$ and $\cb{\vct{u}_{(2)}^{(j_1,-N/2)}}_{j_1=-N/2}^{N/2}$\,.
	\item Set $h_{(3)} = h/4$ and repeat 1 to obtain  $\cb{E_{(3)}^{(j_1,-N/2)}}_{j_1=-N/2}^{N/2}$ and $\cb{\vct{u}_{(3)}^{(j_1,-N/2)}}_{j_1=-N/2}^{N/2}$\,.
	\item Use outputs from 3 to 5 to perform Richardson extrapolation in (\ref{eq:richard}) to obtain \\
	$\cb{E^{(j_1,-N/2)}}_{j_1=-N/2}^{N/2}$ and $\cb{\vct{u}^{(j_1,-N/2)}}_{j_1=-N/2}^{N/2}$ with $O(h^6)$ error.
	\item Compute $\varphi_1$ by (\ref{eq:nphi1}).
	\item \textbf{Do} $j_1=-N/2,\ldots,N/2-1$
	\item[] Set $\tilde{\vct{u}}^{(j_1,-N/2)} = e^{-\I\varphi_1 (j_1 h+\half)}\vct{u}^{(j_1,-N/2)}$.
\item[]	\textbf{End do}
\end{enumerate}
\noindent \textbf{Step 2} [Parallel transport on $\gamma_{\kappa_1}$] 
\begin{enumerate}
	\item Solve (\ref{eq:ne2}) by RK4  to compute $\cb{E_{(1)}^{(j_1,j_2)}}_{j_1,j_2=-N/2}^{N/2}$ and $\cb{\vct{u}_{(1)}^{(j_1,j_2)}}_{j_1,j_2=-N/2}^{N/2}$\,,\\
	where Algorithm 1A is at each time step for evaluating $\frac{\partial}{\partial \kappa_1}E$ and $\frac{\partial}{\partial \kappa_1}\vct{u}$\,.
	\item Set $h_{(2)} = h/2$ and repeat 1 to obtain  $\cb{E_{(2)}^{(j_1,j_2)}}_{j_1,j_2=-N/2}^{N/2}$ and $\cb{\vct{u}_{(2)}^{(j_1,j_2)}}_{j_1,j_2=-N/2}^{N/2}$\,.
	\item Set $h_{(3)} = h/4$ and repeat 1 to obtain  $\cb{E_{(3)}^{(j_1,j_2)}}_{j_1,j_2=-N/2}^{N/2}$ and $\cb{\vct{u}_{(3)}^{(j_1,j_2)}}_{j_1,j_2=-N/2}^{N/2}$\,.
	\item Use outputs from 1 to 3 to perform Richardson extrapolation in (\ref{eq:richard}) to obtain \\
	$\cb{E^{(j_1,j_2)}}_{j_1,j_2=-N/2}^{N/2}$ and $\cb{\vct{u}^{(j_1,j_2)}}_{j_1,j_2=-N/2}^{N/2}$ with $O(h^6)$ error.
\end{enumerate}
\noindent \textbf{Step 3} [Computing the first Chern number]
\begin{enumerate}
	\item Compute  $\cb{z^{(j_1)}}_{j_1=-N/2}^{N/2-1}$ by (\ref{eq:nzz}).
	\item Compute $\cb{\hat{z}_m}_{m=-N/2}^{N/2-1}$ in (\ref{eq:zfour}) and $\cb{z'^{(j_1)}}_{j_1=-N/2}^{N/2-1}$ in (\ref{eq:zpfour}) by FFT.
	\item Compute $C_1$ by (\ref{eq:nchern}).
	\item \textbf{If} $C_1 \ne 0$ \textbf{then} 
	\item[] \textbf{Return} $\cb{E^{(j_1,j_2)}}_{j_1,j_2=-N/2}^{N/2}$ and $\cb{\vct{u}^{(j_1,j_2)}}_{j_1,j_2=-N/2}^{N/2}$\,.
	\item[] \textbf{Stop}[Topological obstruction encountered]
	\item[] \textbf{Else if} $C_1 = 0$ \textbf{then}
	\item Compute $\cb{\varphi_2^{(j_1)}}_{j_1=-N/2}^{N/2-1}$ by (\ref{eq:nphi2}).
	\item \textbf{Do} $j_1,j_2=-N/2,\ldots,N/2-1$
	\item[] Set $\tilde{\vct{u}}^{(j_1,j_2)} = e^{-\I\varphi^{(j_1)}_2 (j_2 h+\half)}\vct{u}^{(j_1,j_2)}$.
	\item[]	\textbf{End do}
	\item[] \textbf{End if}

	\item \textbf{If} Optimal Wannier function is not required \textbf{then} 
	\item[] Go to \textbf{Step 6}.
	\item[] \textbf{End if}
\end{enumerate}

\noindent \textbf{Step 4} [Computing the Berry connection]
\begin{enumerate}
	\item \textbf{Do} $i=1,2,$\ldots$,n$
	\item[] Compute $\cb{\frac{\partial }{\partial \kappa_1} 	\tilde{u}_{i}^{(j_1,j_2)}}_{j_1,j_2=-N/2}^{N/2-1}, \cb{\frac{\partial }{\partial \kappa_2} 	\tilde{u}_{i}^{(j_1,j_2)}}_{j_1,j_2=-N/2}^{N/2-1}$ in (\ref{eq:ndktu}) by FFT.
	\item[] \textbf{End do}
		\item Compute $\cb{A_1^{(j_1,j_2)}}_{j_1,j_2=-N/2}^{N/2-1}, \cb{A_2^{(j_1,j_2)}}_{j_1,j_2=-N/2}^{N/2-1}$ by (\ref{eq:na12}).
	\item Compute $\cb{A_x^{(j_1,j_2)}}_{j_1,j_2=-N/2}^{N/2-1}, \cb{A_y^{(j_1,j_2)}}_{j_1,j_2=-N/2}^{N/2-1}$ by (\ref{eq:nbcxy}).
\end{enumerate}
\noindent \textbf{Step 5} [Eliminating the divergence]
\begin{enumerate}
\item Compute $\hat{A}_x, \hat{A}_y$ in (\ref{eq:naxyfour}) by FFT.
\item Compute $\hat{g}$ by (\ref{eq:ndivf}).
\item Compute $\hat{\psi}$ by (\ref{eq:npsifour}).
\item Compute $\cb{\psi^{(j_1,j_2)}}_{j_1,j_2=-N/2}^{N/2-1}$ in (\ref{eq:npsi}) by FFT.
\item \textbf{Do} $j_1,j_2=-N/2,\ldots,N/2-1$
\item[]	Set $\dbtilde{\vct{u}}^{(j_1,j_2)} = e^{-\I\psi^{(j_1,j_2)}}\tilde{\vct{u}}^{(j_1,j_2)}$\,.
\item[]	\textbf{End do}
\end{enumerate}
\noindent \textbf{Step 6} [Computing the Wannier function]\\
(Replace all $\dbtilde{\vct{u}}$ below by $\tilde{\vct{u}}$ if from \textbf{Step 3}.)
\begin{enumerate}
	\item \textbf{Do} $i=1,2,$\ldots$,n$
	\item[]Compute $\hat{\dbtilde{u}}_i$ in (\ref{eq:dbufour}) by FFT.
	\item[] \textbf{End do}
	\item Compute $W_0$ by (\ref{eq:nwann}).
\end{enumerate}
\textbf{Return}  $\cb{E^{(j_1,j_2)}}_{j_1,j_2=-N/2}^{N/2-1}$\,, $\cb{\dbtilde{\vct{u}}^{(j_1,j_2)}}_{j_1,j_2=-N/2}^{N/2-1}$ and $W_0$\,.

\vspace{0.5cm}
\noindent \underline{\textbf{Algorithm 1A}} \\
\noindent \textit{Input:} $j\in\cb{1,2}\,,\kappa_1,\kappa_2 \in [0,1], \vct{k}=\vct{k}(\kappa_1,\kappa_2)$\,, $H(\vct{k})\,, \frac{\partial}{\partial \kappa_j}H(\vct{k})  \in \mathbb{C}^{n\times n} $,  $\vct{u}(\vct{k}) \in \mathbb{C}^{n}$, $E(\vct{k}) \in \mathbb{R}$

\noindent \textit{Output:} $\frac{\partial}{\partial \kappa_j}\vct{u}(\vct{k}) \in \mathbb{C}^{n}$, $\frac{\partial}{\partial \kappa_j}E(\vct{k}) \in \mathbb{R}$
\begin{enumerate}
\item Compute $\frac{\partial}{\partial \kappa_j}E(\vct{k}) = \vct{u}^*(\vct{k}) \frac{\partial H(\vct{k}) }{\partial \kappa_j}\vct{u}(\vct{k})$.
\item Compute $\vct{q} =  \frac{\partial H(\vct{k}) }{\partial \kappa_j} \vct{u}(\vct{k})$.
\item Compute the singular value decomposition of $H(\vct{k}) - E(\vct{k}) = \mathbf{U} \mathbf{\Sigma} \mathbf{V}^*$ using the QR algorithm.
\item Set $\frac{\partial}{\partial \kappa_j}\vct{u}(\vct{k}) = \mathbf{V}_{1:n,1:n-1} (\mathbf{\Sigma}_{1:n-1,1:n-1})^{-1} (\mathbf{U}_{1:n,1:n-1})^* \vct{q}$.
\end{enumerate}

\section{Numerical results \label{sec:results}}

This section contains numerical results of the algorithms described in Section\,\ref{sec:numerics} applied to a $3\times 3$ matrix on a square lattice in Section\,\ref{sec:eg1} and the Haldane model  \cite{haldane1988model} ($2\times 2$ matrix on a hexagonal lattice), where a topologically trivial version ($C_1 = 0$) is in Section\,\ref{sec:trivial} and a non-trivial version ($C_1 \ne 0$) in Section\,\ref{sec:ntrivial}.

In all examples, we increase $N$, the number of discretization points in both dimensions (see Section\,\ref{sec:init}), while keeping all other parameters the same. We report the accuracy of the eigenvectors computed by numerically solving the parallel transport equation in Step 1--3 in Section\,\ref{sec:numerics}, measured by the largest difference of the projectors in terms of the Frobenius norm 
\begin{align}
	\mathrm{E_{evec}} = \max_{-N/2\le i,j\le N/2-1} \norm{P_{\rm para}^{(i,j)} - P_{\rm eig}^{(i,j)}}_{\rm F}\,,
\end{align}
where $P_{\rm para}^{(i,j)} $ and $P_{\rm eig}^{(i,j)} $ are the projectors formed by eigenvectors computed by solving the ODE and a QR eigenvalue solver, respectively, at $(\kappa_i, \kappa_j)$. Besides, we also report the computed error in the first Chern number by the formula
\begin{align}
 		\mathrm{E}_{\rm Ch} = 	\abs{\tilde{C}_1 - C_1}\,,
\end{align}
where $C_1$ is the exact integer value and $\tilde{C}_1$ is the computed value by (\ref{eq:nchern}) without rounding to an integer:
\begin{align}
	\tilde{C}_1 =  \Re{\frac{h}{2\pi\I} \sum_{j_1 = -N/2}^{N/2-1} \frac{z'^{(j_1)}}{{z^{(j_1)}}}}\,.
	\label{eq:nchernn}
\end{align}
Moreover, we report the largest value of the potential $\tilde{\psi}$ for the divergence (curl-free component) of the Berry connection of  $\dbtilde{\vct{u}}^{(j_1,j_2)}$ after Step 5 (Section\,\ref{sec:ndiv}) as a measure on its accuracy and optimality. We denote the largest value of the potential $\tilde{\psi}$ for the curl-free component by the formula
\begin{align}
		\mathrm{E_{div}} = \max_{-N/2\le i,j\le N/2-1} \abs{\tilde{\psi}^{(i,j)}}\,.
\end{align}

Furthermore, for Example 1 and 2, we report the Wannier center (\ref{eq:mean2}) and variance (\ref{eq:varr}) (computed according to Remark\,\ref{rmk:var}) before and after eliminating the divergence, the time $t_{\rm para}$ for solving the parallel transport equation (Step 1--3 in Section\,\ref{sec:numerics}), and the time $t_{\rm div}$ for eliminating the divergence for achieving the optimal Wannier function (Step 4--5 in Section\,\ref{sec:numerics}). 

In Example 1, we also test the parallel transport scheme in Remark\,\ref{rmk:twist}; instead of solving the ODE explicitly, we replace Step 1--3  in Section\,\ref{sec:numerics} with the approach in Remark\,\ref{rmk:twist}. We report the maximum difference $\mathrm{E_{para}}$ in terms of the phase between the two approaches, given by the formula
\begin{align}
	\mathrm{E_{para}} =   \max_{-N/2\le i,j\le N/2-1}  \| {{\tilde{\vct{u}}^{(j_1,j_2)}  - \tilde{\vct{u}}_{\rm twist}^{(j_1,j_2)}}} \|\,,
	\label{eq:paraerr}
\end{align}
where $\tilde{\vct{u}}^{(j_1,j_2)} $ is obtained by solving the parallel transport equation as described in Section\,\ref{sec:numerics} and $\tilde{\vct{u}}_{\rm twist}^{(j_1,j_2)} $  is computed by first computing the eigenvalues and eigenvectors by QR, followed by the alignment scheme in Remark\,\ref{rmk:twist}. For the other two examples, all results are obtained by the approach in Section\,\ref{sec:numerics}.

All algorithms are implemented in MATLAB R2023b. All timing experiments are performed on a MacBook Pro with an M2 Max CPU.

\subsection{Example 1: A $3\times 3$ matrix model \label{sec:eg1}}

The $3\times 3$ matrix $H$ in this example is taken from Example 2.6 in \cite{kaxiras2019quantum}. The model is defined on a square lattice with the real space lattice spanned by $\vct{a}_1 =\frac{a}{\sqrt{2}} (1,-1)$ and $\vct{a}_2 =\frac{a}{\sqrt{2}} (1,1)$\,, and the reciprocal lattice by $\vct{b}_1 = \frac{\sqrt{2}\pi}{a} (1,-1)$ and $\vct{b}_2 = \frac{\sqrt{2}\pi}{a} (1,1)$. We choose the constant $a=1$ here. The elements in the matrix $H$ are given by the formulas
\begin{align}
	\begin{split}
	H_{11} & = H_{22}  = \epsilon_p+2 t_{p p}\left[\cos \left(\left(k_x-k_y\right) a / \sqrt{2}\right)+\cos \left(\left(k_x+k_y\right) a / \sqrt{2}\right)\right]\,,\\
	H_{33} &  = \epsilon_d+2 t_{d d}\left[\cos \left(\left(k_x-k_y\right) a / \sqrt{2}\right)+\cos \left(\left(k_x+k_y\right) a / \sqrt{2}\right)\right]\,,\\
	H_{13} &= \bar{H}_{31} = -t_{p d} \mathrm{e}^{\mathrm{i} k_x a / \sqrt{2}} 2 \mathrm{i} \sin \left(k_x a / \sqrt{2}\right)\,,\\
	H_{23} & = \bar{H}_{32} = t_{p d} \mathrm{e}^{\mathrm{i} k_x a / \sqrt{2}} 2 \mathrm{i} \sin \left(k_y a / \sqrt{2}\right)\,,\\
	H_{12} & = \bar{H}_{21} = 0\,,
	\end{split}
\end{align}
where the constants are chosen as $t_{dd}=0.1$, $t_{pd}=2$, $t_{pp} = -0.25$,  $t_{pd} = 2$, $\epsilon_d =1 $ and $\epsilon_p=-2$. The eigenvalues parameterized by $\kappa_1$ and $\kappa_2$ are shown in Figure\,\ref{fig:eg1eval}, and the top non-degenerate band is chosen for the construction. This model has time-reversal symmetry and its constructed $\varphi_2$ (see (\ref{eq:phitopo})) is shown in Figure\,\ref{fig:eg1phi2}, which is an even function as proved in Lemma\,\ref{lem:phi}. The first Chern number $C_1 = 0$ and there is no topological obstruction.

Table\,\ref{tab:eg1} shows the timings and errors of the approach in  Section\,\ref{sec:numerics}. The computation time $t_{\rm para}$ scales as $O(N^2)$ as discussed in Section\,\ref{sec:npara2} and the time for eliminating the divergence is dominated by FFTs.
It also shows that the error $\mathrm{E}_{\rm evec}$ of computed eigenvectors by solving the parallel transport equation scales as $O(h^6)$ with $h=1/N$. The error in the divergence $\mathrm{E}_{\rm div}$ shows that, for $N\le 200$, the sampling over $D^*$ is not sufficient to resolve the frequency content of the Berry connection. As a result, for achieving 10-digit accuracy, the optimal choice is roughly $N=200$. 

The three components of the assignment $\tilde{\vct{u}}$ computed by the approach in Section\,\ref{sec:numerics} after Stage 2 (see Section\,\ref{sec:stage2}) is shown in Figure\,\ref{fig:eg1tu1}--\ref{fig:eg1tu3}, together with the absolute value of their Fourier coefficients in the $\log_{10}$ scale. The Fourier coefficients decay exponentially asymptotically.
Their real and imaginary part are even and odd, respectively, under the transform $\kappa_1 \rightarrow -\kappa_1$ and $\kappa_2 \rightarrow -\kappa_2$, as proved in Theorem\,\ref{thm:real}\,. After eliminating the divergence of the Berry connection of $\tilde{\vct{u}}$ in Stage 3 in Section\,\ref{sec:stage3}, the third component of $\dbtilde{\vct{u}}$ is shown in Figure\,\ref{fig:eg1ttu3}. The other components are not shown as they are visually very similar to those of $\tilde{\vct{u}}$ in  Figure\,\ref{fig:eg1tu1}--\ref{fig:eg1tu2}. The Berry connection of $\tilde{\vct{u}}$ is shown in Figure\,\ref{fig:eg1bcdiv}, whose Helmholtz-Hodge decomposition is shown in Figure\,\ref{fig:eg1divf}. After eliminating its divergence, the Berry connection of $\dbtilde{\vct{u}}$ is shown in Figure\,\ref{fig:eg1bcdivless}.

In comparison to Table\,\ref{tab:eg1}, the error $\mathrm{E}_{\rm para}$ in Table\,\ref{tab:eg1twist} indicates the accuracy of  the parallel transport scheme in Remark\,\ref{rmk:twist}; the error is of order $h^2$, so the scheme is a second-order one and 
produces reasonably accurate assignments. 
The advantage of this approach is that it decouples the computation of eigenvectors and Wannier functions, and it can be done with very little computational cost. It is a viable approach when high accuracy is not required.

Table\,\ref{tab:var1} contains the Wannier center and variance of the solution obtained by parallel transport described in Section\,\ref{sec:numerics} and the optimal one after the divergence of the Berry connection is eliminated. The Wannier center is not changed and the variance is reduced. The solution from only doing parallel transport, although not optimal, is very close to the optimal one in terms of the variance.  


\begin{figure}[h]
	\centering	
	\subfigure[]{
		\includegraphics[scale=0.35]{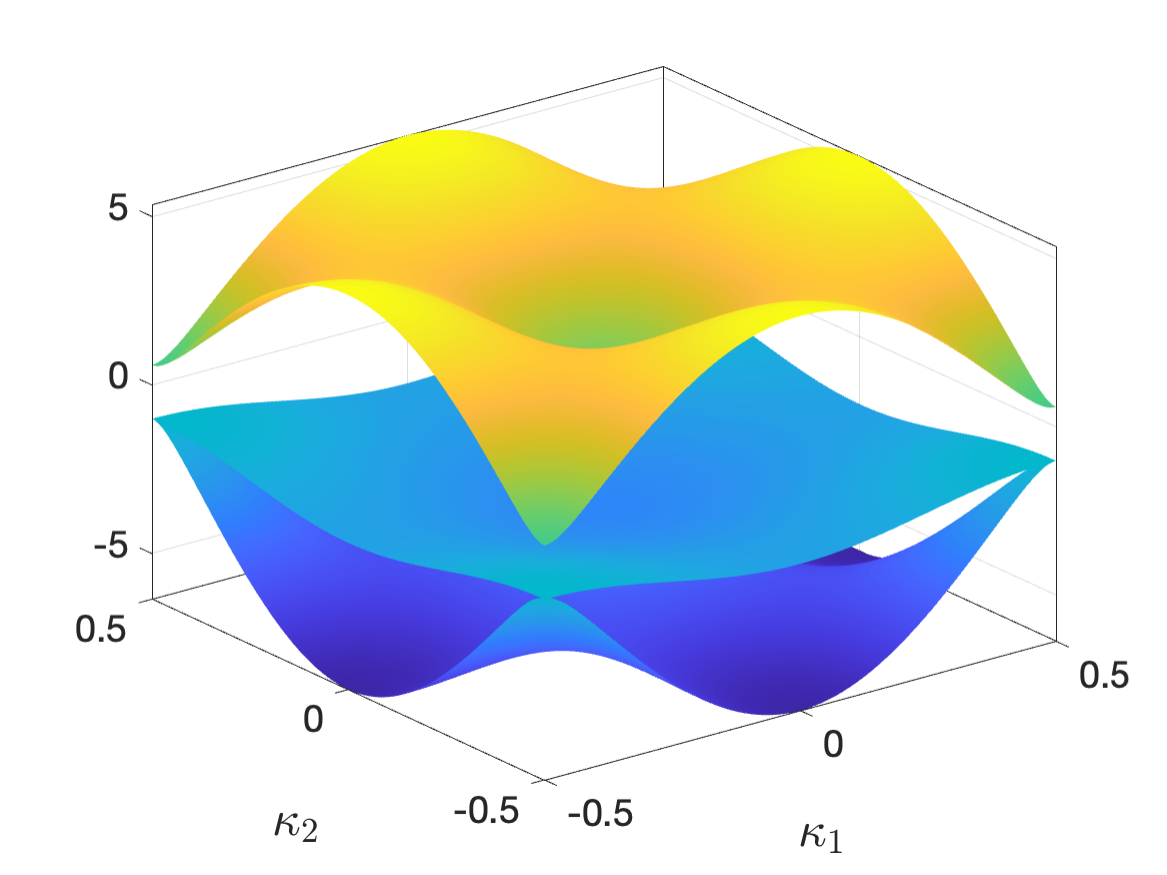}
			\label{fig:eg1eval}
	}
	\subfigure[]{
	\includegraphics[scale=0.35]{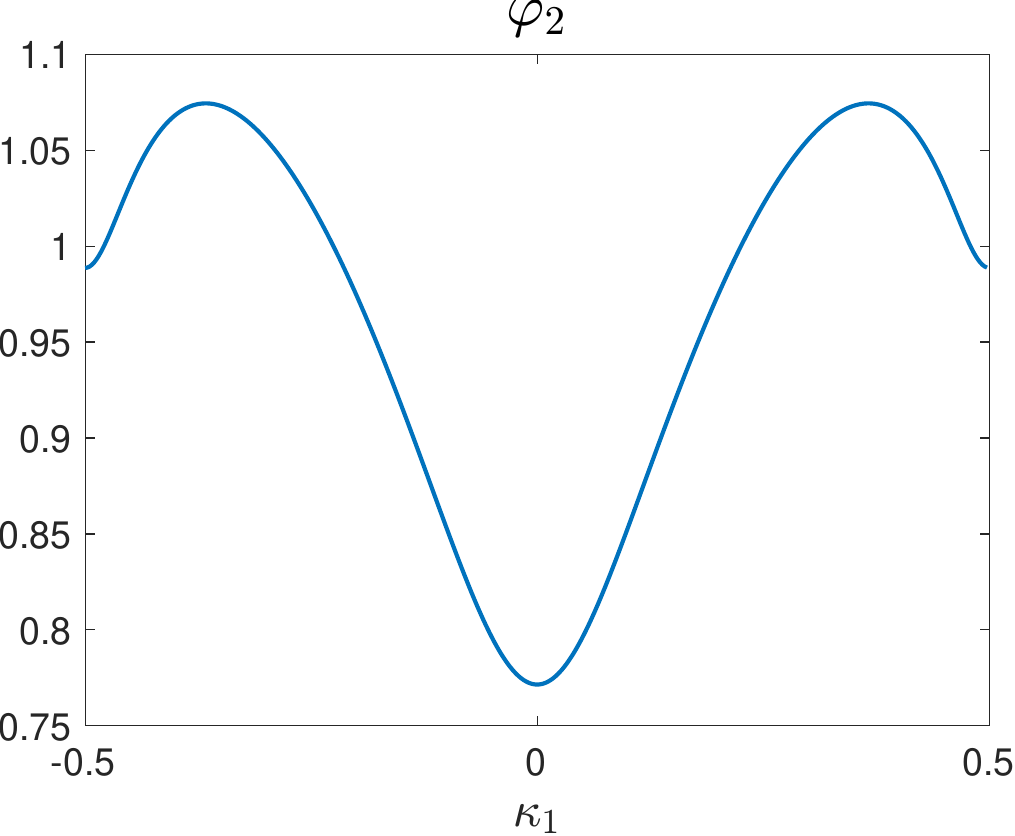}
	\label{fig:eg1phi2}
	}
\vspace{-1em}
\caption{ (a) Plot of eigenvalues of $H$ in Example 1. The top band (non-degenerate) is picked. (b) The phase $\varphi_2$ in (\ref{eq:phitopo}) for Example 1. It is an even function as proved in Lemma\,\ref{lem:phi}.}

\end{figure}

\begin{table}[h!]
	\centering
	\begin{tabular}{c c c c c c}
		\hline
		$N$ & $t_{\rm para}$ (s)& $t_{\rm div}$ (s)& $\mathrm{E_{evec}}$  & $\mathrm{E_{Ch}}$ & $\mathrm{E_{div}}$ \\
		\hline \hline
		50 & 1.76    & 0.053 & 4.16e-10& 6.84e-10 & 1.41e-3\\
		100 & 6.12  & 0.110 & 4.18e-10 & 3.30e-16 & 3.38e-6\\
		200 & 23.2   & 0.305& 6.26e-12 & 2.21e-19 & 3.07e-11\\
		400 & 93.5  & 1.10 & 9.93e-14 & 4.49e-18 & 7.49e-12\\
		\hline 
	\end{tabular}
	\caption{\label{tab:eg1}
		Timings and errors for Example 1 by the approach in Section\,\ref{sec:numerics}. The error $\mathrm{E}_{\rm evec}$ shows sixth-order convergence and the error $\mathrm{E}_{\rm div}$ indicates if the sampling over $D^*$ is sufficient.}
\end{table}

\begin{table}[h]
	\centering
	\begin{tabular}{c  c c c c  c}

		\hline
		$N$ & $t_{\rm para}$ (s)& $t_{\rm div}$ (s)& $\mathrm{E_{para}}$ & $\mathrm{E_{Ch}}$  & $\mathrm{E_{div}}$ 	   \\
		\hline \hline
		50 & 0.131    & 0.081 & 2.59e-3& 6.95e-10  &1.61e-3 \\
		100 & 0.200   & 0.145 & 6.48e-4 & 3.18e-16 &  3.487e-6\\
		200 & 0.488    & 0.375 & 1.62e-4&  2.61e-16 & 7.60e-12\\
		400 & 1.60  &1.32 & 4.05e-5& 2.04e-16  &6.81e-12 \\
		800 & 6.07 &5.13 & 1.01e-5 &   7.23e-19&3.27e-11\\
		\hline 
	\end{tabular}
	\caption{\label{tab:eg1twist}
		Timings and errors for Example 1 by the parallel transport scheme in Remark\,\ref{rmk:twist}. ($t_{\rm para}$ includes the time for obtaining the eigenvalues and eigenvectors.) The error $\mathrm{E_{para}}$ defined in (\ref{eq:paraerr}) measures the accuracy of the scheme and  $\mathrm{E_{para}}$  decreases proportional to $1/N^2$. However, such a parallel transport scheme is computationally cheap and decouples eigenvector and Wannier function computations.}
\end{table}

\begin{table}[h!]
	\centering
	\begin{tabular}{c c c |c c c}
		\multicolumn{3}{c}{After Stage 2} &
		\multicolumn{3}{c}{After Stage 3 (Optimal solution)} \\
		\hline
		$\langle \vct{R}_x \rangle$ & 		$\langle \vct{R}_y \rangle$ & 			$\langle \norm{\vct{R}}^2 \rangle$ - $\norm{\langle \vct{R} \rangle}^2$ & $\langle \vct{R}_x \rangle$ & 		$\langle \vct{R}_y \rangle$ & 			$\langle \norm{\vct{R}}^2 \rangle$ - $\norm{\langle \vct{R} \rangle}^2$ \\
		\hline\hline 
		-0.217677 &  -2.67e-15 & 0.317890  & -0.217677 & -2.23e-15 &	0.313797			\\
		\hline
	\end{tabular}
	\caption{The Wannier center and the variance for Example 1 computed by the approach in Section\,\ref{sec:numerics} before and after eliminating the divergence of the Berry connection computed for $N=400$. The solution after Stage 2 is already close to the optimal one. \label{tab:var1}}
\end{table}

\begin{figure}[h!]
	\centering
    \includegraphics[scale=0.25]{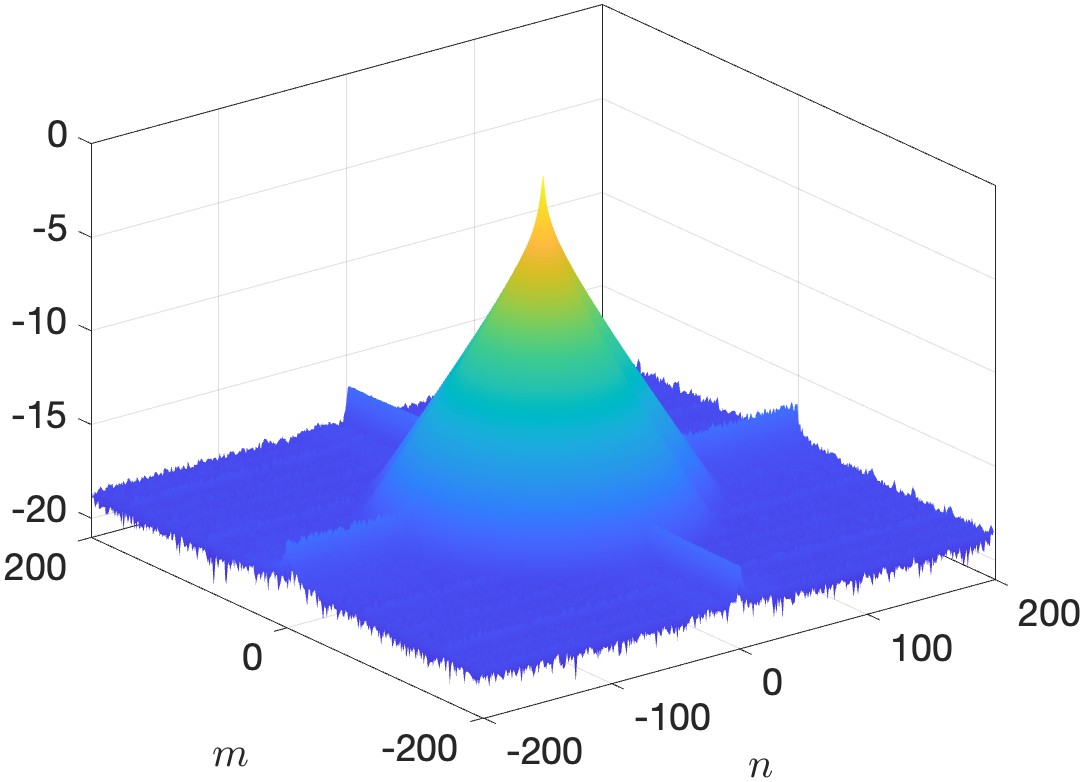}
	\includegraphics[scale=0.40]{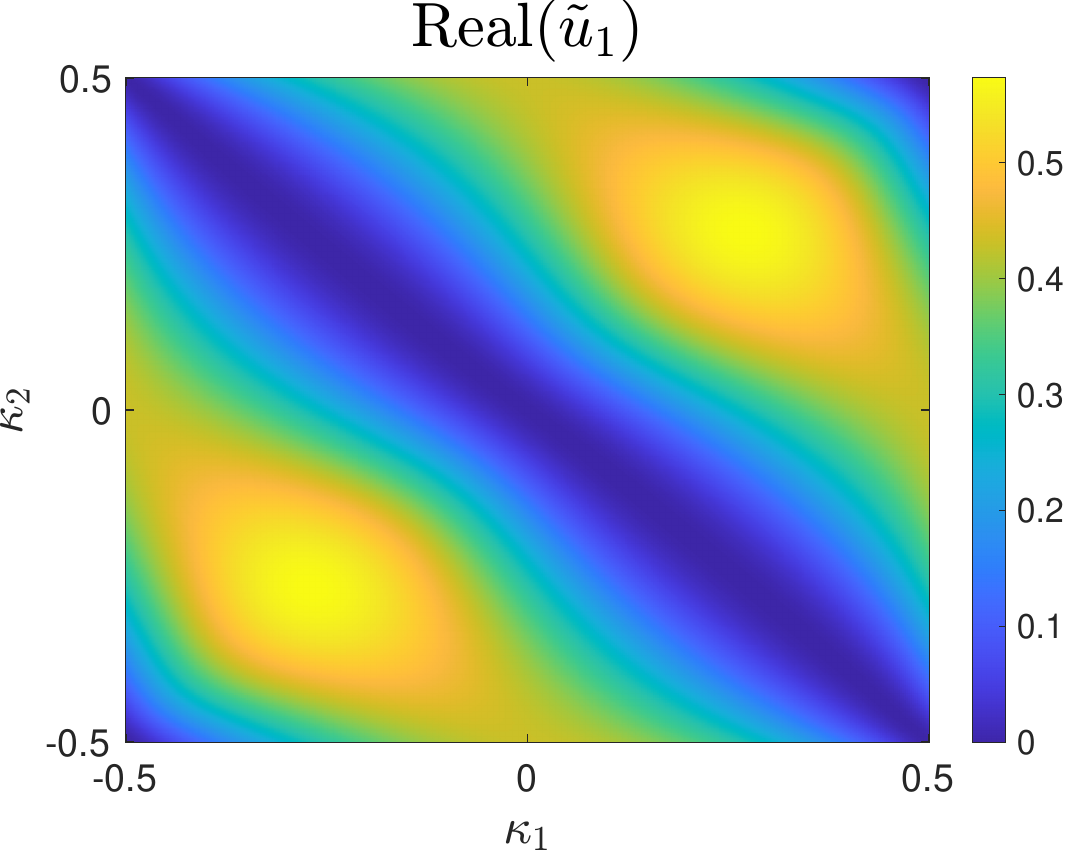}
	\includegraphics[scale=0.40]{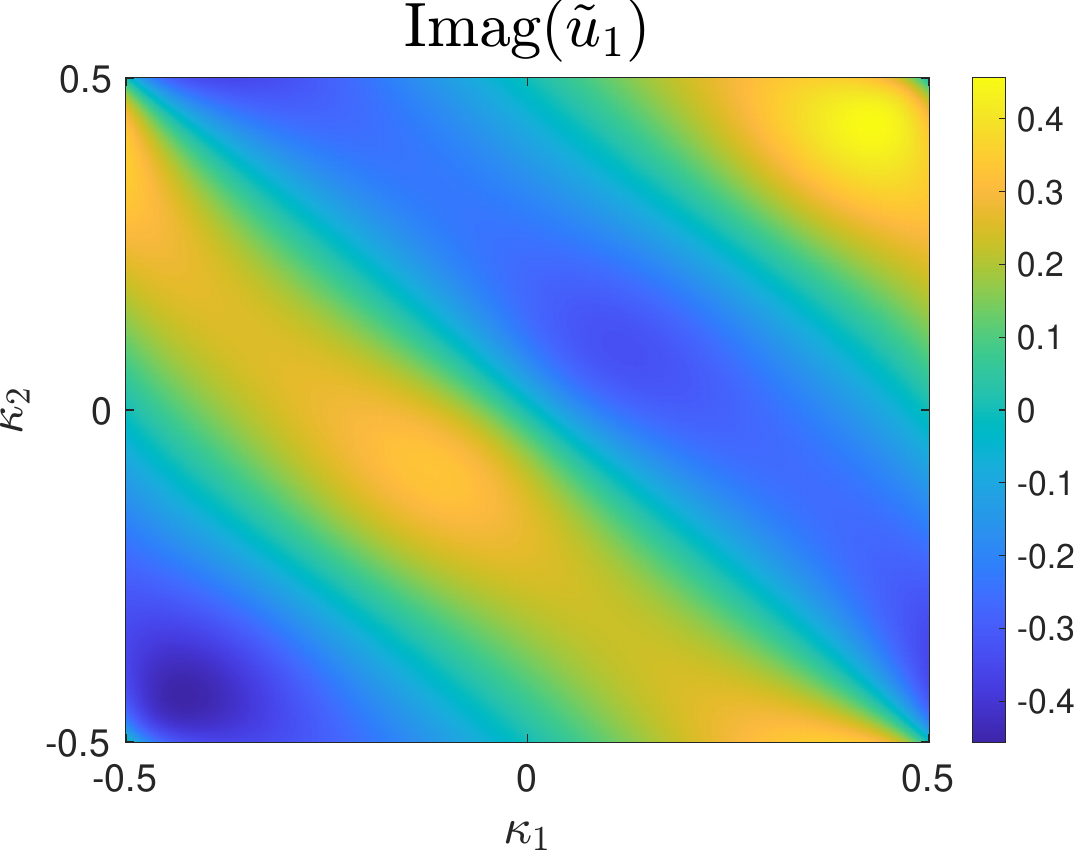}

	\vspace{-1em}
	\caption{ Plot of the real and imaginary  part of the component $\tilde{u}_1$ with the absolute value of the Fourier coefficients of $\tilde{u}_1$ in the  $\log_{10}$ scale  in Example 1. 
		}
	\label{fig:eg1tu1}
\end{figure}

\begin{figure}[h]
	\centering
	
	\includegraphics[scale=0.25]{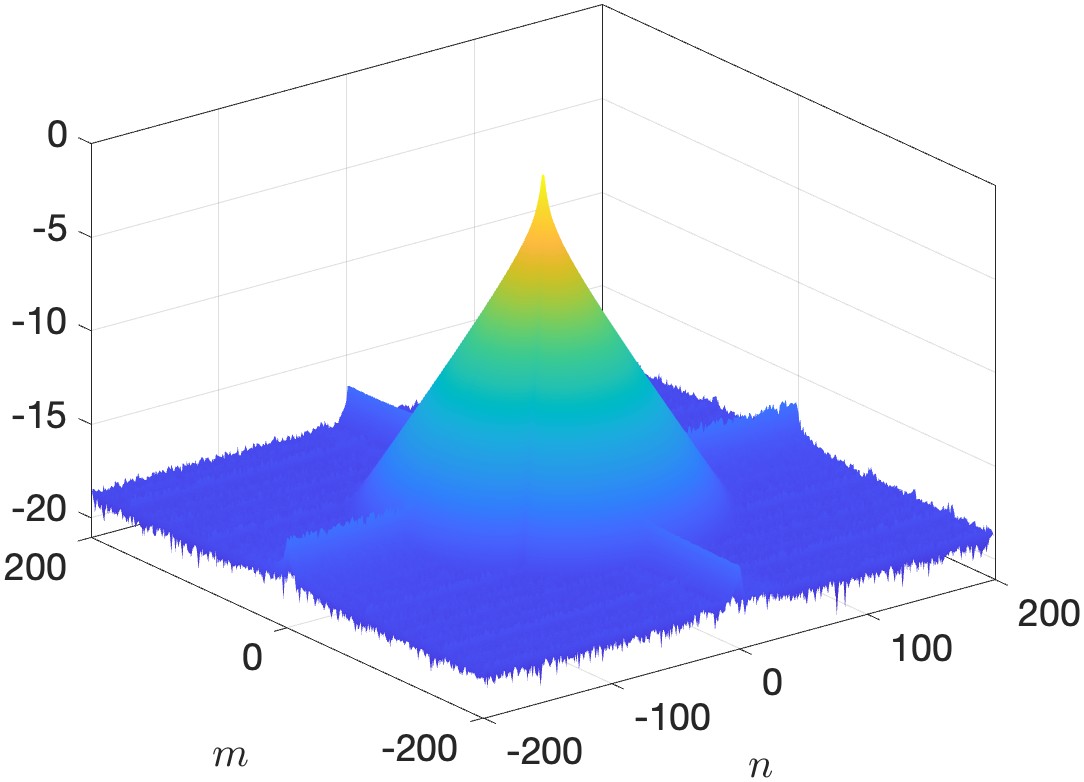}
	
	\includegraphics[scale=0.40]{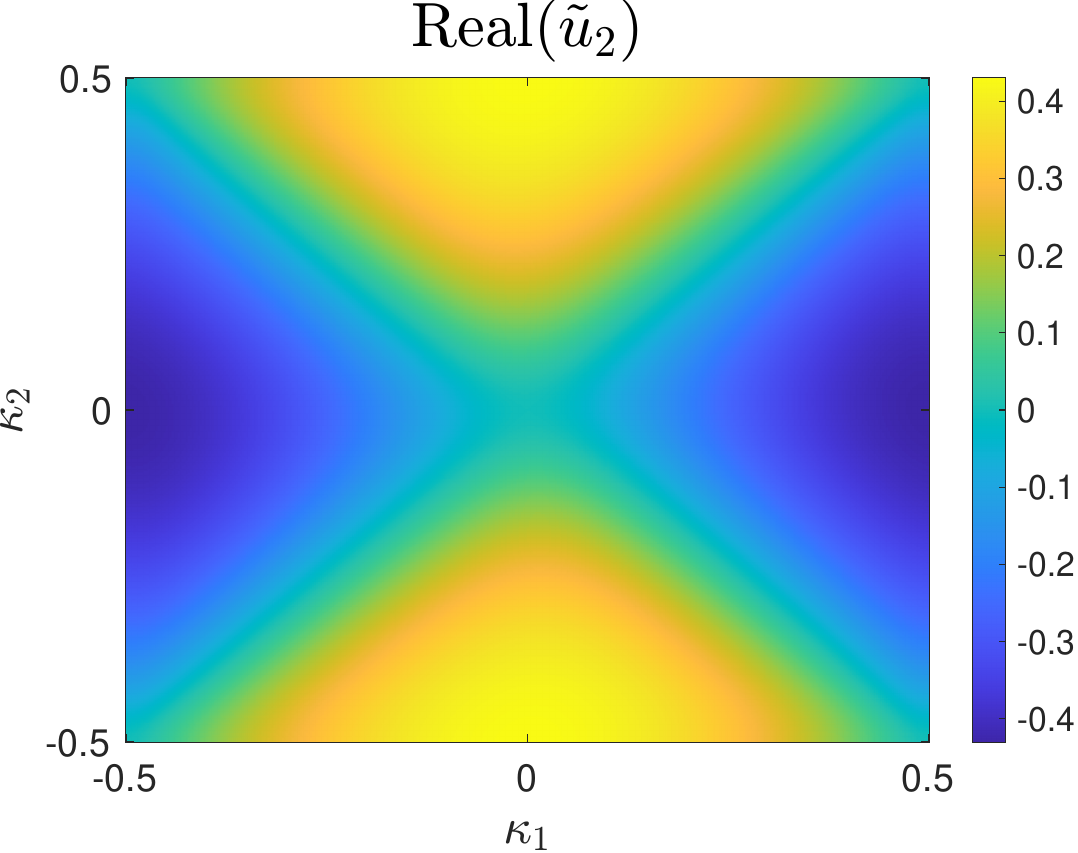}
	\includegraphics[scale=0.40]{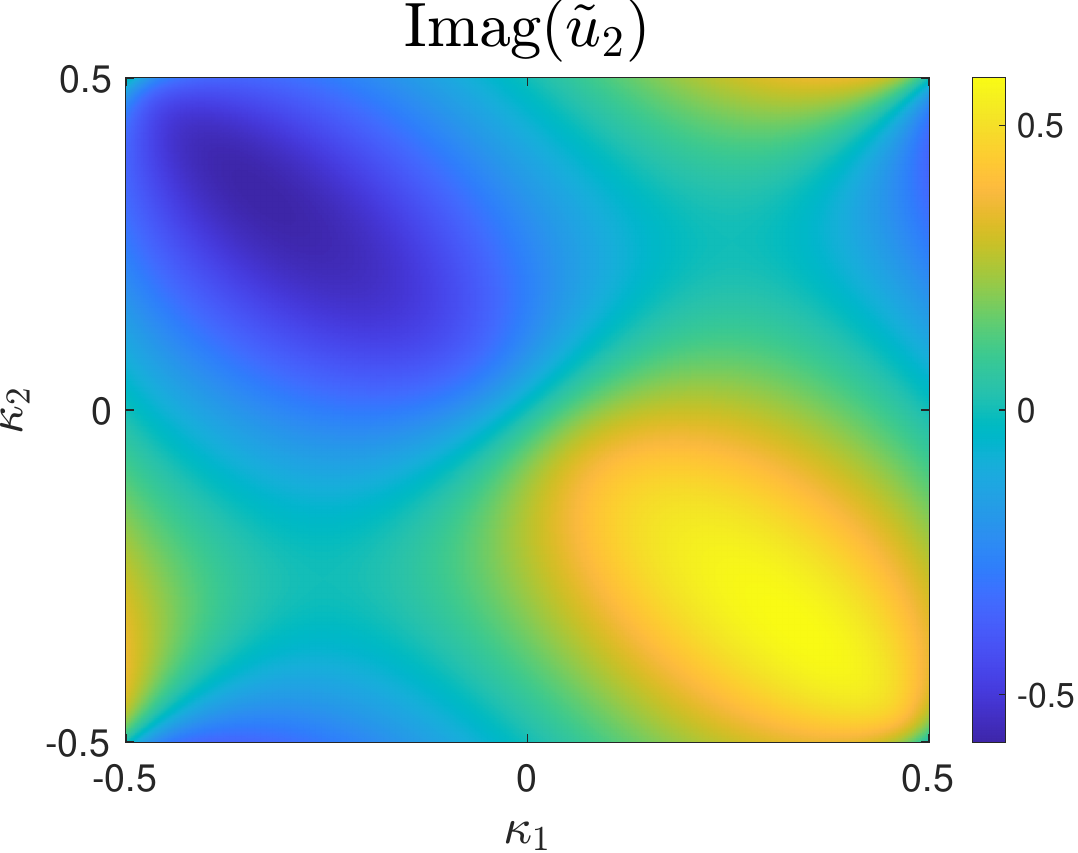}
	
			\vspace{-1em}
	\caption{ Plot of the real and imaginary  part of the component $\tilde{u}_2$ with the absolute value of the Fourier coefficients of $\tilde{u}_2$ in the  $\log_{10}$ scale in Example 1.  }
	\label{fig:eg1tu2}
	
\end{figure}

\begin{figure}[h!]
	\centering
	
	\includegraphics[scale=0.25]{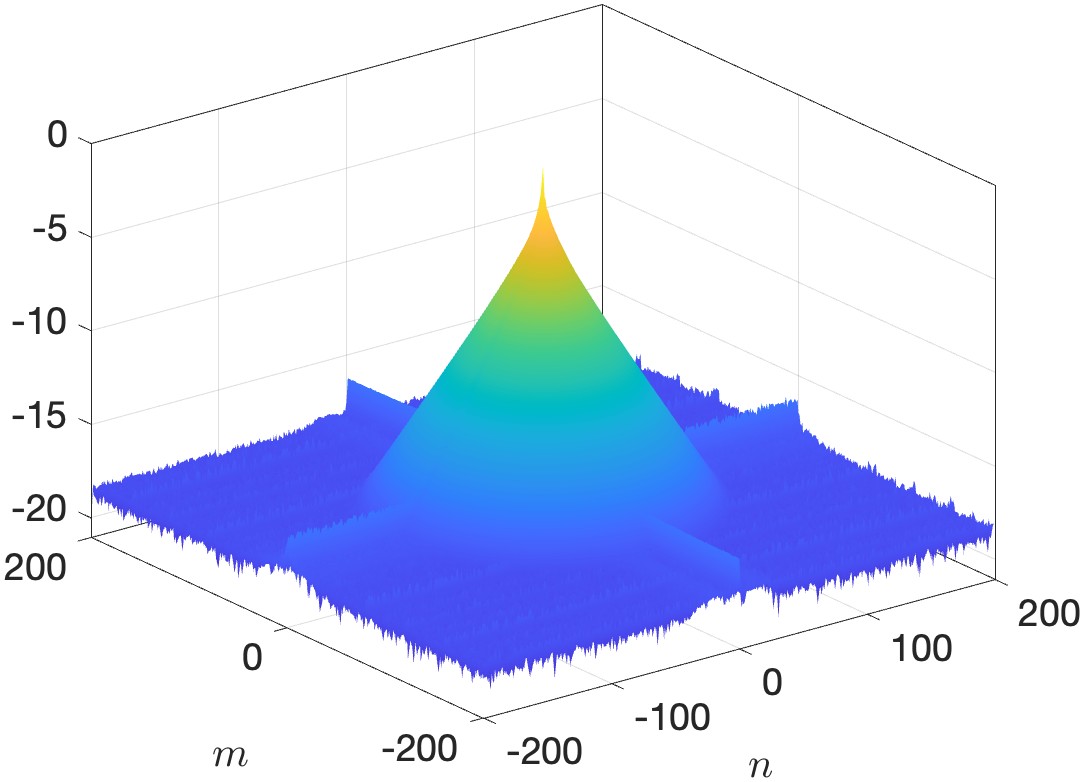}
	
	\includegraphics[scale=0.4]{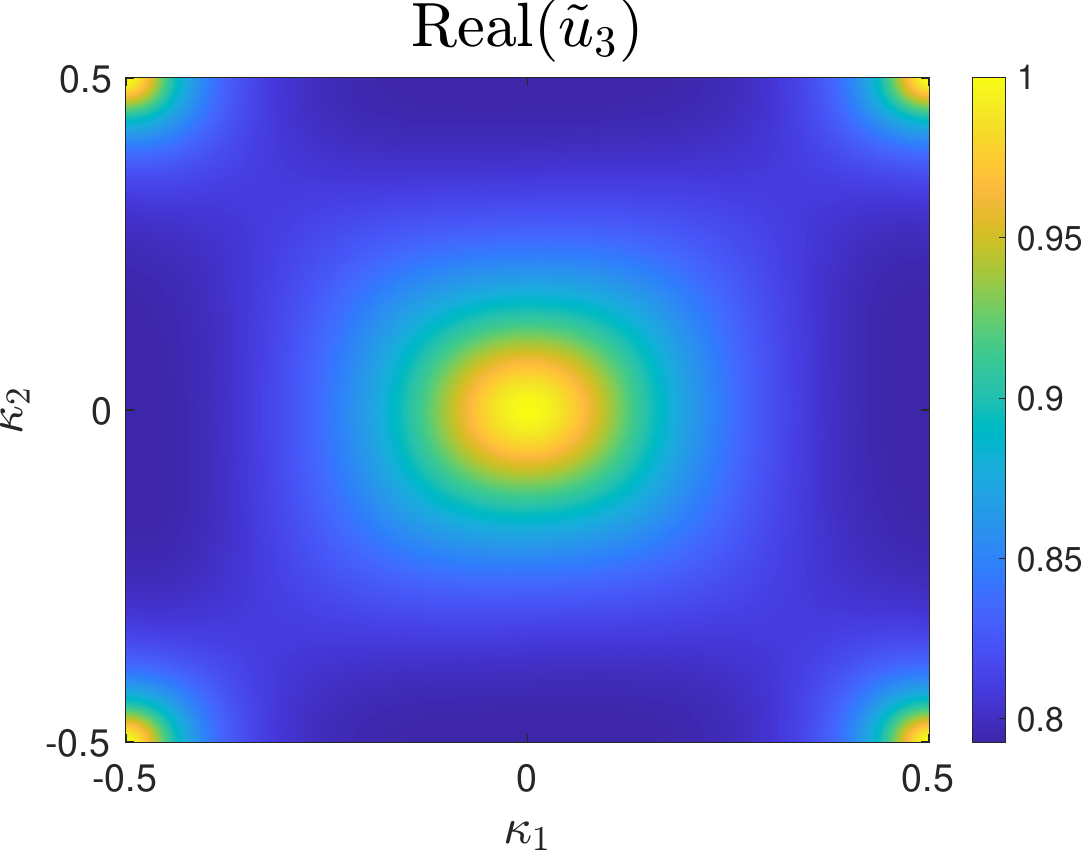}
	\includegraphics[scale=0.4]{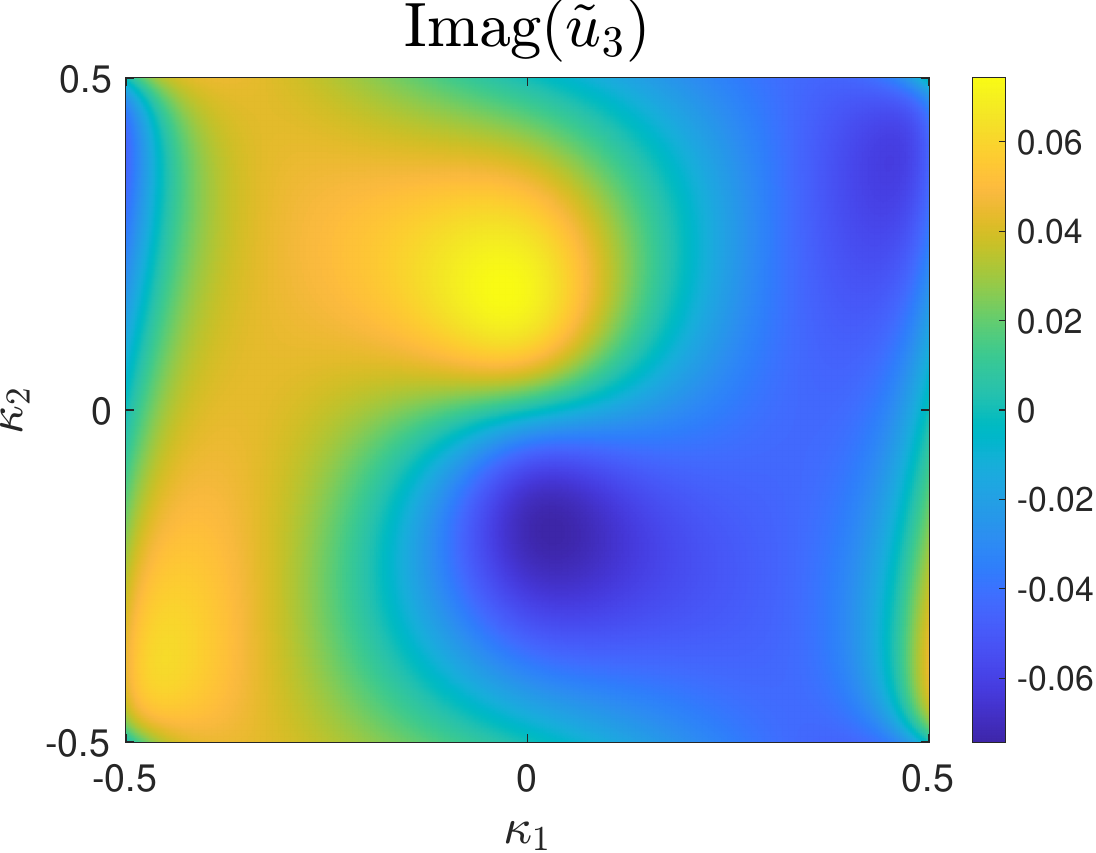}
	
		\vspace{-1em}
	\caption{ Plot of the real and imaginary  part of the component $\tilde{u}_3$ with the absolute value of the Fourier coefficients of $\tilde{u}_3$ in the  $\log_{10}$ scale in Example 1.  }
	\label{fig:eg1tu3}
	
\end{figure}

\begin{figure}[h!]
	\centering	
	\includegraphics[scale=0.25]{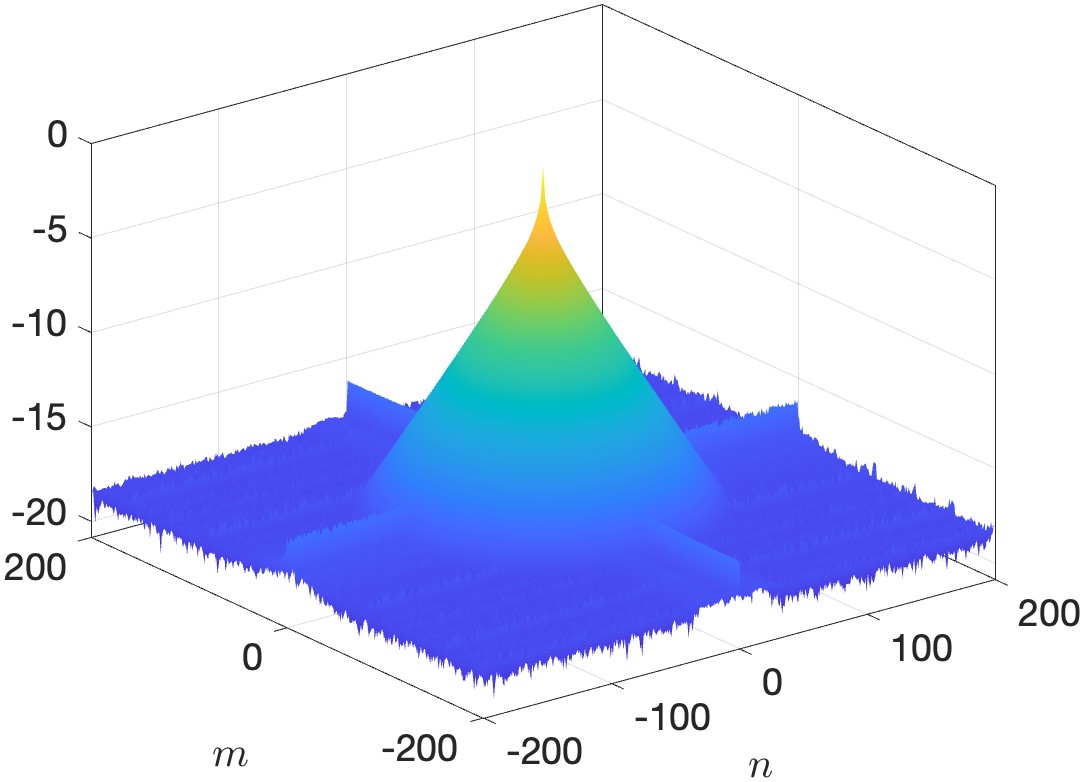}	
	
	\includegraphics[scale=0.4]{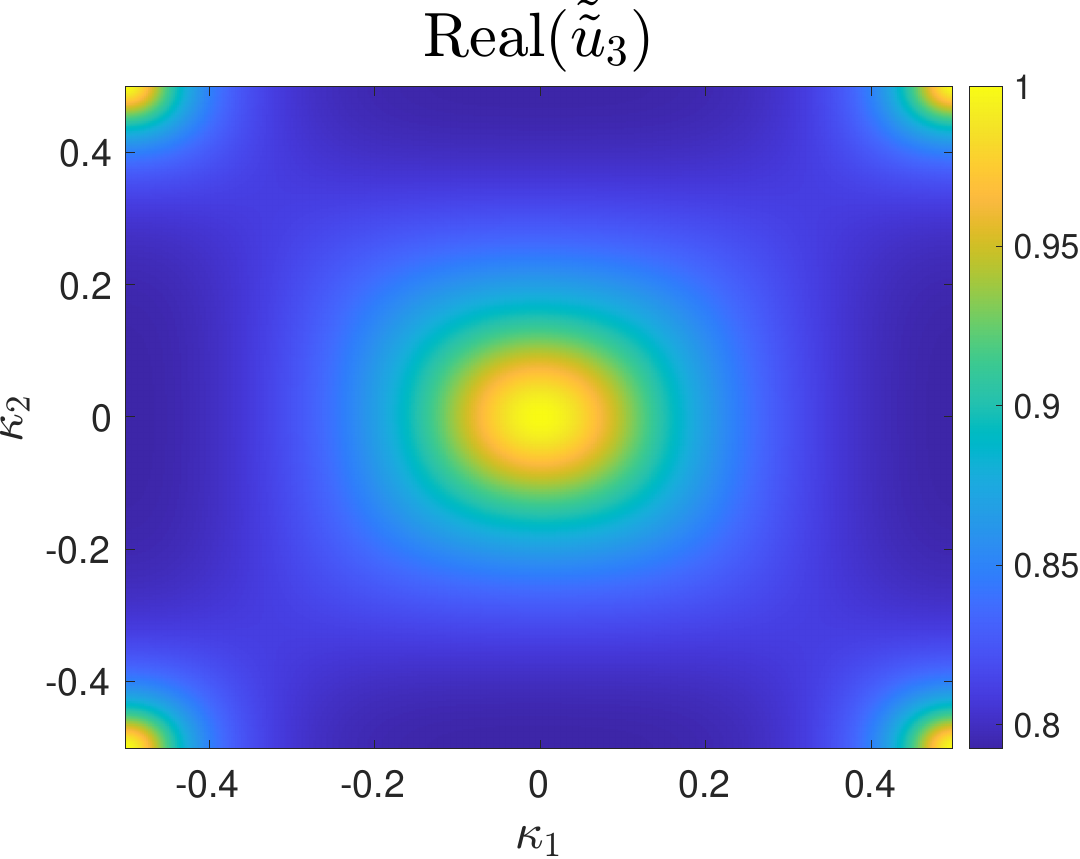}
	\includegraphics[scale=0.4]{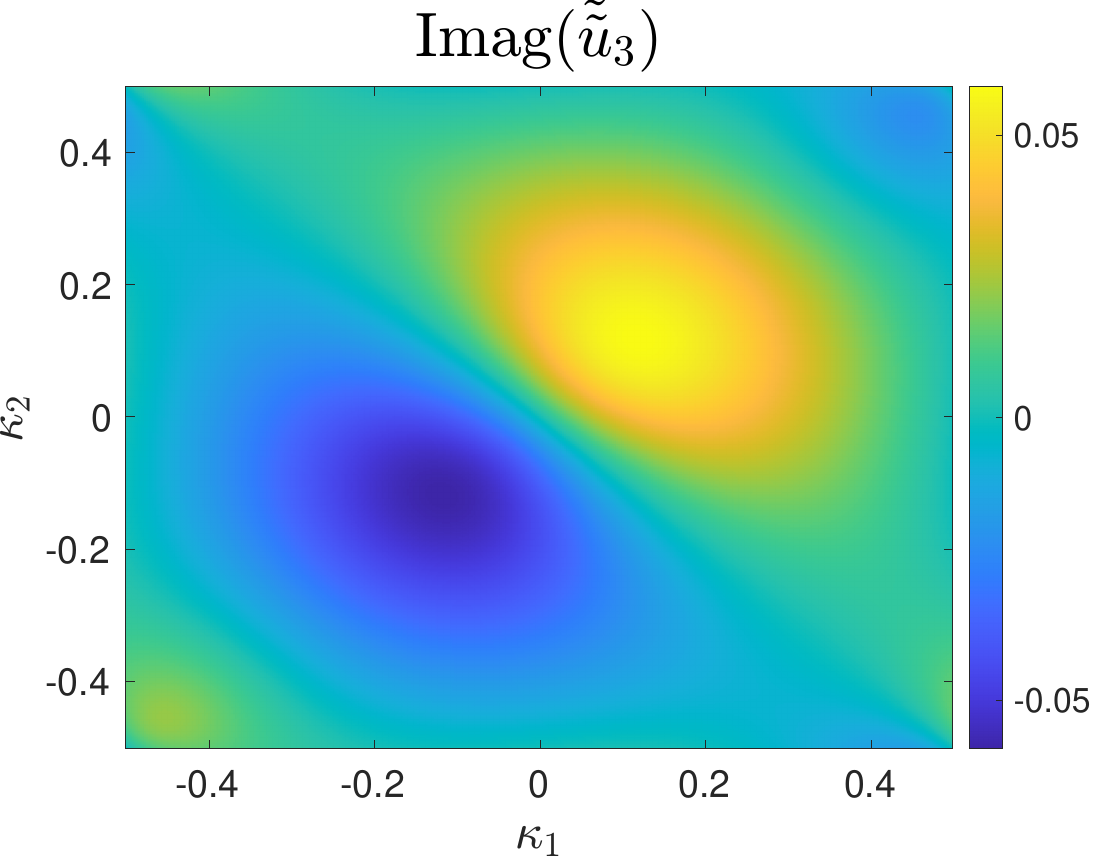}
	
	\vspace{-1em}
	\caption{ Plot of the same quantities as in Figure\,\ref{fig:eg1tu3}  after eliminating the divergence of the Berry connection.   }
\label{fig:eg1ttu3}
	
\end{figure}

\begin{figure}[h!]
	\centering	

	\subfigure[]{
		\includegraphics[scale=0.55]{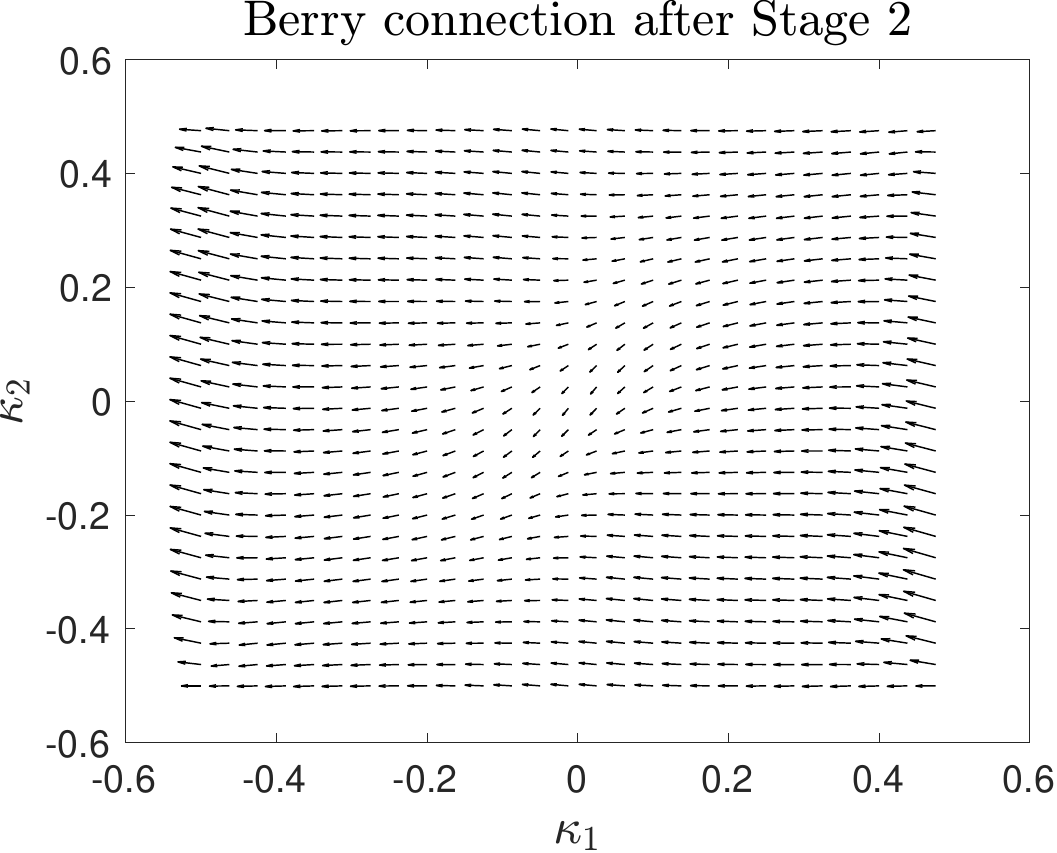}
		\label{fig:eg1bcdiv}
	}
	\subfigure[]{
		\includegraphics[scale=0.55]{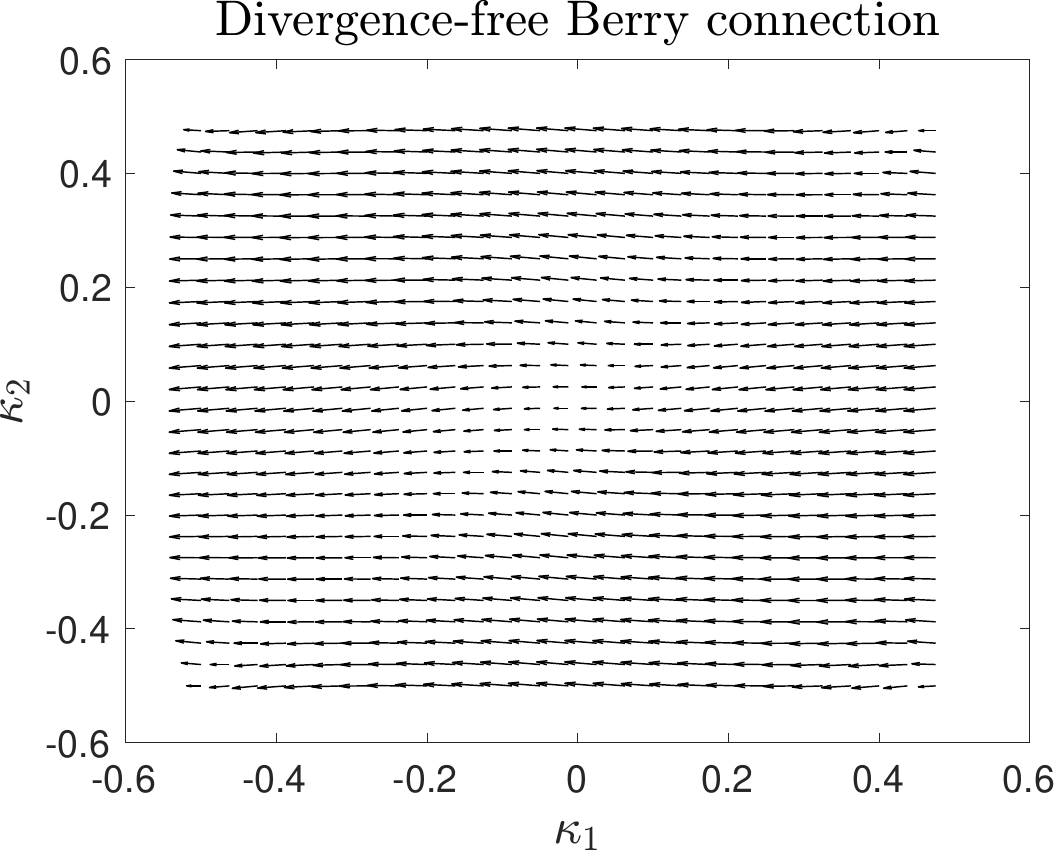}
		\label{fig:eg1bcdivless}
	}
	\vspace{-1em}
	\caption{ (a) Plot of the Berry connection after Stage 2 for Example 1. (b) Plot of the Berry connection in (a) after eliminating its divergence (curl-free component) in Stage 3. All vectors shown are in the $\vct{e}_x,\vct{e}_y$ basis. }

\end{figure}

\begin{figure}[h]
	\centering	
	
		\includegraphics[scale=0.44]{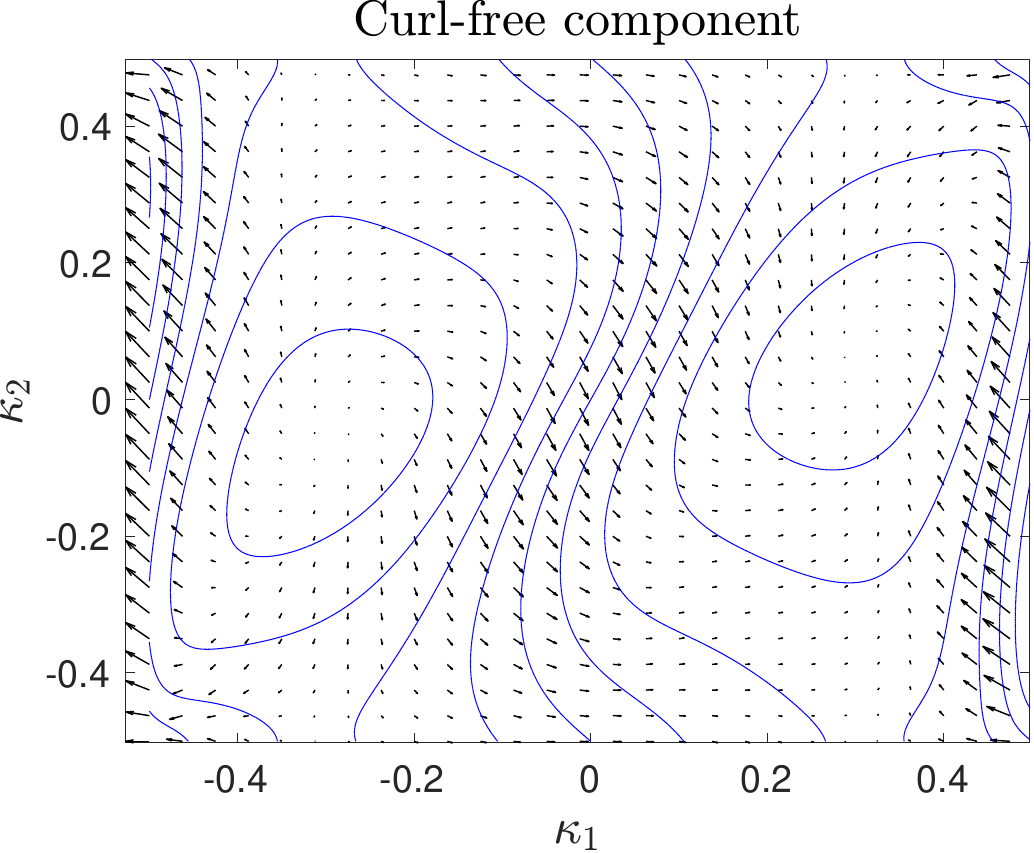}
		
		\includegraphics[scale=0.44]{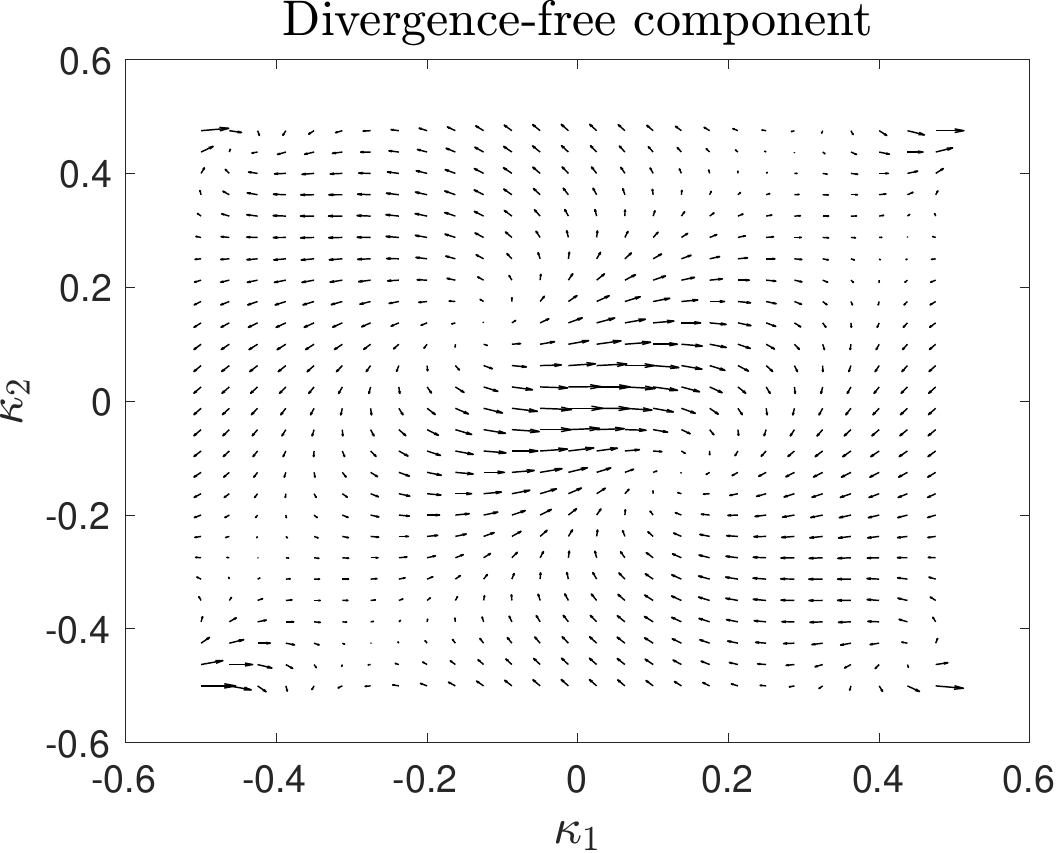}
		
		\includegraphics[scale=0.44]{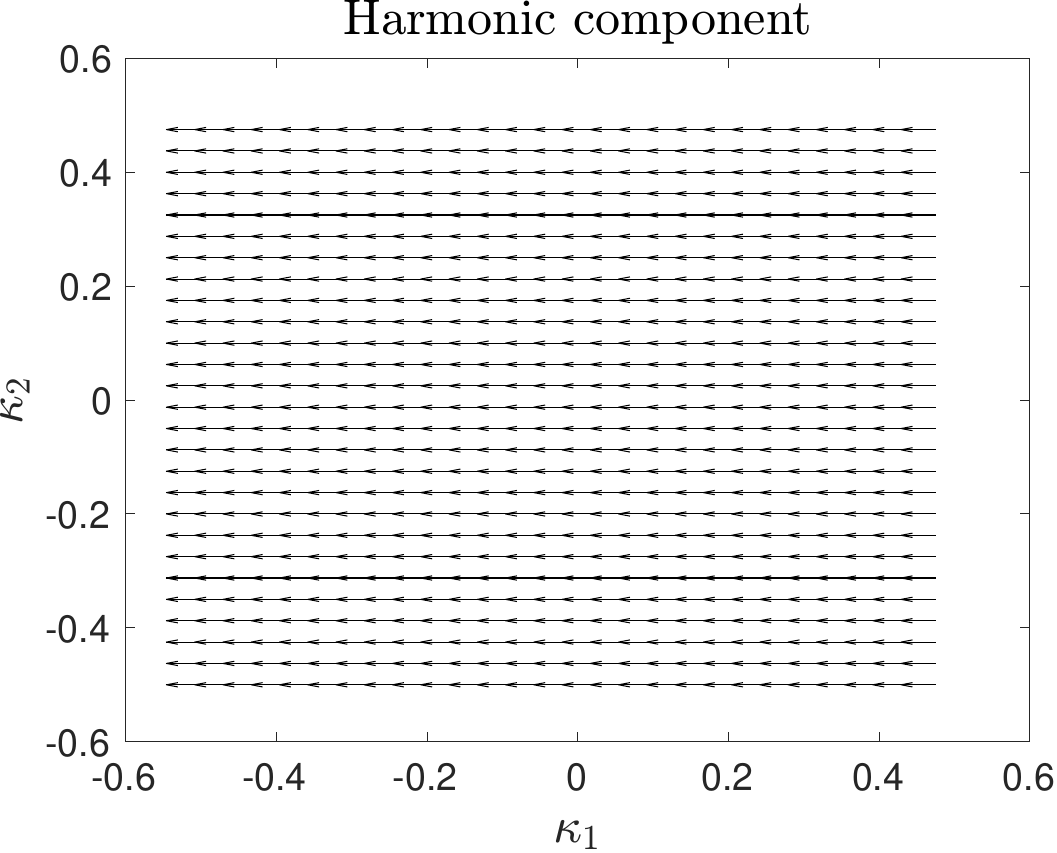}

	\vspace{-1em}
	\caption{ \label{fig:eg1divf} The Helmholtz-Hodge decomposition of the Berry connection in Figure\,\ref{fig:eg1bcdiv}. All vectors shown are in the $\vct{e}_x,\vct{e}_y$ basis. The equipotential lines are shown for the curl-free vector field, which is eliminated to obtain the optimal solution in Figure\,\ref{fig:eg1bcdivless}.}

\end{figure}

\clearpage
\subsection{Example 2: Haldane model (topologically trivial case) \label{sec:trivial}}

In this example, we consider a topologically trivial version of the Haldane model \cite{haldane1988model} on a Honeycomb lattice. It is $2\times 2$ matrix model with time-reversal symmetry, thus topologically trivial (i.e. $C_1 =0$)\,. The real space lattice vectors are given by $\vct{a}_1 = \frac{a}{2} (\sqrt{3},1)$ and $\vct{a}_2 = \frac{a}{2} (\sqrt{3},-1)$ with the reciprocal lattice vectors $\vct{b}_1 = \frac{2\pi}{\sqrt{3} a }(1,\sqrt{3})$ and $\vct{b}_2 = \frac{2\pi}{\sqrt{3} a }(1,-\sqrt{3})$\,. We choose $a=1$ here. The matrix $H$ is given by 
\begin{align}
	H(\vct{k}) =  \sb{\mat{V_0&  t_1(1 + e^{-\I\vct{k}\cdot\vct{a}_1} + e^{-\I\vct{k}\cdot\vct{a}_2})\\ 
						  t_1(1 + e^{\I\vct{k}\cdot\vct{a}_1} + e^{\I\vct{k}\cdot\vct{a}_2})					& -V_0 }}\,,
\end{align}
where the constants are chosen as $t_1 = 1$ and  $V_0 = 0.5$. The eigenvalues as a function of the parameterization $\kappa_1$ and $\kappa_2$ are shown in Figure\,\ref{fig:eg2eval}, and the top band is chosen for the construction. This model has time-reversal symmetry and its constructed $\varphi_2$ (see (\ref{eq:phitopo})) is shown in Figure\,\ref{fig:eg2phi2}, which is an even function as proved in Lemma\,\ref{lem:phi}.

Table\,\ref{tab:eg2} shows the timings and errors listed at the beginning of Section\,\ref{sec:results}. The overall trends are very similar to those in Example 1. The error in the divergence $\mathrm{E}_{\rm div}$ shows that, for $N\le 200$, the sampling over $D^*$ is not sufficient to resolve the frequency content of the Berry connection. As a result, for achieving 10-digit accuracy, the optimal choice is roughly $N=200$.

Similar to Example 1, Table\,\ref{tab:var2} shows that the Wannier center and variance of the solution obtained by parallel transport and the optimal one after the divergence of the Berry connection is eliminated. The Wannier center is not changed and the variance is reduced. The solution from parallel transport, although not optimal, is also close to the optimal one in terms of the variance.  

The two components of the assignment $\tilde{\vct{u}}$ after Stage 2 in Section\,\ref{sec:stage2} is shown in Figure\,\ref{fig:eg2tu1}--\ref{fig:eg2tu2}, together with the absolute value of their Fourier coefficients in the  $\log_{10}$ scale. The Fourier coefficients decay exponentially asymptotically.
Their real and imaginary part are even and odd, respectively, under the transform $\kappa_1 \rightarrow -\kappa_1$ and $\kappa_2 \rightarrow -\kappa_2$, as proved in Theorem\,\ref{thm:real}\,. After eliminating the divergence of the Berry connection of $\tilde{\vct{u}}$ in Stage 3 in Section\,\ref{sec:stage3}, the first component of $\dbtilde{\vct{u}}$ is shown in Figure\,\ref{fig:eg2ttu1}. The other components are not shown as they are visually very similar to those of $\tilde{\vct{u}}$ in  Figure\,\ref{fig:eg2tu1}.
The Berry connection of $\tilde{\vct{u}}$ is shown in Figure\,\ref{fig:eg2bcdiv}, whose Helmholtz-Hodge decomposition is shown in Figure\,\ref{fig:eg2divf}. After eliminating its divergence, the Berry connection of $\dbtilde{\vct{u}}$ is shown in Figure\,\ref{fig:eg2bcdivless}.

\begin{figure}[h]
	\centering	
	\subfigure[]{
		\includegraphics[scale=0.35]{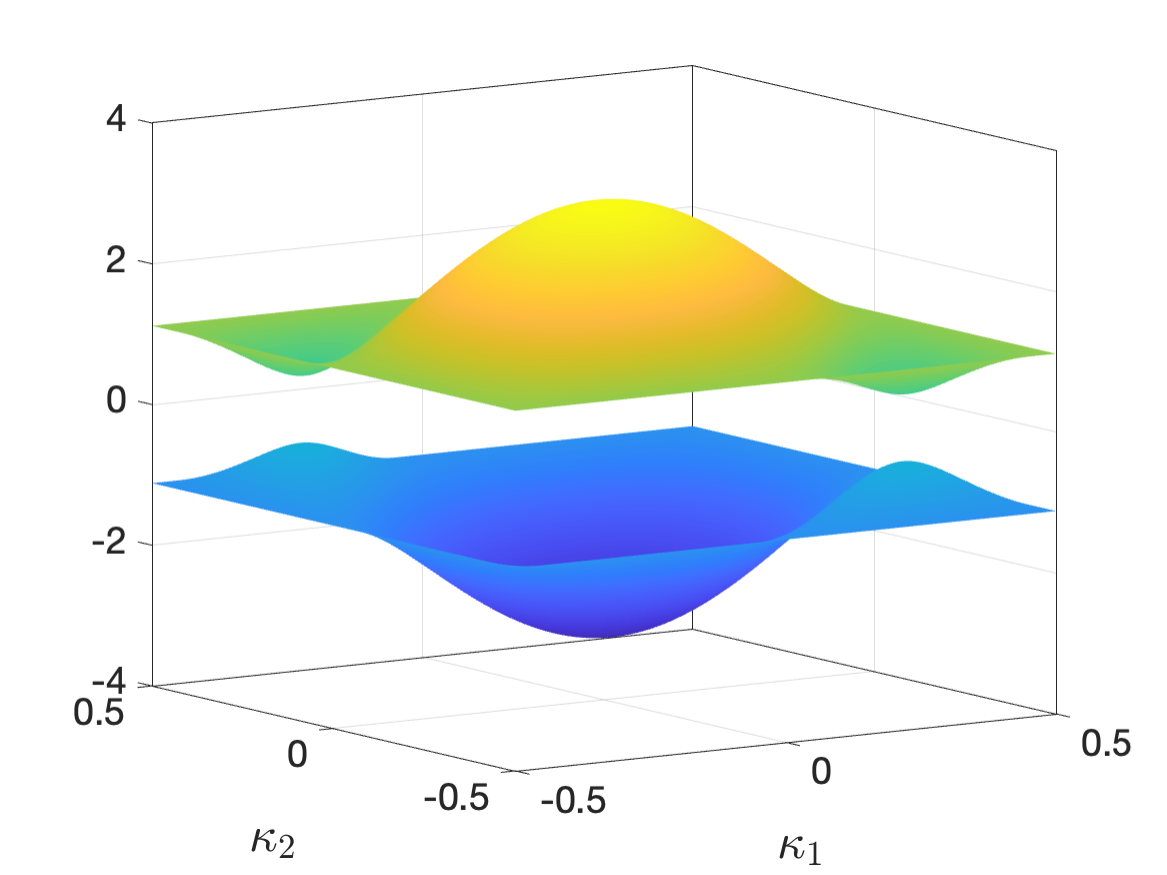}
		\label{fig:eg2eval}
	}
	\subfigure[]{
		\includegraphics[scale=0.35]{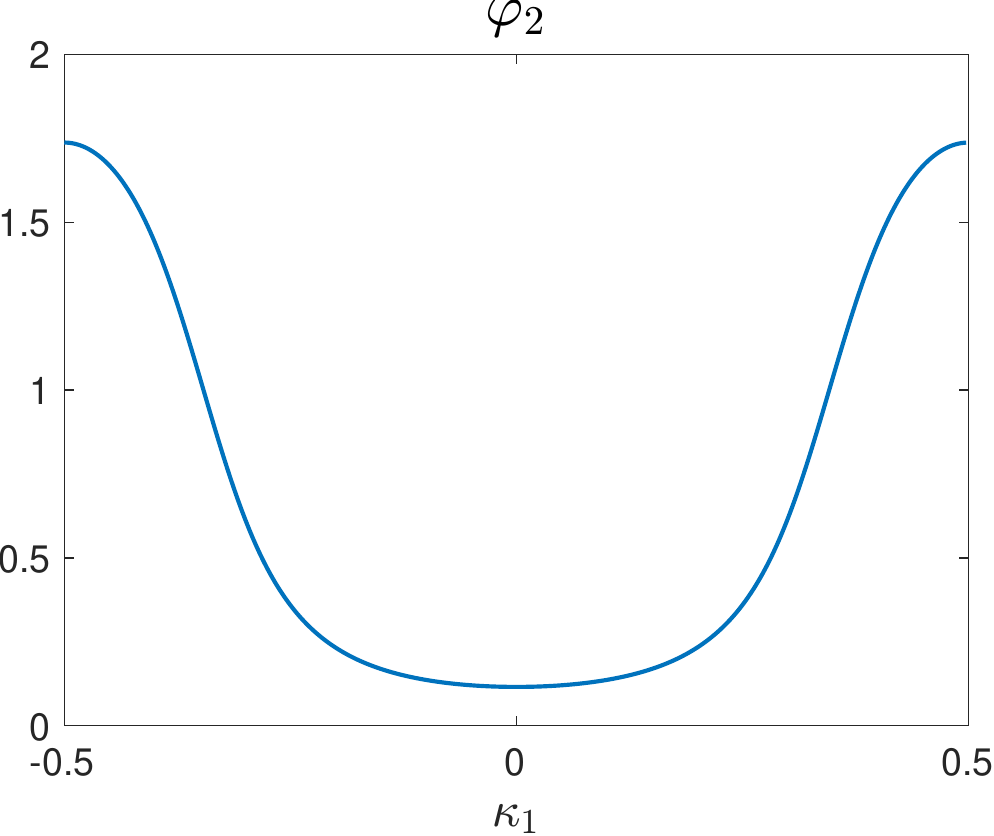}
		\label{fig:eg2phi2}
	}
	\vspace{-1em}
	\caption{ (a) Plot of eigenvalues of $H$ in Example 2. The band on top is picked. (b) The phase $\varphi_2$ in (\ref{eq:phitopo}) for Example 2. It is an even function as proved in Lemma\,\ref{lem:phi}.}	
	
\end{figure}

\begin{table}[h]
	\centering
	\begin{tabular}{c c c c c c}
		\hline
		$N$ & $t_{\rm para}$ (s)& $t_{\rm div}$ (s)& $\mathrm{E_{evec}}$  & $\mathrm{E_{Ch}}$ & $\mathrm{E_{div}}$ \\
		\hline 
		\hline
		50 & 1.30    & 0.040 & 5.66e-10& 4.56e-13 & 1.14e-4\\
		100 & 4.51  & 0.091 & 7.09e-12 & 1.50e-17 & 1.68e-9\\
		200 & 16.6   & 0.289& 1.07e-13 & 2.93e-17 & 2.53e-12\\
		400 & 66.3  & 1.07 & 2.10e-14 & 5.00e-18 & 1.14e-11\\
		\hline 
	\end{tabular}
	\caption{\label{tab:eg2}
		Timings and errors for Example 2. The error $\mathrm{E}_{\rm evec}$ shows sixth-order convergence and the error $\mathrm{E}_{\rm div}$ indicates if the sampling over $D^*$ is sufficient.}
\end{table}

\begin{table}[h!]
	\centering
	\begin{tabular}{c c c |c c c}
		\multicolumn{3}{c}{After Stage 2} &
		\multicolumn{3}{c}{After Stage 3 (Optimal solution)} \\
		\hline
		$\langle \vct{R}_x \rangle$ & 		$\langle \vct{R}_y \rangle$ & 			$\langle \norm{\vct{R}}^2 \rangle$ - $\norm{\langle \vct{R} \rangle}^2$ & $\langle \vct{R}_x \rangle$ & 		$\langle \vct{R}_y \rangle$ & 			$\langle \norm{\vct{R}}^2 \rangle$ - $\norm{\langle \vct{R} \rangle}^2$ \\
		\hline\hline 
		-0.184913 & 4.29e-16 & 0.270171 & -0.184913 & -2.23e-16 &	0.233954			\\
		\hline
	\end{tabular}
	\caption{The Wannier center and variance for Example 2 before and after eliminating the divergence of the Berry connection computed for $N=400$. The solution after Stage 2 is already relatively close to the optimal one. \label{tab:var2}}
\end{table}

\begin{figure}[h]
	\centering
	\includegraphics[scale=0.25]{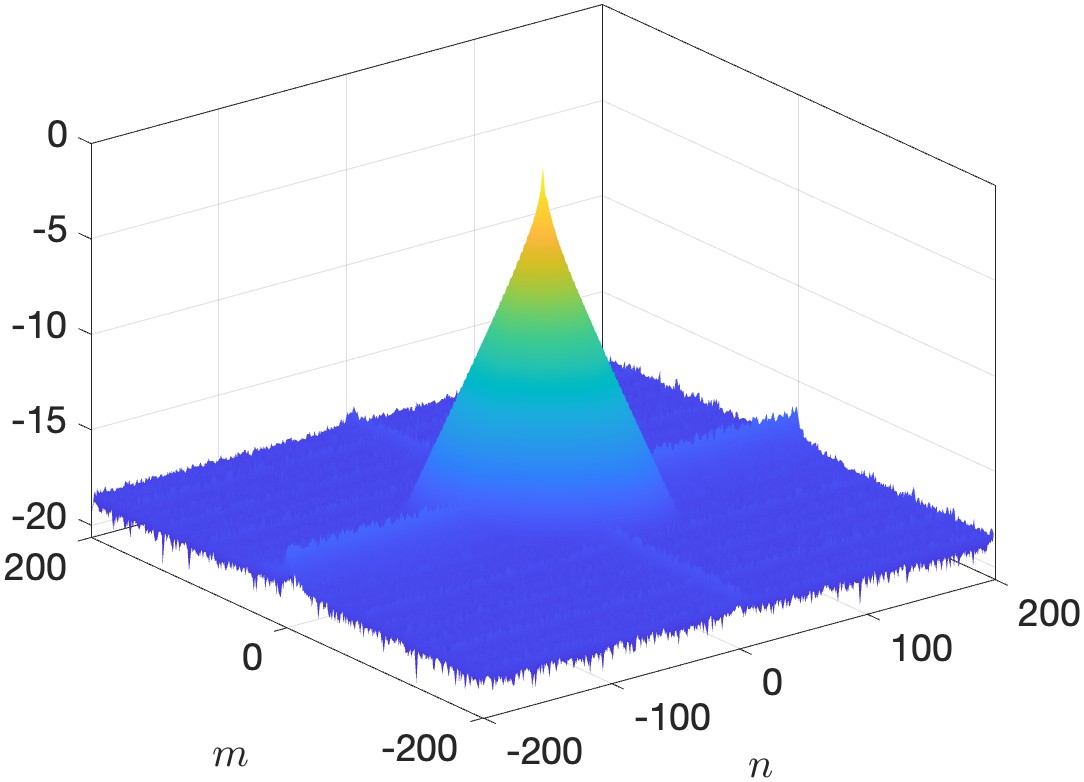}		
	
	\includegraphics[scale=0.40]{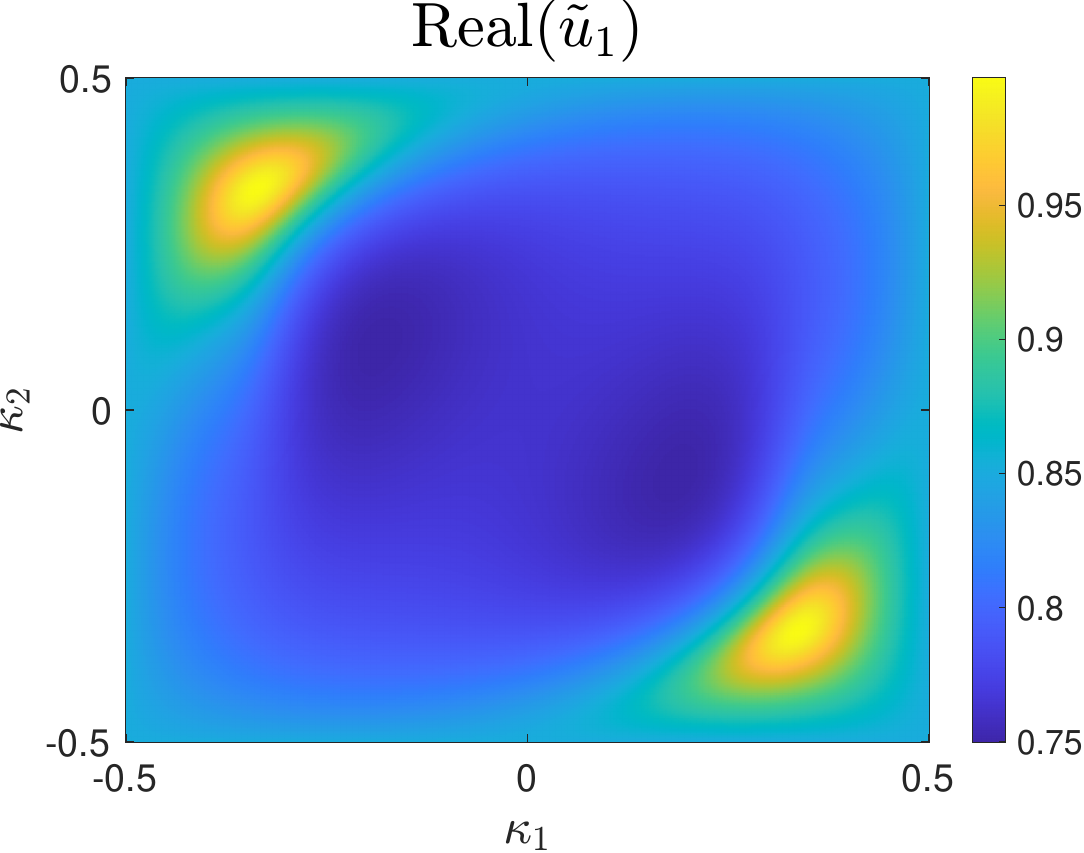}
	\includegraphics[scale=0.40]{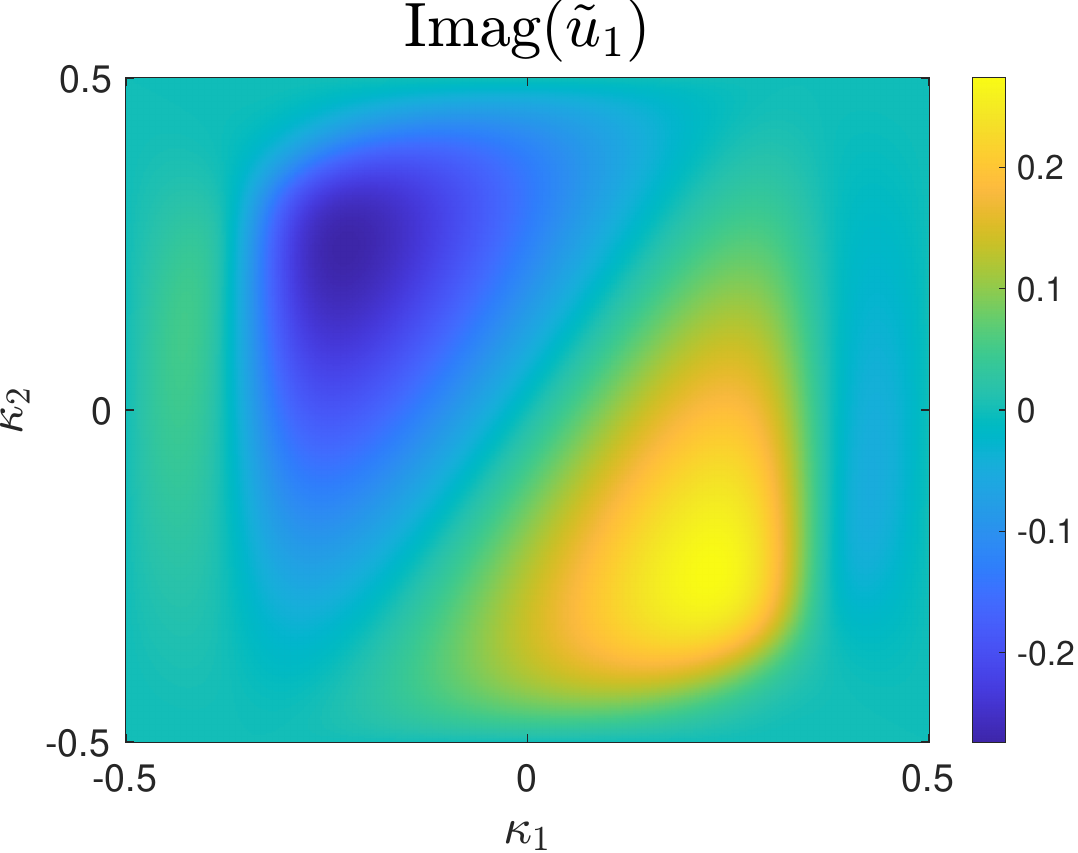}
	
		\vspace{-1em}
	\caption{ Plot of the real and imaginary  part of the component $\tilde{u}_1$ with the absolute value of the Fourier coefficients of $\tilde{u}_1$ in the $\log_{10}$ scale in Example 2.  }
\label{fig:eg2tu1}
\end{figure}

\begin{figure}[h]
	\centering

	\includegraphics[scale=0.25]{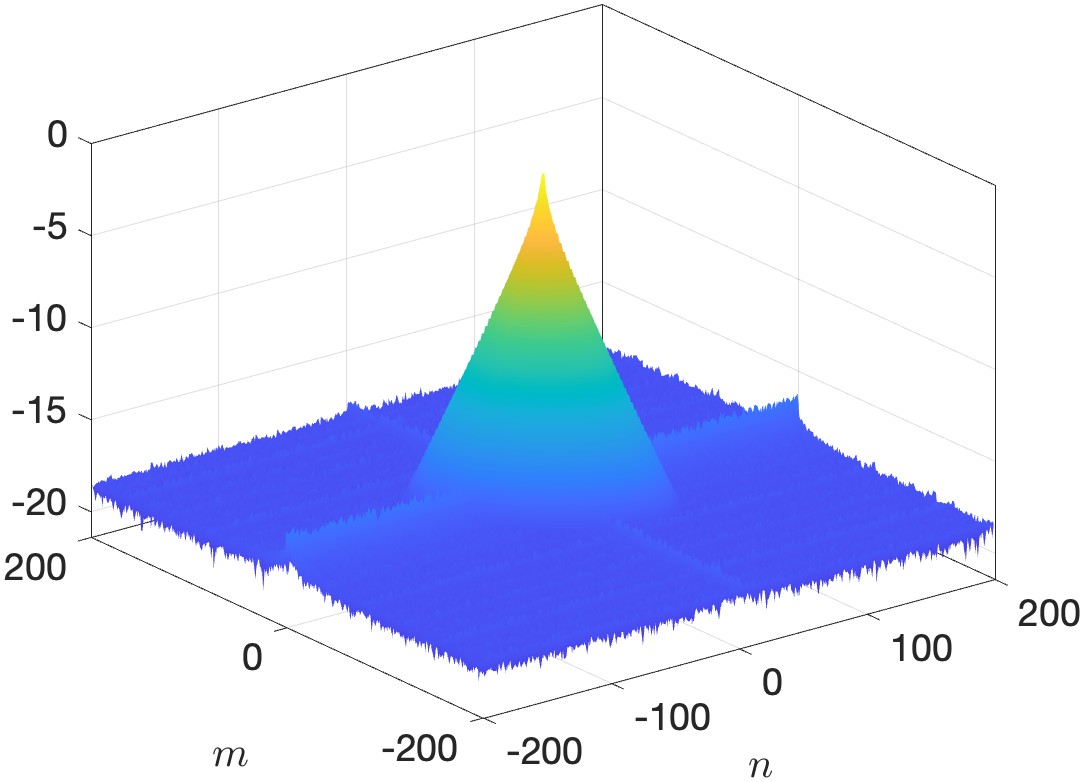}		
	
	\includegraphics[scale=0.4]{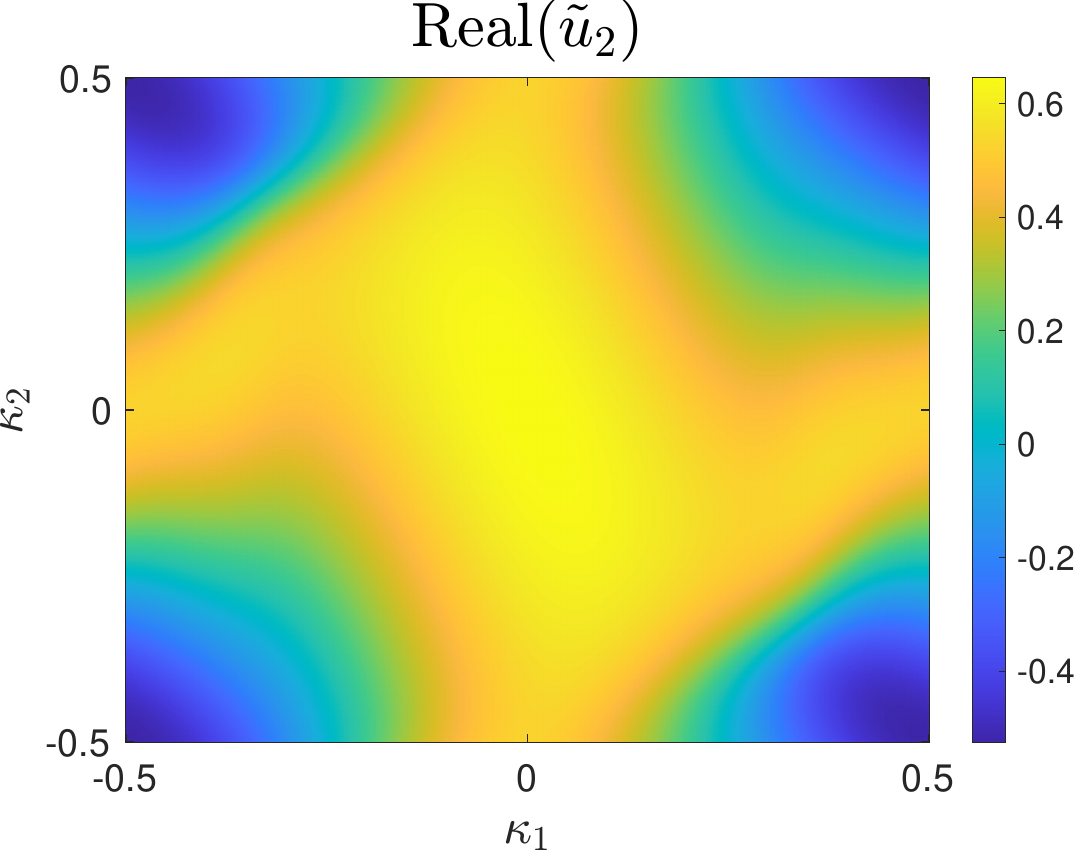}
	\includegraphics[scale=0.4]{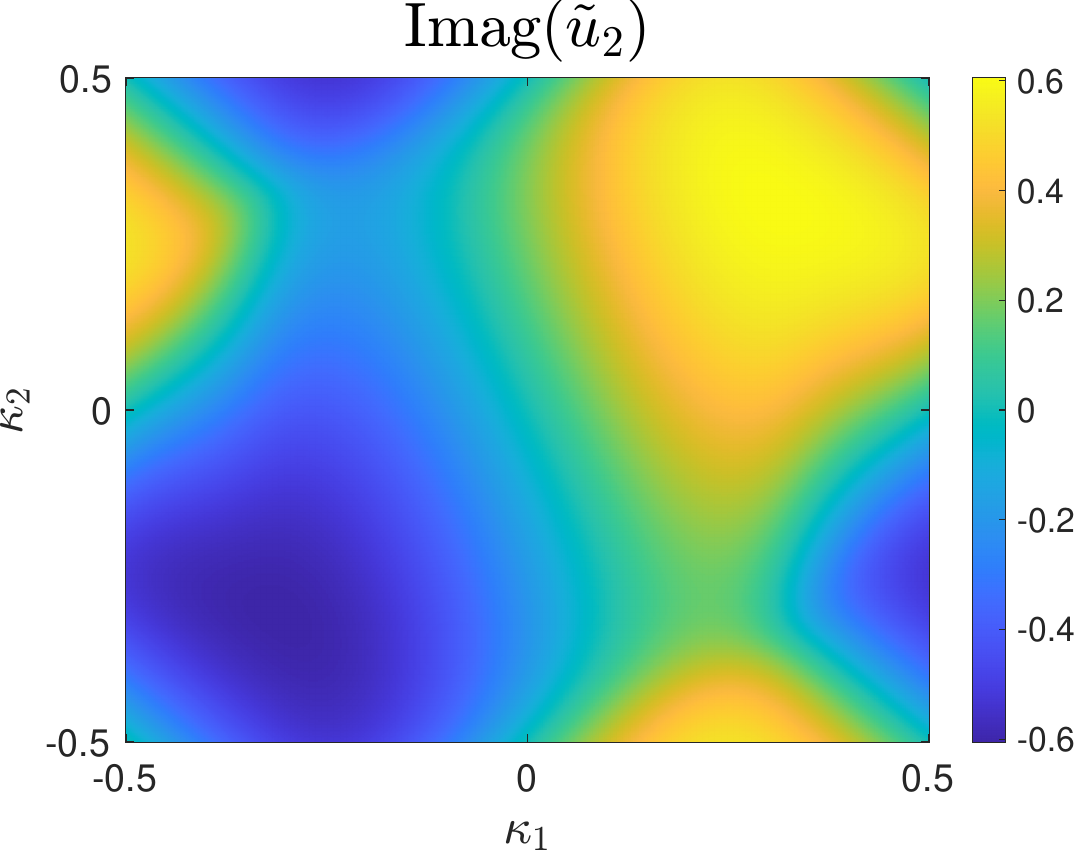}

		\vspace{-1em}
		\caption{ Plot of the real and imaginary  part of the component $\tilde{u}_2$ with the absolute value of the Fourier coefficients of $\tilde{u}_2$ in the $\log_{10}$ scale in Example 2.  }
\label{fig:eg2tu2}
\end{figure}
\begin{figure}[h]
	\centering	
		\includegraphics[scale=0.25]{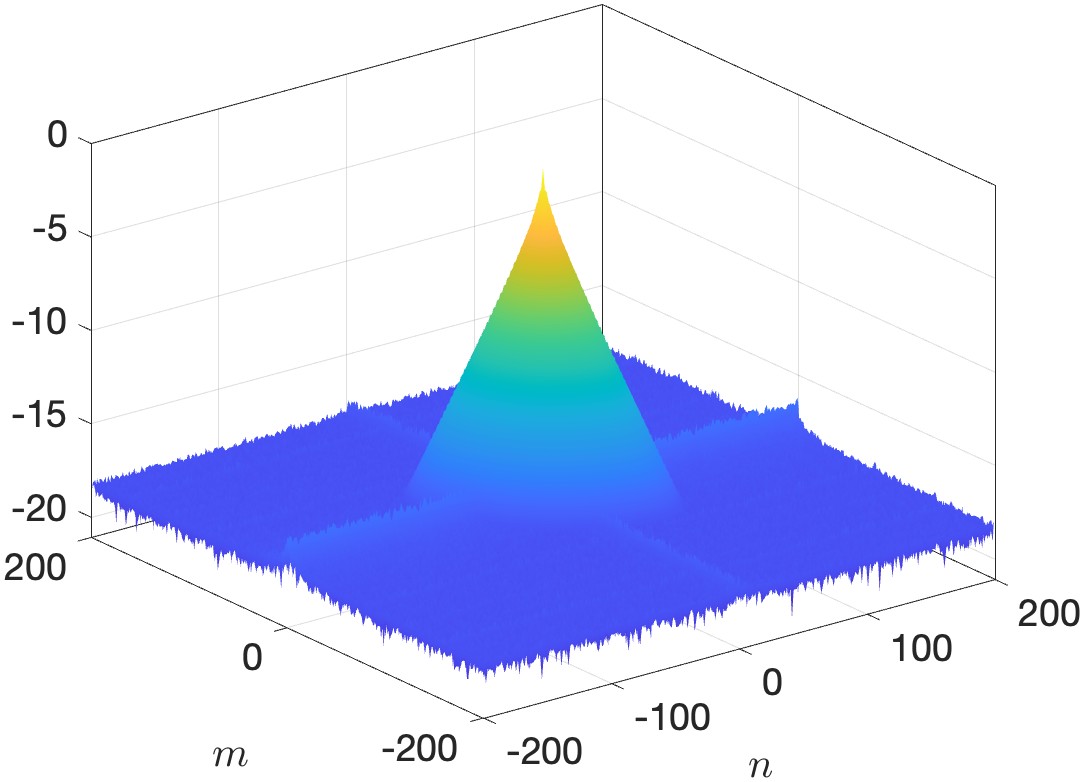}	
		
	\includegraphics[scale=0.4]{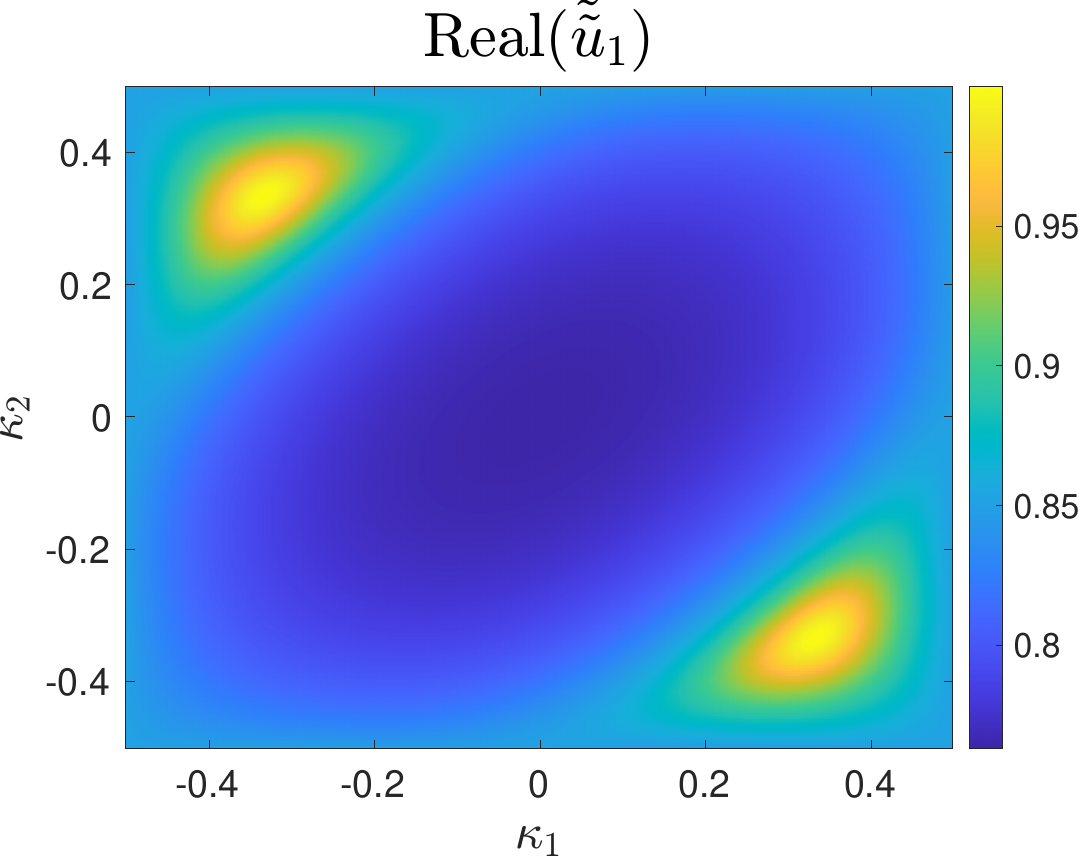}
	\includegraphics[scale=0.4]{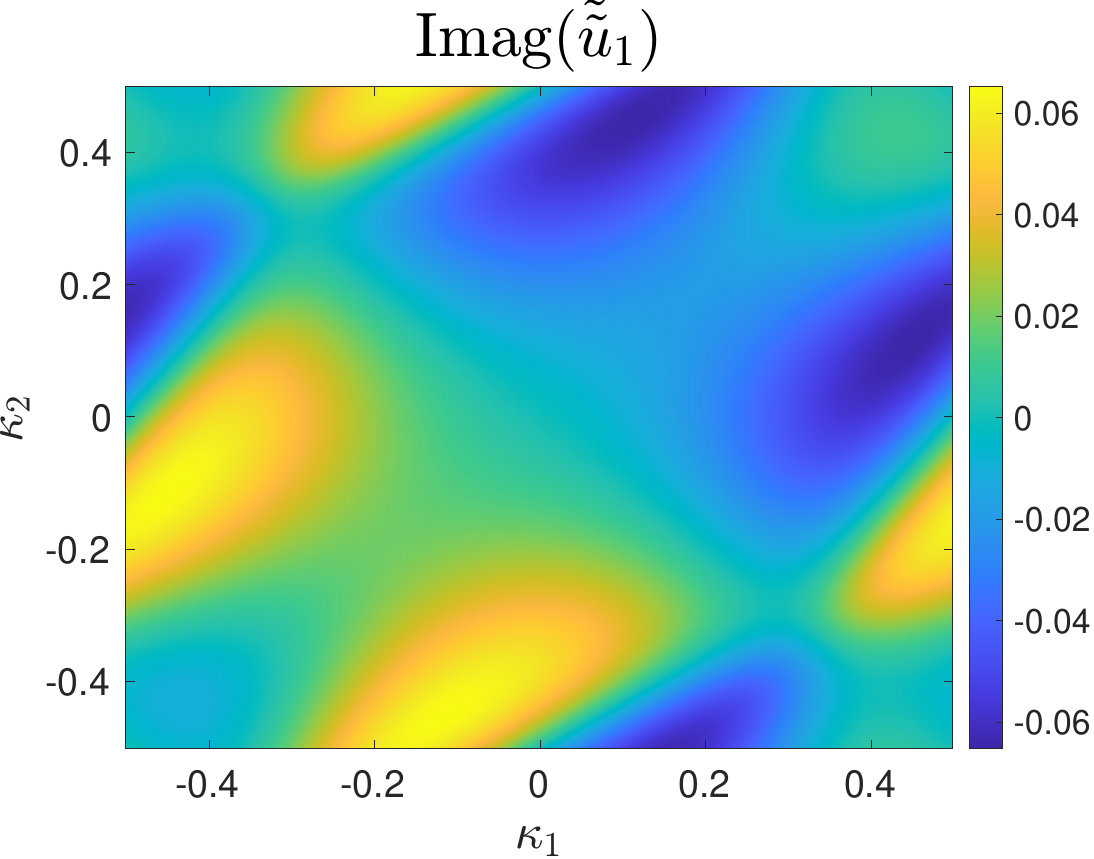}

	\vspace{-1em}
	\caption{ Plot of the quantities in  Figure\,\ref{fig:eg2tu1}  after eliminating the divergence of the Berry connection. Other components are not shown as they are visually very similar to those in  Figure\,\ref{fig:eg2tu1}.  }
	\label{fig:eg2ttu1}	
\end{figure}
\begin{figure}[h]
	\centering	
	
	\subfigure[]{
		\includegraphics[scale=0.55]{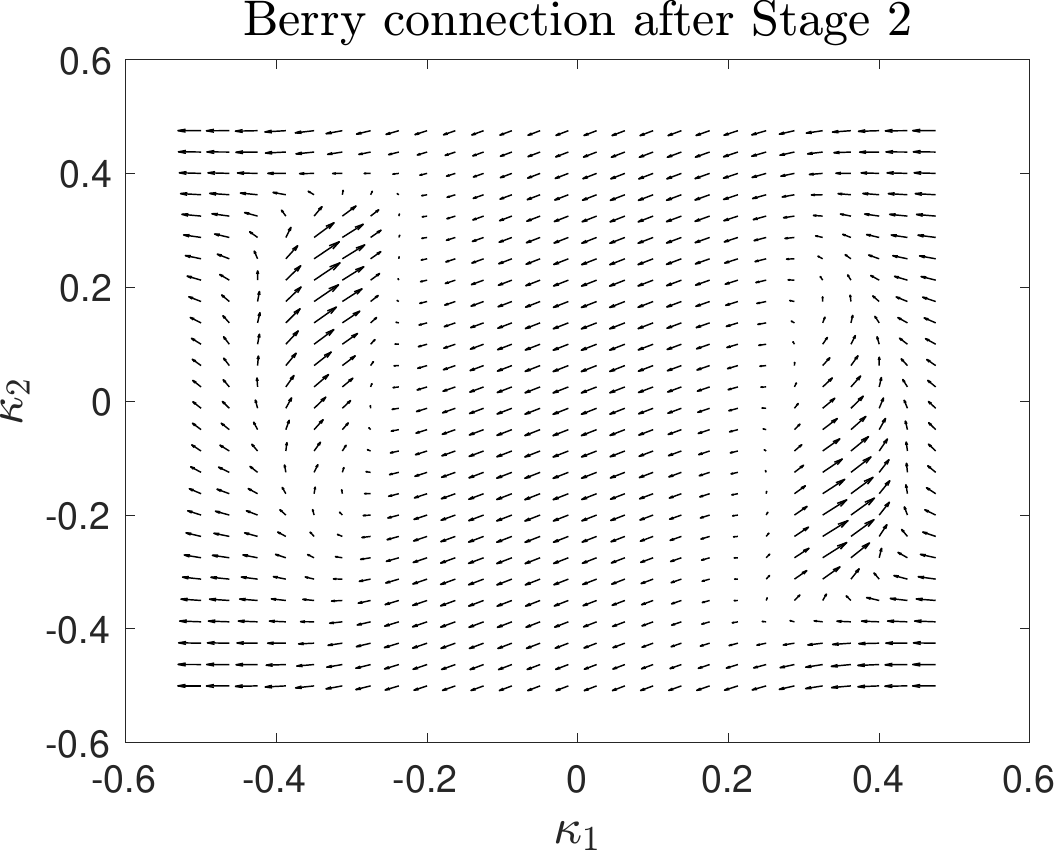}
		\label{fig:eg2bcdiv}
	}
	\subfigure[]{
		\includegraphics[scale=0.55]{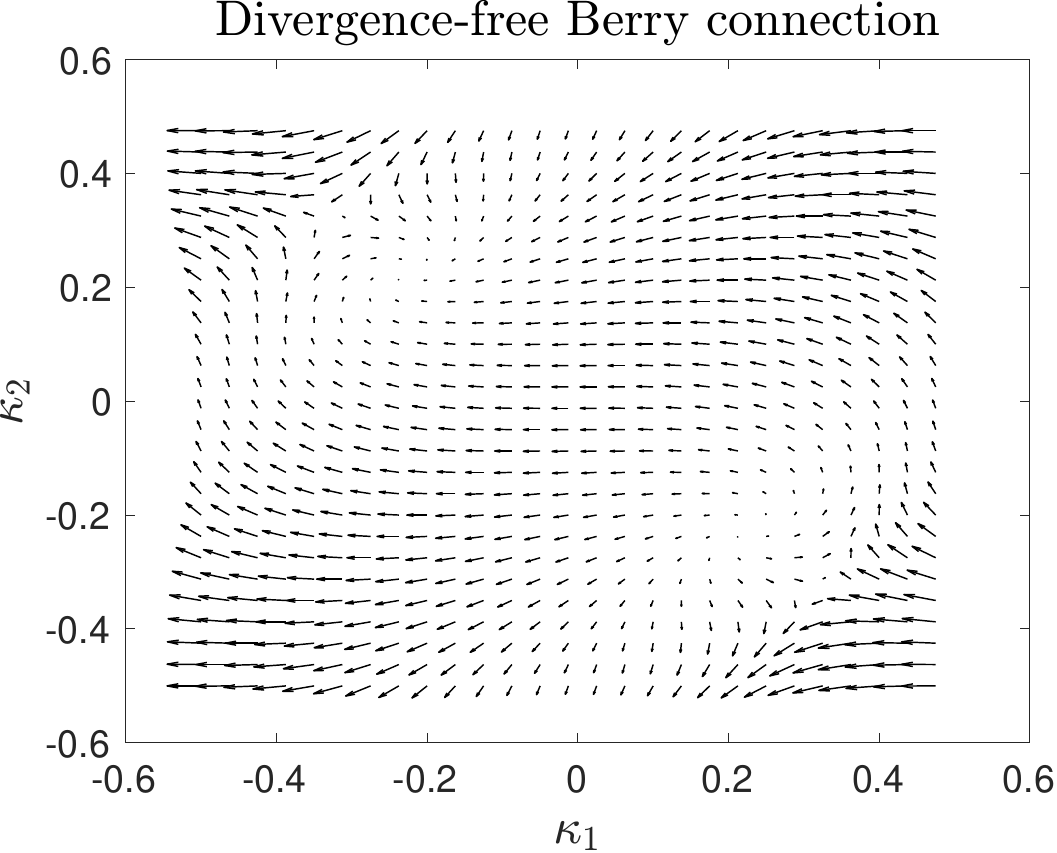}
		\label{fig:eg2bcdivless}
	}
	\vspace{-1em}
	\caption{ (a) Plot of the Berry connection after Stage 2 for Example 2. (b) Plot of the Berry connection after eliminating its curl-free component in Stage 3. All vectors shown are in the $\vct{e}_x,\vct{e}_y$ basis. }

\end{figure}
\begin{figure}[h]
	\centering	
	\includegraphics[scale=0.44]{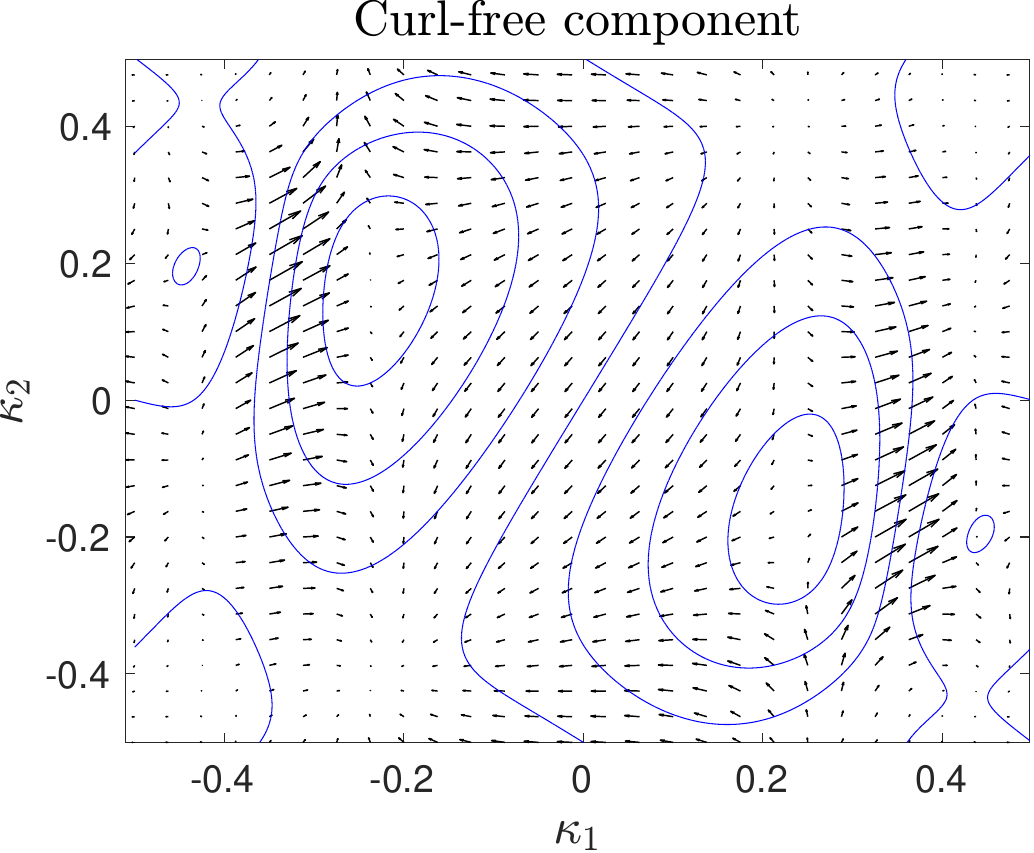}
	
	\includegraphics[scale=0.44]{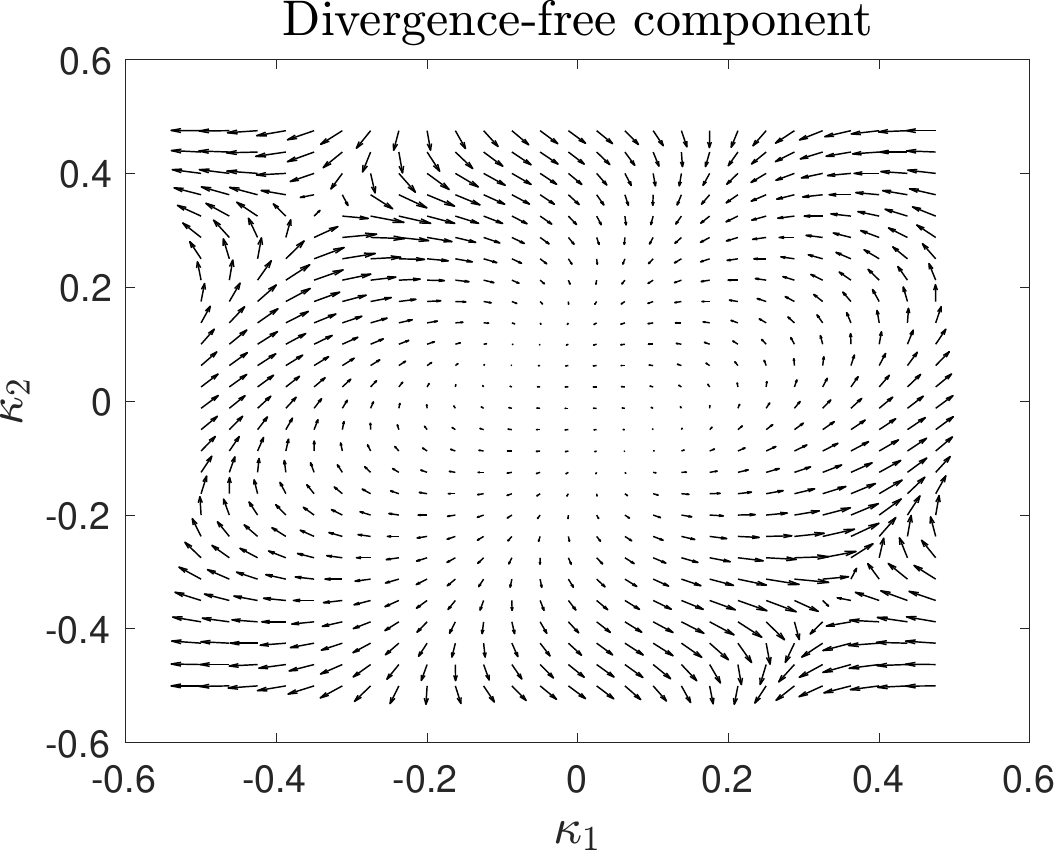}
	
	\includegraphics[scale=0.44]{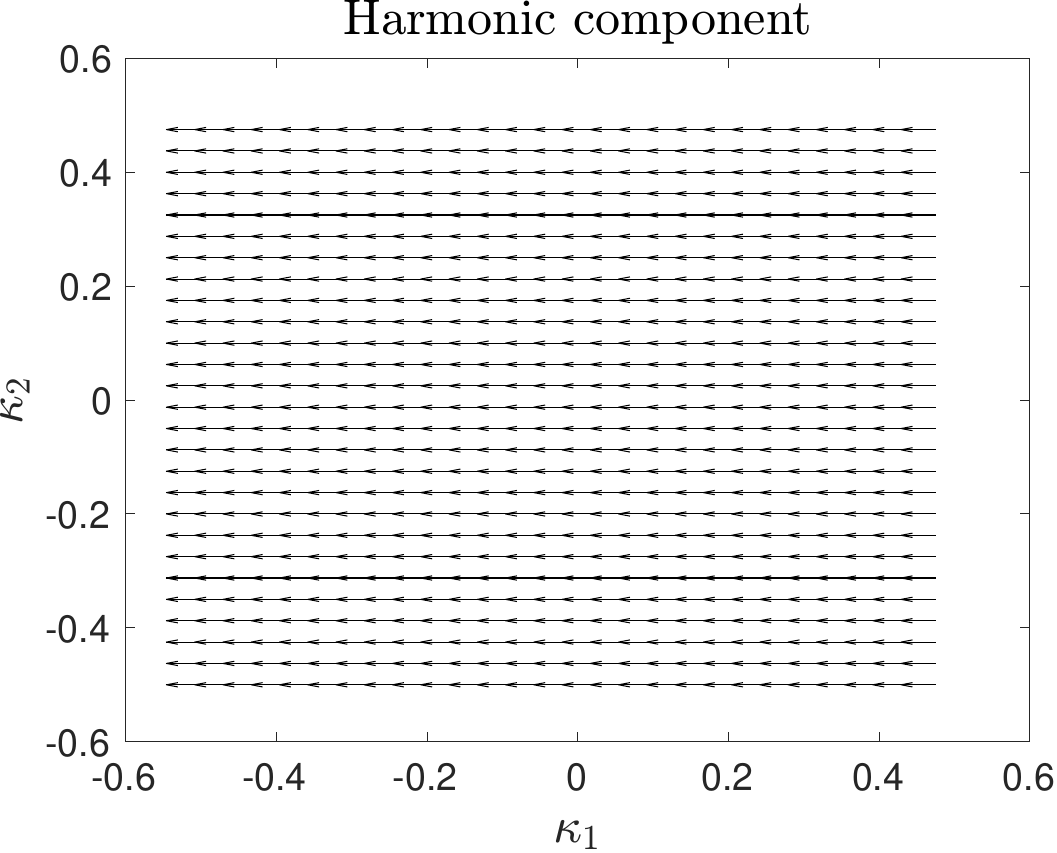}

	\vspace{-1em}
	\caption{ The Helmholtz-Hodge decomposition of the Berry connection in Figure\,\ref{fig:eg2bcdiv}. All vectors shown are in the $\vct{e}_x,\vct{e}_y$ basis. The equipotential lines are shown for the curl-free vector field, which is eliminated to obtain the optimal solution in Figure\,\ref{fig:eg2bcdivless}. 	\label{fig:eg2divf}}		
\end{figure}

\clearpage
\subsection{Example 3: Haldane model (topologically non-trivial case) \label{sec:ntrivial}}
In this example, we consider a topologically non-trivial version of the Haldane model in Section\,\ref{sec:trivial}. All the lattice vectors are identical to those in Section\,\ref{sec:trivial} with a modified matrix $H$ given by the formula
\begin{align}
	\begin{split}
	H(\vct{k}) &=  \sb{\mat{V_0 &  t_1(1 + e^{-\I\vct{k}\cdot\vct{a}_1} + e^{-\I\vct{k}\cdot\vct{a}_2})\\ 
			t_1(1 + e^{\I\vct{k}\cdot\vct{a}_1} + e^{\I\vct{k}\cdot\vct{a}_2})					& -V_0 }} \\
		& + t_2 \bb{\sin(\vct{k}\cdot \vct{a}_1) - \sin(\vct{k}\cdot \vct{a}_2) - \sin(\vct{k}\cdot (\vct{a}_1 - \vct{a}_2) )}\sb{\mat{1&0\\0&-1}}\,,
	\end{split}
\end{align}
where the constants are chosen as $t_1 = 1$, $t_2 = -0.45$ and  $V_0 = 0.5$. The eigenvalues as a function of the parameterization $\kappa_1$ and $\kappa_2$ are shown in Figure\,\ref{fig:eg2neval}, and the top band is chosen for the construction. The first Chern number $C_1 = 1$, so a topological obstruction is present (see the discussion below Theorem\,\ref{thm:stage2}). As a result,  in Figure\,\ref{fig:eg2nphi2}, its constructed $\varphi_2$ (see (\ref{eq:phitopo})) is not periodic, as opposed to that in Figure\,\ref{fig:eg1phi2} and \ref{fig:eg2phi2} for Example 1 and 2.

The two components of the assignment $\tilde{\vct{u}}$ after Stage 2 in Section\,\ref{sec:stage2} is shown in Figure\,\ref{fig:eg2ntu1}--\ref{fig:eg2ntu2}, together with the absolute value of their Fourier coefficients in the  $\log_{10}$ scale. The Fourier coefficients decay exponentially asymptotically in one direction but very slowly in the other one since by construction $\tilde{\vct{u}}(\vct{k}(\kappa_1,\kappa_2))$ is analytic and periodic in $\kappa_1$, but only analytic in $\kappa_2$, which can be seen from Figure\,\ref{fig:eg2ntu1}--\ref{fig:eg2ntu2}. 

Since this example is computationally similar to the trivial case in Section\,\ref{sec:trivial}, we only show the convergence of the computed Chern number in Table\,\ref{tab:eg3}. 

\begin{figure}[h]
	\centering	
	\subfigure[]{
		\includegraphics[scale=0.35]{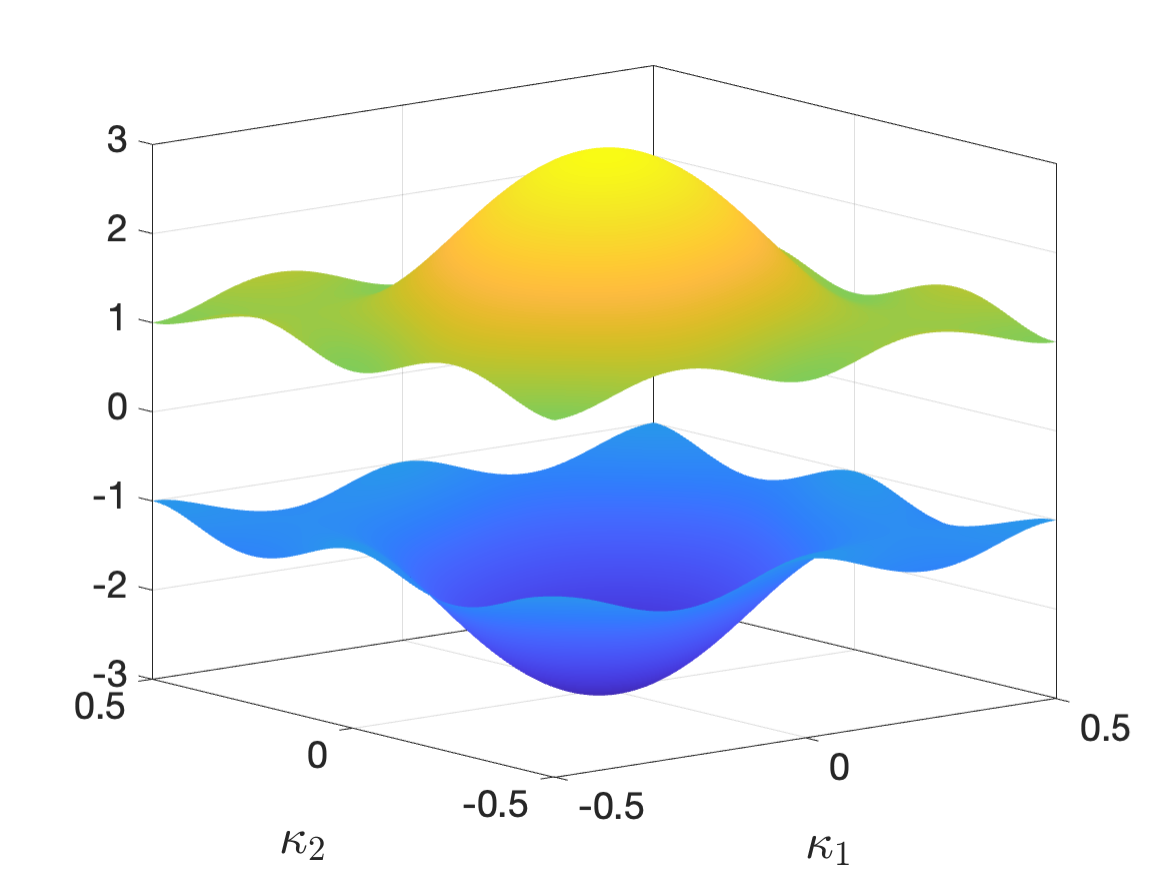}
		\label{fig:eg2neval}
	}
	\subfigure[]{
		\includegraphics[scale=0.35]{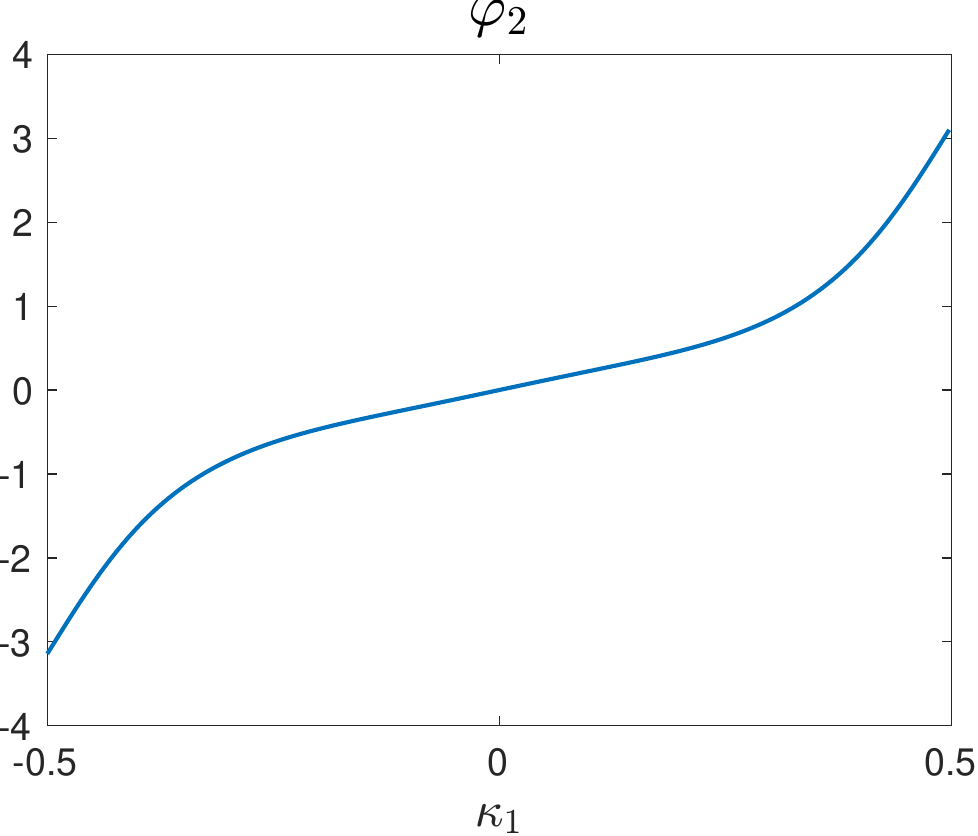}
		\label{fig:eg2nphi2}
	}
	\vspace{-1em}
	\caption{ (a) Plot of eigenvalues of $H$ in Example 3. The band on top is picked. (b) The phase $\varphi_2$ in (\ref{eq:phitopo}) for Example 3. It is discontinuous due to the nonzero Chern number.}	
	
\end{figure}

\begin{table}[h]
	\centering
	\begin{tabular}{c c }
		\hline
		$N$ & $\mathrm{E_{Ch}}$  \\
		\hline \hline
		50 & 2.69e-14   \\
		100 & 1.11e-16   \\
		\hline 
	\end{tabular}
	\caption{\label{tab:eg3}
		Errors in the computed first Chern number $C_1$ for Example 3 ($C_1 = 1$). }
\end{table}

\begin{figure}[h]
	\centering
	\includegraphics[scale=0.25]{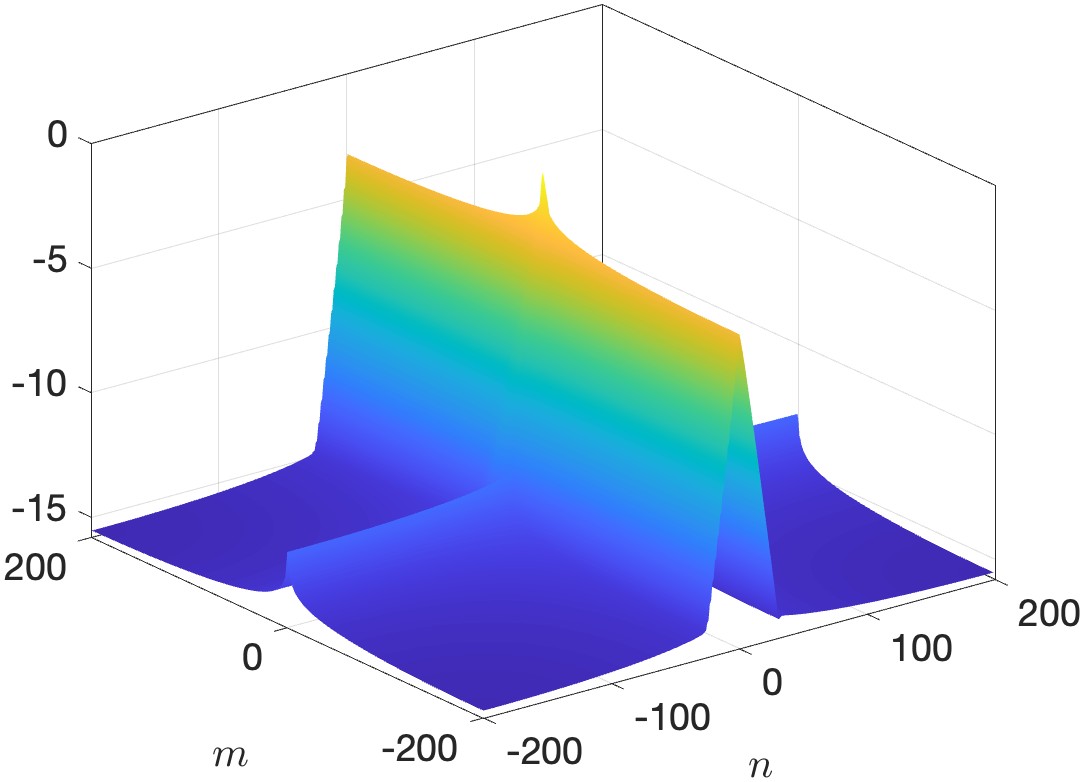}		
	
	\includegraphics[scale=0.43]{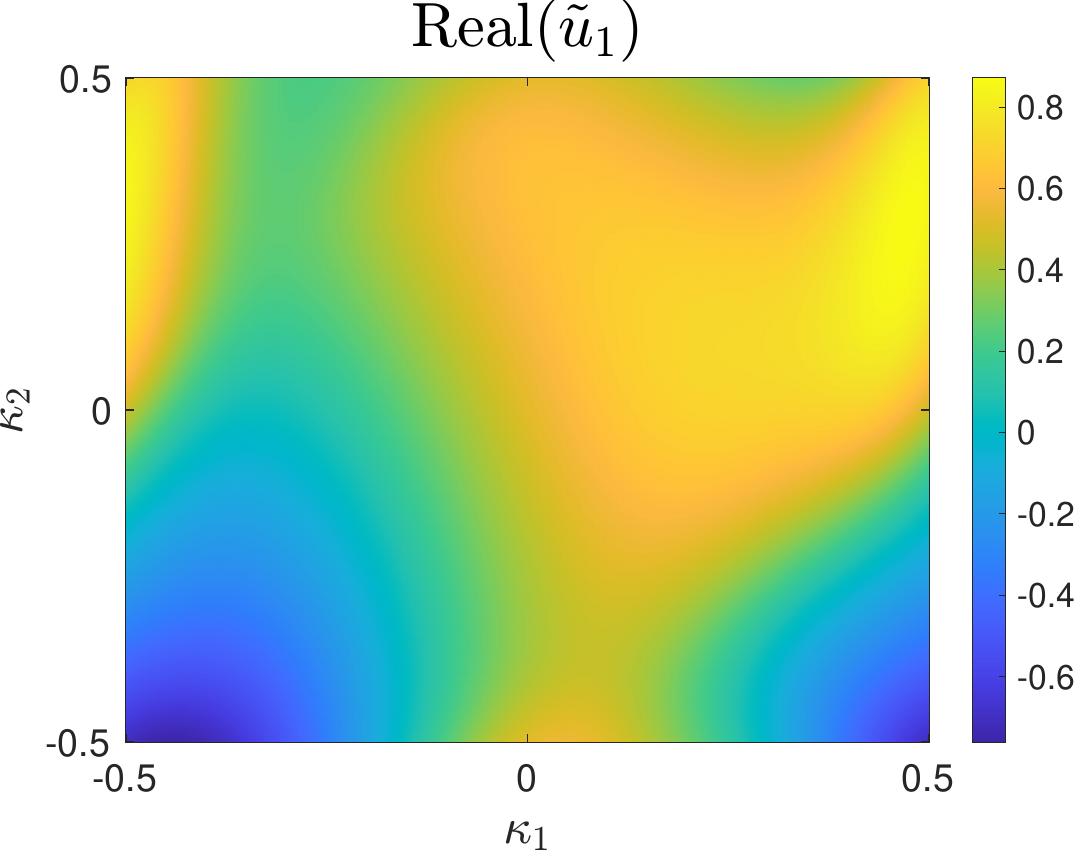}
	\includegraphics[scale=0.43]{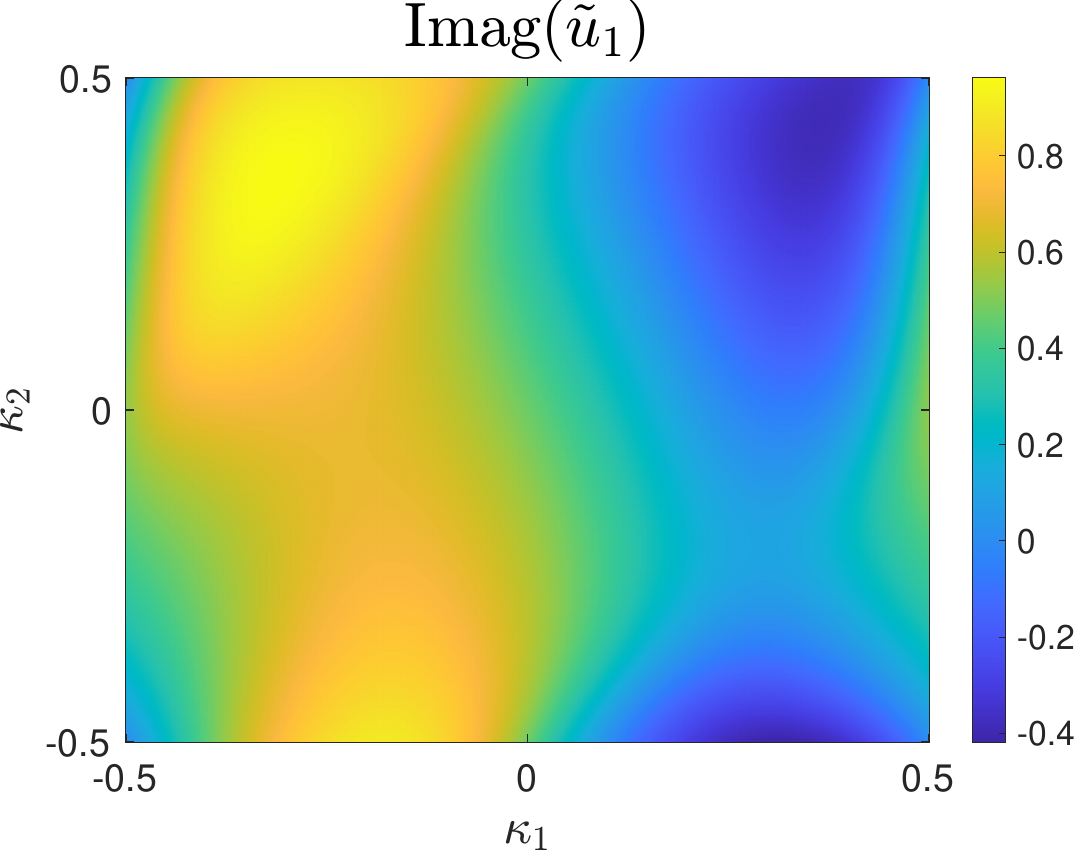}
	
	\vspace{-1em}
	\caption{ Plot of the real and imaginary  part of the component $\tilde{u}_1$ with the absolute value of the Fourier coefficients of $\tilde{u}_1$ in the $\log_{10}$ scale in Example 3. The Fourier coefficients decay exponentially only in one direction since $\tilde{u}_1$ is analytic and periodic in $\kappa_1$ but only analytic in $\kappa_2$\,.}
	\label{fig:eg2ntu1}
\end{figure}
\begin{figure}[h]
	\centering
	
	\includegraphics[scale=0.25]{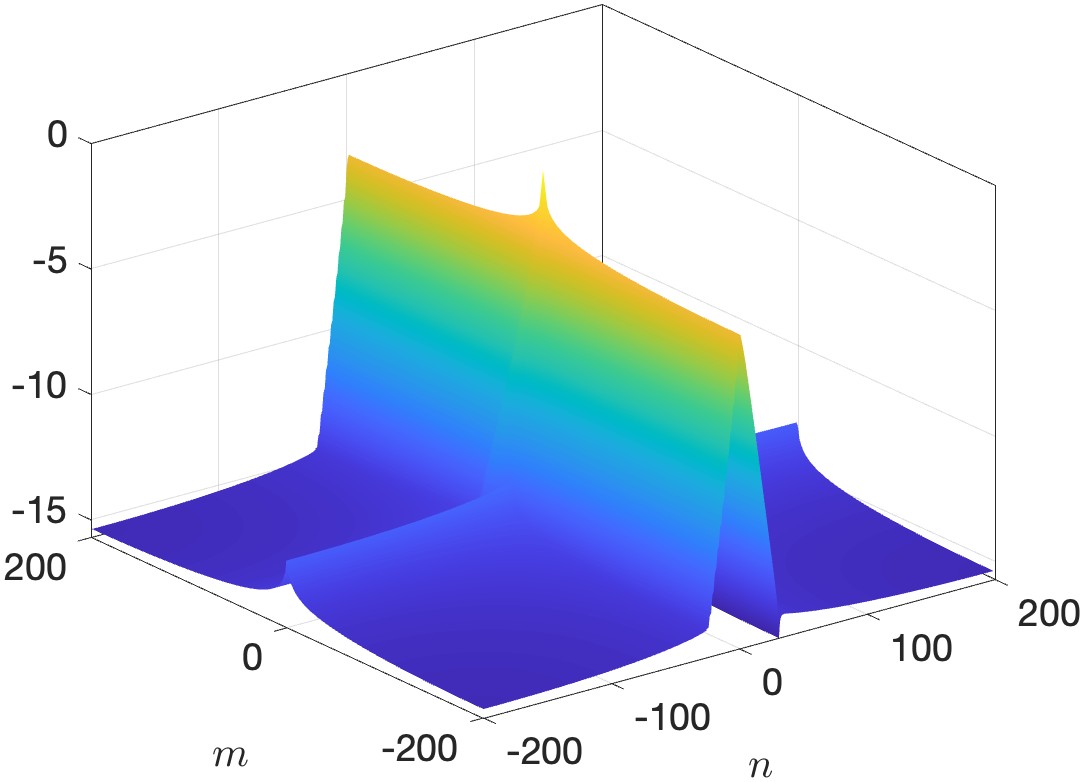}		
	
	\includegraphics[scale=0.43]{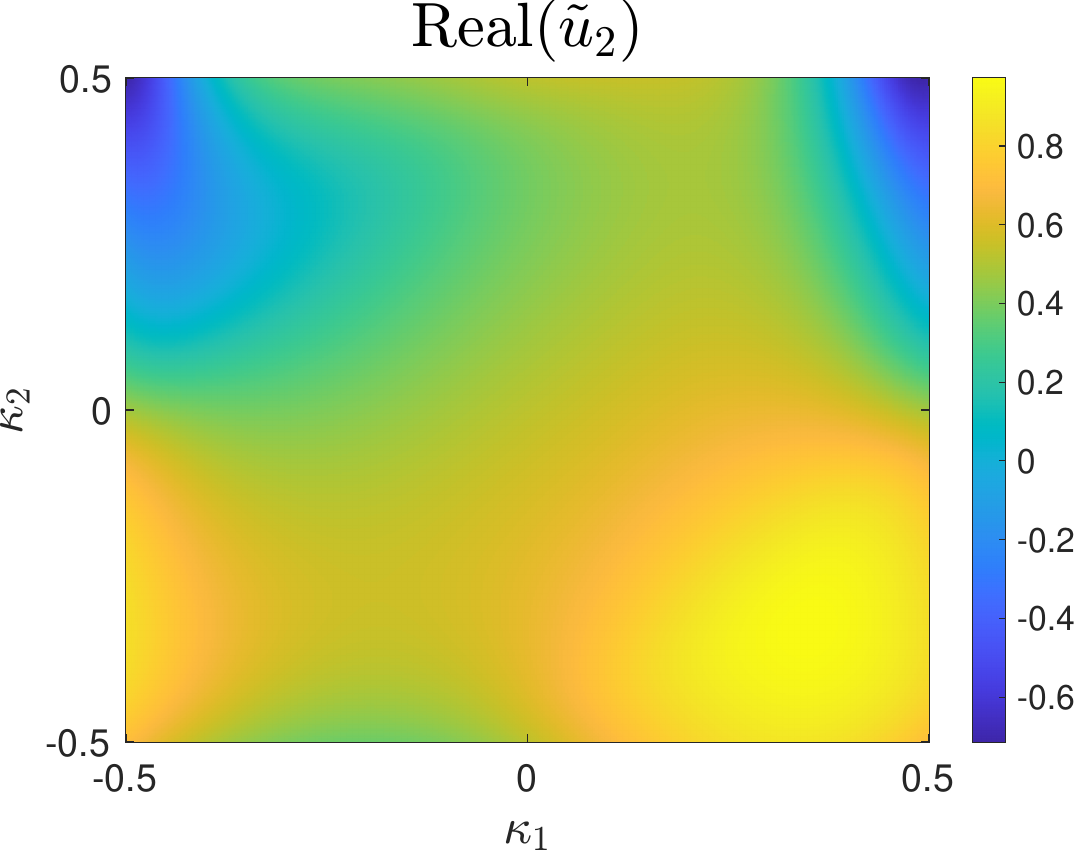}
	\includegraphics[scale=0.43]{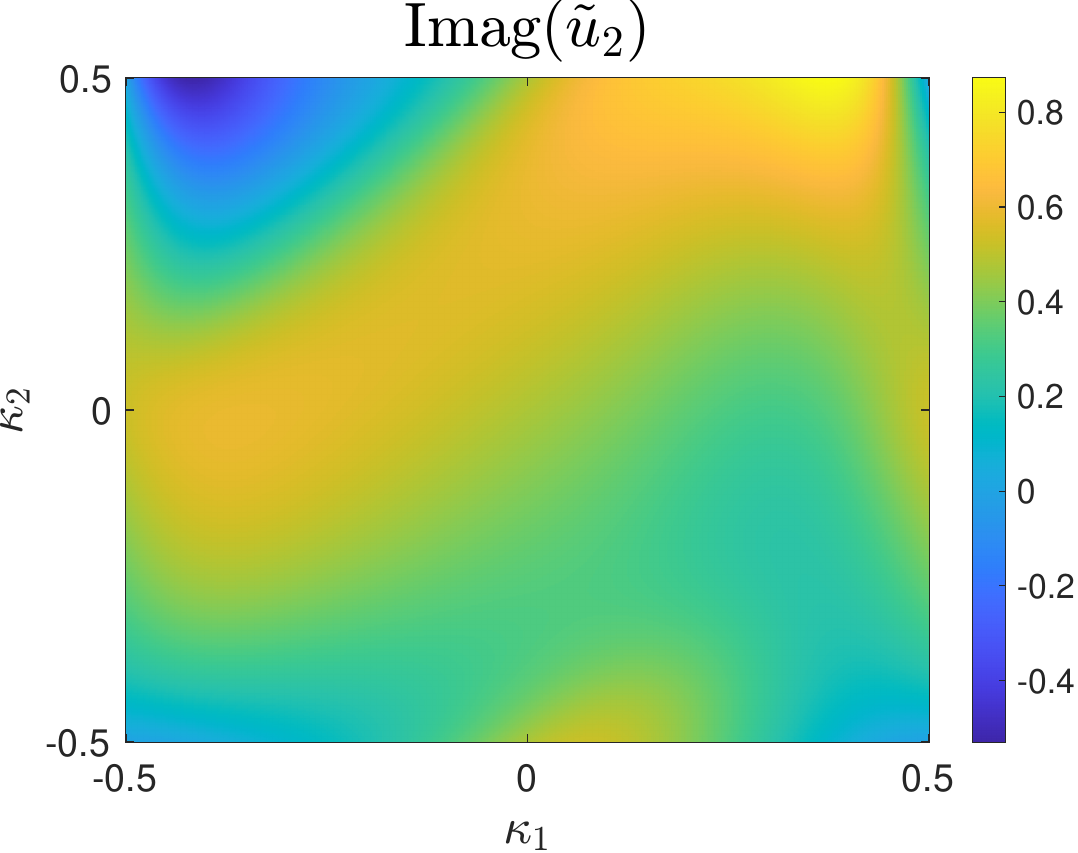}

	\vspace{-1em}
	\caption{ Plot of the real and imaginary  part of the component $\tilde{u}_2$ with the absolute value of the Fourier coefficients of $\tilde{u}_2$ in the $\log_{10}$ scale in Example 3. The Fourier coefficients decay exponentially only in one direction since $\tilde{u}_2$ is analytic and periodic in $\kappa_1$ but only analytic in $\kappa_2$. }
	\label{fig:eg2ntu2}
\end{figure}

\clearpage
\section{Conclusion and generalizations \label{sec:final}}
In this paper, we presented a rapidly convergent scheme for computing the globally optimal Wannier functions of an isolated single band in two dimensions. In the absence of topological obstructions, we proved that parallel transport (with simple corrections) alone leads to assignments of eigenvectors corresponding to exponentially localized Wannier functions. Furthermore, when the model possesses time reversal symmetry, we proved that the resulting Wannier functions are automatically real. Then a single gauge transform can be applied to eliminate the divergence of the Berry connections of the corresponding assignments, obtaining the globally optimal assignments for Wannier functions with minimum spread. We illustrated the efficiency and accuracy of the schemes with several examples. We compared the results above to those obtained via the parallel transport approach in Remark\,\ref{rmk:twist}, which decouples the computation of band structure information and Wannier functions. This alternative approach produces the same Wannier functions with second-order accuracy, but it has the advantage of being simple and computationally cheap once the band structure information is known.

The schemes in this paper can be further accelerated by exploiting symmetries of the lattice (and time-reversal symmetry if available); at least partial symmetry information can be easily utilized to avoid unnecessary computation of parallel transport and band structure information without affecting the equispaced grids scheme in this paper. Moreover, better time-stepping schemes for parallel transport can be used to further improve convergence or to reduce the number of matrix inversions.

Since the scheme in this paper builds two-dimensional assignments of eigenvectors based on one-dimensional ones, the analysis and algorithms can be easily generalized to three-dimensional cases. The three-dimensional construction will be made of two-dimensional slices obtained by the schemes in this paper, and all Chern-like numbers will appear in the construction. 

If we replace the construction of one-dimensional Wannier functions in this paper with the Fourier approach in \cite{gopal2024high} for Schr\"odinger operators, the analysis in this paper carries over to the operator case almost verbatim.  For the algorithms, when high accuracy is not required, a similar parallel transport scheme to that in Remark\,\ref{rmk:twist} can be applied, followed by eliminating the divergence of Berry connections as in this paper.
When higher accuracy is needed, better parallel transport schemes are required. 
The scheme in this paper and \cite{gopal2024high} can be applied, but they are not suitable for the operator case if the matrices to be inverted are large. As a result, better numerical schemes for the inversion are needed. These questions are currently under vigorous investigation.

For extensions to the isolated multiband case,  although the overall idea remains the same, there are notable differences. For example, the Berry connections in the multiband case become Lie-algebra valued; in two dimensions, the elimination of divergence cannot be done in a single step, but the potential theory approach will significantly reduce the number of optimization steps. Extensions of \cite{gopal2024high} and this paper to the isolated multiband case have been worked out and are now in preparation for publication.

\section*{Acknowledgement}
The author thanks John Schotland for drawing their attention to this subject. 
The author also thanks Vladimir Rokhlin for his support and helpful discussions pertaining to this work.

\clearpage
\bibliographystyle{plain}

\bibliography{ref} 

\newpage
\section{Appendix}
\subsection{Alternative approach: direct computation of optimal gauge\label{sec:method2}}
In this section, we describe an alternative approach to the method in section\,\ref{sec:method1}. 
 In this approach, we directly compute the gauge-invariant parts of the  Berry connection in (\ref{eq:adecom4}), which uniquely determines the globally optimal assignment of eigenvectors (and the Wannier function). It proceeds by first computing the Berry curvature, from which a Poisson's equation solve is done to obtain the potential $F$ in (\ref{eq:adecom4}). The potential $F$ gives the Berry connection satisfying the integrability condition in 
Theorem\,\ref{thm:intecond}, so that a Pfaffian system of the form in (\ref{eq:odekn1}) can be solved to define an analytic (not necessarily $\Lambda^*$-periodic) assignment of eigenvectors on $D^*$. To make the assignment $\Lambda^*$-periodic, the constant harmonic components $(h_x, h_y$) is obtained, thus completing the construction. We do not provide any numerical procedure for this approach since the key ingredients required are essentially identical to those in Section\,\ref{sec:numerics}.

It should be observed that, although the method in this section is conceptually simpler, it is more expensive than the first method in Section\,\ref{sec:method1} due to the computation of the Berry curvature without doing the assignment. Furthermore, since it uses the gauge-invariant property of the Berry curvature for a single band, it cannot be generalized to the multiband case, where the Berry curvature tensor is only gauge-covariant.

\subsubsection{Stage 1: computing the Berry curvature}
In order to compute the Berry curvature $\Omega_{xy}$, we compute the family of eigenvalues $E$ and eigenvectors $\vct{u}$:
\begin{equation}
	H(\vct{k})\vct{u}(\vct{k}) = E(\vct{k})\vct{u}(\vct{k})\,,\quad \vct{k}\in D^*\,.
	\label{eq:eigstage1}
\end{equation}
It should be observed that the phase choice of $\vct{u}$ can be arbitrary since the Berry curvature $\Omega_{xy}$ is gauge-invariant (see Remark\,\ref{rmk:bc}). Next, we compute $\Omega_{xy}$ by the formula
\begin{equation}
	\Omega_{xy}(\vct{k}) = \I \frac{\partial}{\partial {k_x} }\vct{u}^*(\vct{k}) \frac{\partial}{\partial {k_y} } \vct{u}(\vct{k}) - \I \frac{\partial}{\partial {k_y} } \vct{u}^*(\vct{k})\frac{\partial}{\partial {k_x} }\vct{u}(\vct{k})\,, \quad \vct{k}\in D^*\,,
	\label{eq:bcvcom}
\end{equation}
where the derivatives are given by 
\begin{equation}
	\frac{\partial}{\partial {k_x} }\vct{u} = -(H-E)^{\dagger} \frac{\partial}{\partial {k_x} }H \vct{u}\,,\quad 	\frac{\partial}{\partial {k_y} } \vct{u} = -(H-E)^{\dagger} \frac{\partial}{\partial {k_y} } H \vct{u}\,.
\end{equation}

\subsubsection{Stage 2: computing the divergence-free component }
Next, based on the decomposition in (\ref{eq:adecom}), we compute the gauge-invariant vector fields. More explicitly, this involves finding an analytic and $\Lambda^*$-periodic $F$ that generates the divergence-free component and the harmonic components $(h_x,h_y)$. By Theorem\,\ref{thm:opt}, they uniquely (up to a lattice constant) determine the assignment of eigenvectors.

According to (\ref{eq:divless3}), we solve the following Poisson's equation
\begin{equation}
	\frac{\partial^2 F}{\partial {k^2_x} } + \frac{\partial^2 F}{\partial {k^2_y} }   = -g\,, \quad \mbox{$F$ is $\Lambda^*$-periodic on $D^*$}\,,
		\label{eq:divless3}
\end{equation}
where $g$ is given by the Berry curvature given by (\ref{eq:bcvcom}):
\begin{equation}
	g = \Omega_{xy}\,.
\end{equation}
However, by Theorem\,\ref{thm:poiss}, the equation is solvable if and only if the integral of $\Omega_{xy}$ over $D^*$ is zero. By the definition of the first Chern number $C_1$ in (\ref{eq:chern}), this corresponds to the condition 
\begin{equation}
	\int_{D^*} \D \vct{k}\, \Omega_{xy}(\vct{k}) = 2\pi C_1 = 0\,.
\end{equation}
Hence, we conclude that (\ref{eq:divless3}) is solvable if and only if $C_1=0$. In other words, when $C_1\ne 0$, we encounter the topological obstruction in Theorem\,\ref{thm:chern} and no such $F$ can be found. If we assume that $C_1 = 0$, we observe that (\ref{eq:divless3}) is easily solvable for the following reason. Since the projector $P$ is analytic and periodic on $D^*$,
so is $\Omega_{xy}$ by its definition in (\ref{eq:bcurv}). Hence, the Fourier coefficients of $\Omega_{xy}$ decay exponentially. By (\ref{eq:poisol2}), the solution $F$ is given by its Fourier series as
\begin{equation}
		F (\vct{k}) = \sum_{\substack{\vct{R}\in\Lambda\\ \vct{R}\ne \vct{0}}} \frac{g_{\vct{R}}}{\norm{\vct{R}}^2}\cdot e^{\I \vct{R}\cdot\vct{k}}\,, \quad \vct{k}\in D^*\,,
\end{equation}
where $g_{\vct{R}}$ is the Fourier coefficient of $g$ at $\vct{R}\in\Lambda$:
\begin{equation}
	g_{\vct{R}} = \frac{V_{\rm puc}}{(2\pi)^2} \int_{D^*} \D \vct{k} \, e^{-\I\vct{R}\cdot\vct{k}} g(\vct{k}) = \frac{V_{\rm puc}}{(2\pi)^2} \int_{D^*} \D \vct{k} \, e^{-\I\vct{R}\cdot\vct{k}} \Omega_{xy}(\vct{k})\,.
\end{equation}
After obtaining $F$, we compute the divergence-free Berry connection by the formula
\begin{align}
	A_x(\vct{k}) =&\frac{\partial }{\partial {k_y} }F(\vct{k}) = \sum_{\substack{\vct{R}\in\Lambda\\ \vct{R}\ne \vct{0}}} \I R_y \frac{g_{\vct{R}}}{\norm{\vct{R}}^2}\cdot e^{\I \vct{R}\cdot\vct{k}}\,,\\
	 A_y(\vct{k}) =& -\frac{\partial }{\partial {k_x} }F(\vct{k}) = -\sum_{\substack{\vct{R}\in\Lambda\\ \vct{R}\ne \vct{0}}} \I R_x \frac{g_{\vct{R}}}{\norm{\vct{R}}^2}\cdot e^{\I \vct{R}\cdot\vct{k}}\,,
	 \label{eq:bcint}
\end{align}
where $\vct{R}=(R_x,R_y)$\,.

We have obtained the Berry connection $(A_x,A_y)$ in (\ref{eq:bcint}). Due to Theorem\,\ref{thm:intecond}, the following Pfaffian system defined on $D^*$ is integrable:
\begin{equation}
	\frac{\partial }{\partial {k_x} } \vct{\tilde{u}} =\frac{\partial P}{\partial {k_x} } \vct{\tilde{u}} -\I A_x \vct{\tilde{u}}\,, \quad\frac{\partial }{\partial {k_y} } \vct{\tilde{u}} =\frac{\partial P}{\partial {k_y} } \vct{\tilde{u}} -\I A_y \vct{\tilde{u}}\,,
	\label{eq:pfafstage2}
\end{equation}
subject to the initial condition
\begin{equation}
	\vct{\tilde{u}}(\vct{k}^0) = \vct{u}(\vct{k}^0)\,,
\end{equation}
where $\vct{k}^0$ can be any point in $D^*$ and $\vct{u}(\vct{k}^0)$ is computed in (\ref{eq:eigstage1}).
However, we observe that the solution $\vct{\tilde{u}}$ to (\ref{eq:pfafstage2}) only defines an analytic assignment on $D^*$ viewed as a subset of $\R^2$ since $D^*$ is not simply connected when viewed as a torus. Hence, the solution $\vct{\tilde{u}}$ is not necessarily $\Lambda^*$-periodic on $D^*$.

\subsubsection{Stage 3: compute the harmonic component}

Next, we compute the harmonic components $(h_x,h_y)$ in (\ref{eq:adecom}) in order to turn $\vct{\tilde{u}}$ into a $\Lambda^*$-periodic assignment.  To do so, we parameterize $D^*$ by $(\kappa_1,\kappa_2)\in T$ so that we compute $(h_x,h_y)$ by its components $h_1$ and $h_2$ in the $\vct{b}_1$ and $\vct{b}_2$ direction via the formula
\begin{equation}
	(h_x,h_y) = h_1 \vct{b}_1 + h_2 \vct{b}_2\,.
\end{equation} 
By (\ref{eq:diffck}) and the definition of the Berry connection, the Berry connection $A_1,A_2$ in the $\vct{b}_1$ and $\vct{b}_2$ directions are given by the formulas
\begin{equation}
	A_1 = \vct{b}_1\cdot \vct{e}_x A_x + \vct{b}_1\cdot \vct{e}_y A_y\,,\quad 	A_2 = \vct{b}_2\cdot \vct{e}_x A_x + \vct{b}_2\cdot \vct{e}_y A_y\,,
	\label{eq:bckf}
\end{equation}
where $A_x,A_y$ are given in (\ref{eq:bcint}). Then the system in (\ref{eq:pfafstage2}) in the $\vct{b}_1$ and $\vct{b}_2$ directions are given by
\begin{equation}
\frac{\partial }{\partial {\kappa_1} } \vct{\tilde{u}}(\vct{k}(\kappa_1,\kappa_2)) = \frac{\partial P(\vct{k}(\kappa_1,\kappa_2)) }{\partial {\kappa_1} } \vct{\tilde{u}}(\vct{k}(\kappa_1,\kappa_2)) -\I A_1(\vct{k}(\kappa_1,\kappa_2)) \vct{\tilde{u}}(\vct{k}(\kappa_1,\kappa_2))\,,
		\label{eq:pfafstage31}
\end{equation}
\begin{equation}
\frac{\partial }{\partial {\kappa_2} } \vct{\tilde{u}}(\vct{k}(\kappa_1,\kappa_2)) =\frac{\partial P(\vct{k}(\kappa_1,\kappa_2)) }{\partial {\kappa_2} }  \vct{\tilde{u}}(\vct{k}(\kappa_1,\kappa_2)) -\I A_2(\vct{k}(\kappa_1,\kappa_2)) \vct{\tilde{u}}(\vct{k}(\kappa_1,\kappa_2))\,.
		\label{eq:pfafstage32}
\end{equation}
In order to compute $h_1$, we solve  (\ref{eq:pfafstage31}) along the line $\gamma_1$ across $T$ in Figure\,\ref{fig:path}, where $\kappa_2$ can be any value in $\sb{-\half,\half}$ and $\kappa_1$ runs from $-\half$ to $\half$. 
We specify the initial condition as 
\begin{equation}
	\vct{\tilde{u}}(\vct{k}(-1/2,\kappa_2)) =  \vct{u}(\vct{k}(-1/2,\kappa_2))\,,
			\label{eq:initstage3}
\end{equation}
where $\vct{u}(\vct{k}(-1/2,\kappa_2))$ is computed in (\ref{eq:eigstage1}). Since $\vct{\tilde{u}}$ is not guaranteed to be $\Lambda^*$-periodic on $D^*$, the solution $\vct{\tilde{u}}(\vct{k}(\kappa_1,\kappa_2))$ to (\ref{eq:pfafstage31}) may not be the same at $\kappa_1=-\half$ and $\kappa_1=\half$\,. However, $\vct{\tilde{u}}(\vct{k}(-1/2,\kappa_2))$ and $\vct{\tilde{u}}(\vct{k}(1/2,\kappa_2))$ are eigenvectors to the same eigenvalue due to periodicity, so they can only differ by a phase factor $e^{\I h_1 }$:
\begin{equation}
	\vct{\tilde{u}}(\vct{k}(-1/2,\kappa_2)) = e^{\I h_1 }\vct{\tilde{u}}(\vct{k}(1/2,\kappa_2))\,,
	\label{eq:h1diff}
\end{equation}
from which we compute $h_1$ by the formula
\begin{equation}
	h_1 = -\I\log(\vct{\tilde{u}}^*(\vct{k}(-1/2,\kappa_2))\vct{\tilde{u}}(\vct{k}(1/2,\kappa_2)))\,.
	\label{eq:h1stage3}
\end{equation}
Before we show that adding $h_1$ to $A_1$ will turn $\tilde{\vct{u}}(\vct{k}(\kappa_1,\kappa_2))$ into a periodic function in $\kappa_1\in\sb{-\half,\half}$ for any $\kappa_2$, we prove that $h_1$ in (\ref{eq:h1stage3}) is independent of $\kappa_1$. Consider solving the system defined by (\ref{eq:pfafstage31}) and  (\ref{eq:pfafstage32}), subject to the initial condition in (\ref{eq:initstage3}), along the rectangle in Figure\,\ref{fig:path}, where $\gamma_1$ is the path from which we obtain (\ref{eq:h1stage3}) and $\gamma_3$ is parallel to $\gamma_1$ at $\kappa'_1$. Since the system is integrable by construction, the solution after going around the square is the same as the initial condition in (\ref{eq:h1stage3}). Furthermore, (\ref{eq:pfafstage32}) on $\gamma_2$ and $\gamma_4$ are identical due to the periodicity of $\frac{\partial P}{\partial {\kappa_2}} $ and $A_2$ in (\ref{eq:pfafstage32}). Since  $\gamma_2$ and $\gamma_4$ are in the opposite direction, their contribution cancels. Consequently, the contribution from solving (\ref{eq:pfafstage31}) on $\gamma_1$ and $\gamma_3$ also cancels. Thus, we conclude that $h_1$ in (\ref{eq:h1stage3}) is independent $\kappa_1$; the same $h_1$ would be obtained by solving (\ref{eq:pfafstage31}) on any line parallel to $\gamma_1$.

\begin{figure}[h]
	\centering	
		\includegraphics[scale=0.50]{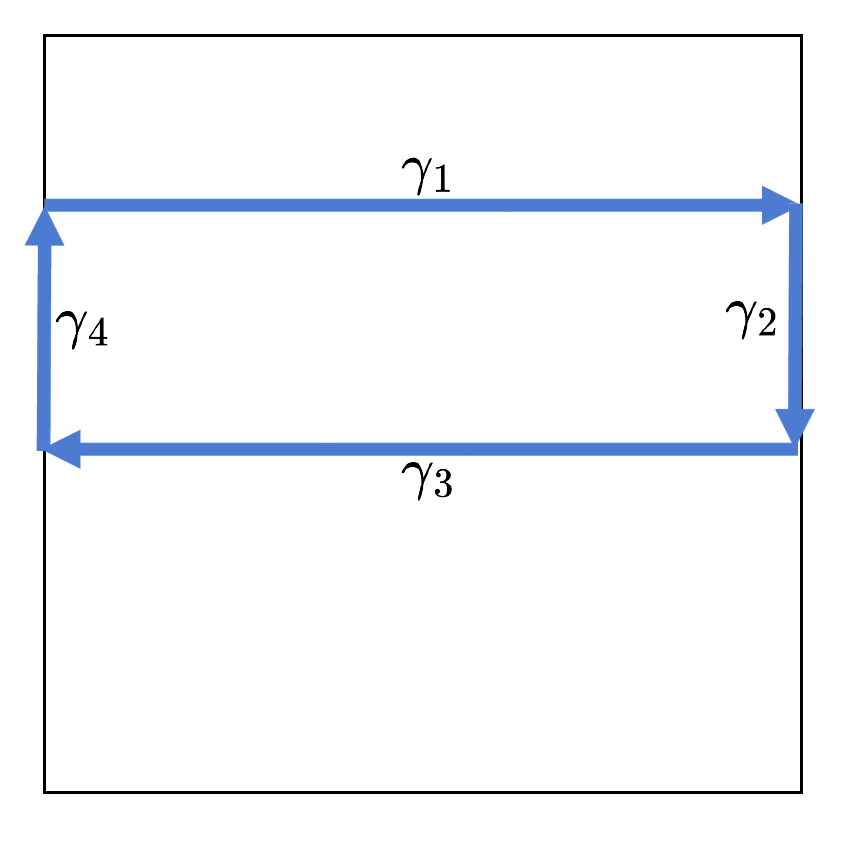}	
				\includegraphics[scale=0.50]{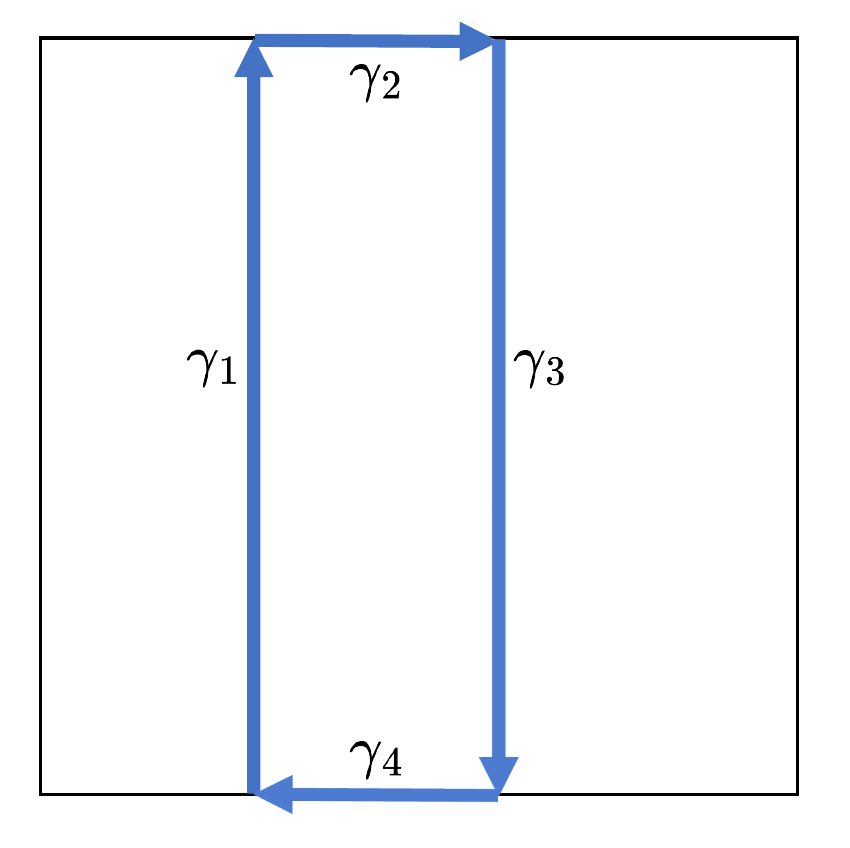}	
		\vspace{-1em}
	\caption{ The path for computing the harmonic component $h_1$ (left) and $h_2$ (right). \label{fig:path}}	
\end{figure}

The procedures for computing $h_2$ are similar to those above for $h_1$ by switching the role of $\kappa_1$ and $\kappa_2$. 
We solve (\ref{eq:pfafstage32})  along the line $\gamma_1$, where $\kappa_2$ runs from $-\half$ to $\half$ for any $\kappa_1\in \sb{-\half,\half}$ with the initial condition 
\begin{equation}
	\vct{\tilde{u}}(\vct{k}(\kappa_1,-1/2)) =  \vct{u}(\vct{k}(\kappa_1,-1/2))\,,
			\label{eq:initstage3}
\end{equation} 
where $\vct{u}(\vct{k}(\kappa_1,-1/2))$ is computed in (\ref{eq:eigstage1}).
By periodicity, $\vct{\tilde{u}}(\vct{k}(\kappa_1,-1/2))$ and $\vct{\tilde{u}}(\vct{k}(\kappa_1,1/2)) $ can only differ by a phase factor $e^{\I h_2}$:
\begin{equation}
	\vct{\tilde{u}}(\vct{k}(\kappa_1,-1/2)) = e^{\I h_2 }\vct{\tilde{u}}(\vct{k}(\kappa_1,1/2))\,,
	\label{eq:h2diff}
\end{equation}
from which we compute $h_2$ by the formula
\begin{equation}
	h_2 = -\I\log(\vct{\tilde{u}}^*(\vct{k}(\kappa_1,-1/2))\vct{\tilde{u}}(\vct{k}(\kappa_1,1/2)) )\,.
	\label{eq:h2stage3}
\end{equation}
The argument for showing that $h_1$ is independent of $\kappa_2$ can be applied to show that $h_2$ in (\ref{eq:h2stage3}) is independent $\kappa_1$.

Next, we apply the following gauge transformation to the solution $\tilde{\vct{u}}(\vct{k}(\kappa_1,\kappa_2))$ to (\ref{eq:pfafstage31}) subject to the initial condition (\ref{eq:initstage3}):
\begin{equation}
	\dbtilde{\vct{u}}(\vct{k}(\kappa_1,\kappa_2)) = e^{-\I h_1 \kappa_1}e^{-\I h_2 \kappa_2}\tilde{\vct{u}}(\vct{k}(\kappa_1,\kappa_2))\,.
	\label{eq:gth1}
\end{equation}
Obviously, the Berry connection $A_1$ and $A_2$ in (\ref{eq:pfafstage31}) and (\ref{eq:pfafstage32})  are modified accordingly as
\begin{equation}
	\tilde{A}_1 = A_1 + h_1\,,\quad 	\tilde{A}_2 = A_2 + h_2\,,
	\label{eq:tbck}
\end{equation}
and the new $\dbtilde{\vct{u}}$ satisfies 
\begin{equation}
	\frac{\partial }{\partial {\kappa_1}} \vct{\dbtilde{u}}(\vct{k}(\kappa_1,\kappa_2)) =\frac{\partial P(\vct{k}(\kappa_1,\kappa_2)) }{\partial {\kappa_1} }  \vct{\dbtilde{u}}(\vct{k}(\kappa_1,\kappa_2)) -\I \tilde{A}_1(\vct{k}(\kappa_1,\kappa_2)) \vct{\dbtilde{u}}(\vct{k}(\kappa_1,\kappa_2))\,,
		\label{eq:pfafstage311}
\end{equation}
\begin{equation}
	\frac{\partial }{\partial {\kappa_2}} \vct{\dbtilde{u}}(\vct{k}(\kappa_1,\kappa_2)) =\frac{\partial P(\vct{k}(\kappa_1,\kappa_2)) }{\partial {\kappa_2} }   \vct{\dbtilde{u}}(\vct{k}(\kappa_1,\kappa_2)) -\I \tilde{A}_2(\vct{k}(\kappa_1,\kappa_2)) \vct{\dbtilde{u}}(\vct{k}(\kappa_1,\kappa_2))\,.
		\label{eq:pfafstage322}
\end{equation}
By (\ref{eq:gth1}) and (\ref{eq:h1diff}), it is easy to see that $\dbtilde{\vct{u}}$ satisfies
\begin{equation}
	\dbtilde{\vct{u}}(\vct{k}(-1/2,\kappa_2)) = \dbtilde{\vct{u}}(\vct{k}(1/2,\kappa_2))\,,\quad \kappa_2\in\sb{-\half,\half}.
	\label{eq:coni}
\end{equation}
The periodicity of $\partial_{\kappa_1}P$ and $\tilde{A}_1$ in (\ref{eq:pfafstage311}) show that 
\begin{equation}
	\frac{\partial}{\partial {\kappa_1}}\dbtilde{\vct{u}}(\vct{k}(-1/2,\kappa_2)) = \frac{\partial}{\partial {\kappa_1}}\dbtilde{\vct{u}}(\vct{k}(1/2,\kappa_2))\,,\quad \kappa_2\in\sb{-\half,\half}.
\end{equation}
By repetitively differentiating (\ref{eq:pfafstage311}), we conclude that the derivative of order $n=2,3,\ldots$ satisfies
\begin{equation}
	\frac{\partial^n}{\partial {\kappa^n_1}}\dbtilde{\vct{u}}(\vct{k}(-1/2,\kappa_2)) = \frac{\partial^n }{\partial {\kappa^n_1}}\dbtilde{\vct{u}}(\vct{k}(1/2,\kappa_2))\,,\quad \kappa_2\in\sb{-\half,\half}.
\end{equation}
By (\ref{eq:gth1}), (\ref{eq:h2diff}), (\ref{eq:pfafstage322}) and the periodicity of $\partial_{\kappa_2}P$ and $\tilde{A}_2$, we conclude similarly that 
\begin{equation}
	\dbtilde{\vct{u}}(\vct{k}(\kappa_1,-1/2)) = \dbtilde{\vct{u}}(\vct{k}(\kappa_1,1/2))\,,\quad \kappa_1\in\sb{-\half,\half}\,,
\end{equation}
and
\begin{equation}
	\frac{\partial^n}{\partial {\kappa^n_2}}\dbtilde{\vct{u}}(\vct{k}(\kappa_1,-1/2)) = \frac{\partial^n }{\partial {\kappa^n_2}}\dbtilde{\vct{u}}(\vct{k}(\kappa_1,1/2))\,,\quad \kappa_1\in\sb{-\half,\half}\,,
		\label{eq:conf}
\end{equation}
for $n=1,2,\ldots$.
For simplicity, suppose that the initial condition for the system defined by (\ref{eq:pfafstage311}) and (\ref{eq:pfafstage322}) is specified  at $\kappa_1=\kappa_2=-\frac{1}{2}$ by the formula
\begin{equation}
	\dbtilde{\vct{u}} (\vct{k}(-1/2,-1/2))= \vct{u}(\vct{k}(-1/2,-1/2))\,,
	\label{eq:initf}
\end{equation}
where $ \vct{u}(\vct{k}(-1/2,-1/2))$ is computed in (\ref{eq:eigstage1}).
Then, by (\ref{eq:coni}-\ref{eq:conf}), together with the integrability of the system defined by (\ref{eq:pfafstage311}) and (\ref{eq:pfafstage322}), we conclude that the solution $\dbtilde{\vct{u}}$ defines an analytic and $\Lambda^*$-periodic assignment on $D^*$. 
\begin{remark}
	Since the system defined by (\ref{eq:pfafstage311}) and (\ref{eq:pfafstage322}) is integrable, the paths for which they are solved for finding $\dbtilde{\vct{u}}$ over $D^*$ is irrelevant. For example, the paths can be taken to be the same as those in Section\,\ref{sec:method1}.
\end{remark}

\subsubsection{Realty of Wannier functions}
In the previous section, we have constructed the optimal assignment $\dbtilde{\vct{u}}$ on $D^*$. There is one remaining degree of freedom -- the phase choice for the initial condition in (\ref{eq:initf}).
In this section, we show that, when $H$ has time-reversal symmetry (see (\ref{eq:trs})), $\dbtilde{\vct{u}}$  will have the symmetry
\begin{equation}
	\bar{\dbtilde{\vct{u}}}(\vct{k}) = \dbtilde{\vct{u}}(-\vct{k})\,,\quad \vct{k}\in D^*
\end{equation}
provided the initial condition in (\ref{eq:initf}) is chosen to be real: 
\begin{equation}
	\bar{{\vct{u}} }(\vct{k}(-1/2,-1/2)) = {\vct{u}} (\vct{k}(-1/2,-1/2))\,.
		\label{eq:initfreal}
\end{equation}

First, we observe that (\ref{eq:initfreal}) is always possible by considering the complex conjugate of (\ref{eq:eigstage1}) at $\kappa_1=\kappa_2=-\half$:
\begin{align}
\begin{split}
	\bar{H}(\vct{k}(-1/2,-1/2))\bar{{\vct{u}} }(\vct{k}(-1/2,-1/2)) =& E(\vct{k}(-1/2,-1/2)) \bar{{\vct{u}} }(\vct{k}(-1/2,-1/2))\\
	 =& H(\vct{k}(-1/2,-1/2))\bar{{\vct{u}} }(\vct{k}(-1/2,-1/2))\,,\label{eq:trseig}
\end{split}
\end{align}
where we used (\ref{eq:trs}) and the periodicity of $H$ to obtain
\begin{equation}
\bar{H}(\vct{k}(-1/2,-1/2)) = H(\vct{k}(1/2,1/2)) = H(\vct{k}(-1/2,-1/2))\,.	
\end{equation}
The relation in (\ref{eq:trseig}) shows that both ${\vct{u}}(\vct{k}(-1/2,-1/2))$ and $\bar{{\vct{u}} }(\vct{k}(-1/2,-1/2))$ are eigenvectors of the same (non-degenerate) eigenvalue, so they can only differ by a phase factor of the form ${\vct{u}}^*(\vct{k}(-1/2,-1/2)) = {\vct{u}}(\vct{k}(-1/2,-1/2))e^{2\I\varphi_0}$
 for some real $\varphi_0$\,. If (\ref{eq:initfreal}) does not hold, we can pick ${\vct{u}}(\vct{k}(-1/2,-1/2))e^{\I\varphi_0}$ in stead of ${\vct{u}}(\vct{k}(-1/2,-1/2))$ as the vector in (\ref{eq:initf}) so that (\ref{eq:initfreal}) holds.

We also observe that, since the Berry curvature satisfies $\Omega_{xy}(\vct{k}) = -\Omega_{xy}(-\vct{k}) $ if $H$ has time-reversal symmetry (see Remark\,\ref{rmk:trsp}),  formulas in (\ref{eq:bcint}) show that $A_x$ and $A_y$  stratify 
\begin{equation}
	A_x(-\vct{k}) = A_x(\vct{k})\,,\quad A_y(-\vct{k}) = A_y(\vct{k})\,.
	\label{eq:bcff}
\end{equation}
Combining (\ref{eq:bcff}) with (\ref{eq:bckf}) and (\ref{eq:tbck}) shows that
\begin{equation}
	\tilde{A}_1(-\vct{k}) = 	\tilde{A}_1(\vct{k}) \,,\quad 	\tilde{A}_2(-\vct{k}) = 	\tilde{A}_2(\vct{k})\,.
		\label{eq:bcfff}
\end{equation}
Furthermore, since $\bar{P}(\vct{k})=P(-\vct{k})$, we have 
\begin{equation}
	\frac{\partial}{\partial \kappa_1} \bar{P}(\vct{k}) = -\frac{\partial}{\partial \kappa_1} P(-\vct{k})\,,\quad \frac{\partial}{\partial \kappa_2} \bar{P}(\vct{k}) = -\frac{\partial}{\partial \kappa_2} P(-\vct{k})\,.
		\label{eq:ptrsf}
\end{equation}
By (\ref{eq:ptrsf}) and (\ref{eq:bcfff}), we observe that (\ref{eq:pfafstage311}) and (\ref{eq:pfafstage322}) are unchanged under complex conjugation followed by replacing $\vct{k}$ with $-\vct{k}$. This implies that, if we transform the solution $\dbtilde{\vct{u}}$  by the same operations in the following 
\begin{equation}
	\dbtilde{\vct{u}}(\vct{k}) \rightarrow \bar{\dbtilde{\vct{u}}}(-\vct{k})\,,
\end{equation}
the transformed solution will still satisfy  (\ref{eq:pfafstage311}) and (\ref{eq:pfafstage322}).
Furthermore, by periodicity, the vector on the right of (\ref{eq:initf}) remains an eigenvector at $\kappa_1=\kappa_2=\half$.  By the realty of the initial condition in (\ref{eq:initfreal}), the transformed solutions have the initial condition as the original one. Namely, we have
\begin{align}
	&\dbtilde{\vct{u}}(\vct{k}(-1/2,-1/2)) = \vct{u}(\vct{k}(-1/2,-1/2)) = \bar{\dbtilde{\vct{u}}}(\vct{k}(1/2,1/2))\,.
\end{align}
As a result, the transformed solution satisfies the same equations as $\dbtilde{\vct{u}}$ with the same initial condition. By the uniqueness theorem of initial value problems, we conclude that $\dbtilde{\vct{u}}(\vct{k}) = \bar{\dbtilde{\vct{u}}}(-\vct{k})$. By complex conjugation, we conclude that
\begin{equation}
	\bar{\dbtilde{\vct{u}}}(\vct{k}) = 	\dbtilde{\vct{u}}(-\vct{k})\,,\quad \vct{k}\in D^*\,.
\end{equation}
Hence the Fourier coefficients of $\dbtilde{\vct{u}}$ are real, so is its corresponding Wannier function defined by (\ref{eq:wann}).


\subsection{Derivation of formulas in Lemma\,\ref{lem:var} \label{sec:appvar}}
The following contains the derivation for the moment functions in Lemma\,\ref{lem:var}. 

For the Wannier center, by (\ref{eq:odekn1}) and the decomposition (\ref{eq:adecom}) , we have
\begin{align}
	\begin{split}
	\langle \vct{R} \rangle &= \I \frac{V_{\rm puc}}{(2\pi)^2}\int_{D^*} \D\vct{k}\, \tilde{\vct{u}}^*(\vct{k})\grad_{\vct{k}} \tilde{\vct{u}}(\vct{k})\\
	&= \frac{V_{\rm puc}}{(2\pi)^2} \int_{D^*} \D\vct{k}\, \vct{A}(\vct{k})	\\
	&= (h_x,h_y)\,,
		\end{split}
\end{align}
where we have used (\ref{eq:kxy0}) to obtain (\ref{eq:ap1}) so that integrals of the derivatives of $\psi, F$ vanishes. 

For the second moment $	\langle \norm{\vct{R}}^2 \rangle$, we have
\begin{align}
		\begin{split}
	\langle \norm{\vct{R}}^2 \rangle &= \frac{V_{\rm puc}}{(2\pi)^2}\int_{D^*} \D\vct{k}\, \norm{\grad_{\vct{k}} \tilde{\vct{u}}(\vct{k})}^2 \\
	&= \frac{V_{\rm puc}}{(2\pi)^2}\int_{D^*} \D\vct{k}\,\bb{ \norm{\frac{\partial P(\vct{k})}{\partial k_x}\tilde{\vct{u}}(\vct{k})}^2 + \norm{\frac{\partial P(\vct{k})}{\partial k_y}\tilde{\vct{u}}(\vct{k})}^2  +  \norm{\vct{A}(\vct{k})}^2}\,. \label{eq:ap1}
	\end{split}
\end{align}
We observe that we have used (\ref{eq:para2}) so that there is no cross terms that contain both the derivatives of the projector and $\vct{A}$. 
The last term in (\ref{eq:ap1}) is given by 
\begin{align}
&	\frac{V_{\rm puc}}{(2\pi)^2}\int_{D^*} \D\vct{k}\, \norm{\vct{A}(\vct{k})}^2 \nonumber \\
	&= \frac{V_{\rm puc}}{(2\pi)^2}\int_{D^*} \D\vct{k} \bigg[ \bb{\frac{\partial \psi}{\partial k_x}}^2 + \bb{\frac{\partial \psi}{\partial k_y}}^2 + \bb{\frac{\partial F}{\partial k_x}}^2 + \bb{\frac{\partial F}{\partial k_y}}^2 \nonumber\\
	& \hspace{2cm} + 2\bb{\frac{\partial \psi}{\partial k_x}\frac{\partial F}{\partial k_y} - \frac{\partial \psi}{\partial k_y}\frac{\partial F}{\partial k_x}}\bigg] + h_x^2 + h_y^2\,\label{eq:ap3}\\
	&=\frac{V_{\rm puc}}{(2\pi)^2}\int_{D^*} \D\vct{k} \bigg[ \bb{\frac{\partial \psi}{\partial k_x}}^2 + \bb{\frac{\partial \psi}{\partial k_y}}^2 + \bb{\frac{\partial F}{\partial k_x}}^2 + \bb{\frac{\partial F}{\partial k_y}}^2\bigg] + h_x^2 + h_y^2\,,
\end{align}
where we applied (\ref{eq:kxy0}) to make terms linear in the derivatives of $\psi$ and $F$ vanish in the first equality, and
we applied integration by parts to $\frac{\partial \psi}{\partial k_x}$ and $\frac{\partial \psi}{\partial k_y}$ in the second equality and the fact that both $\psi$ and $F$ are $\Lambda^*$-periodic to make the cross terms between derivatives of $\psi$ and $F$ vanish.

\end{document}